\documentclass[a4paper]{article} 

\usepackage{amsmath,amssymb}
\usepackage{amsthm}
\usepackage[usenames,dvipsnames,svgnames,table]{xcolor}
\usepackage{stmaryrd}
\usepackage{fancyvrb}
\usepackage{enumitem}
\setlist{itemsep=5pt}

\usepackage[disable]{todonotes}
\usepackage{subfigure}

\usepackage{xifthen}
\usepackage{xargs}
\makeatletter
\newcommand{\newreptheorem}[2]{
	\newenvironment{rep#1}[2][]{\@begintheorem{Restatement of #2 \ref{##2}\ifthenelse{\isempty{##1}}{}{ (##1)}}{\unskip}}{\@endtheorem}
}
\let\newoldtheorem\newtheorem
\renewcommandx{\newtheorem}[3][2=NoValueAtAll]{
	\ifthenelse{\equal{#2}{NoValueAtAll}}{\newoldtheorem{#1}{#3}}{\newoldtheorem{#1}[#2]{#3}}
	\newreptheorem{#1}{#3}
}
\makeatother

\newtheorem{defn}{Definition}
\newtheorem{thm}{Theorem}
\newtheorem{lem}[thm]{Lemma}
\newtheorem{prop}{Proposition}

\newtheorem{remark}{Remark}
\newtheorem{example}{Example}

\newcommand{\citeN}[1]{\cite{#1}}

\newcommand{\FDtask}[1]{{{#1}}}
\newcommand{\JBtask}[1]{{{#1}}} 
\newcommand{\MVtask}[1]{{{#1}}} 
\newcommand{\TYtask}[1]{{{#1}}} 

\newcommand{\BNFcce}{{\bf ::=}}
\newcommand{\BNFmid}{\;\bigr\rvert\;}

\newcommand{\PROGRAM}{\mathtt{P}}
\newcommand{\FUNCTION}{\mathtt{F}}

\newcommand{\main}{\mathtt{main}}
\newcommand{\s}{\mathtt{s}}
\newcommand{\e}{\mathtt{e}}
\newcommand{\emain}{\e_{\main}}
\newcommand{\fname}{\mathtt{d}}
\newcommand{\xname}{\mathtt{x}}
\newcommand{\yname}{\mathtt{y}}
\newcommand{\zname}{\mathtt{z}}
\newcommand{\oname}{\mathtt{b}}
\newcommand{\ocname}{\mathtt{bc}}

\newcommand{\lvalueSet}{\mathcal{L}}

\newcommand{\anyvalue}{\mathtt{v}}
\newcommand{\lvalue}{\ell}
\newcommand{\fvalue}{\phi}
\newcommand{\fvaluealt}{\psi}
\newcommand{\funvalue}{\mathtt{f}}

\newcommand{\truevalue}{\mathtt{True}}
\newcommand{\falsevalue}{\mathtt{False}}

\newcommand{\dc}{\mathtt{c}}
\newcommand{\dcOf}[2]{#1(#2)}

\newcommand{\auxNAME}{\textit{aux}}
\newcommand{\aux}[1]{\auxNAME(#1)}
\newcommand{\bodyNAME}{\textit{body}}
\newcommand{\body}[1]{\bodyNAME(#1)}
\newcommand{\argsNAME}{\textit{args}}
\newcommand{\args}[1]{\argsNAME(#1)}

\newcommand{\FVname}{\textbf{FV}}
\newcommand{\FV}[1]{\FVname(#1)}

\newcommand{\defK}{\mathtt{def}}
\newcommand{\nbrK}{\mathtt{nbr}}
\newcommand{\repK}{\mathtt{rep}}
\newcommand{\ifK}{\mathtt{if}}

\newcommand{\foldK}{\mathtt{foldhood}}

\newcommand{\letK}{\mathtt{let}}
\newcommand{\inK}{\mathtt{in}}

\newcommand{\eqSymK}[1]{\mathrm{ \{ #1 \} }}
\newcommand{\toSymK}[1]{\mathrm{ \texttt{=>} #1 }}

\newcommand{\selfK}{\mathtt{uid}}
\newcommand{\muxK}{\mathtt{mux}}
\newcommand{\minHoodK}{\mathtt{minHood}}

\newcommand{\snsNumK}{\mathtt{snsNum}}

\newcommand{\minHoodLoc}{\mathtt{minHoodLoc}}
\newcommand{\nbrlt}{\mathtt{nbrlt}}

\newcommand{\type}{\textit{T}}
\newcommand{\ltype}{\textit{L}}

\newcommand{\ftypeOf}[1]{\mathtt{field}(#1)}

\newcommand{\btype}{\mathtt{bool}}
\newcommand{\ntype}{\mathtt{num}}

\newcommand{\tvar}{t}

\newcommand{\surfaceTyping}[3]{
  \begin{array}{l@{\;}c}
    \stackrel{~}{{\tiny \textrm{[#1]}}} & #2 \\ \hline 
    \multicolumn{2}{c}{\rule{0pt}{0.9\normalbaselineskip} #3}
  \end{array}
}
\newcommand{\nullsurfaceTyping}[2]{
  \surfaceTyping{#1}{}{#2}
}

\newcommand{\deviceIdSet}{\textbf{D}}
\newcommand{\builtinop}[3]{\llparenthesis #1 \rrparenthesis_{#2}^{#3}}
\newcommand{\filter}{F}

\newcommand{\Trees}{\Theta}
\newcommand{\emptyseq}{\bullet}

\newcommand{\devset}{I}
\newcommand{\Topo}{\tau}
\newcommand{\Sens}{\Sigma}
\newcommand{\Envi}{\textit{Env}}
\newcommand{\EnviS}[2]{#1,#2}
\newcommand{\SystS}[2]{\langle #1;#2\rangle}
\newcommand{\Field}{\Psi}
\newcommand{\Cfg}{N}
\newcommand{\System}{\mathcal{S}}
\newcommand{\wfn}[1]{\textit{WFE}(#1)}
\newcommand{\senstate}{\sigma}

\newcommand{\nettran}[3]{#1\xrightarrow{#2} #3}
\newcommand{\act}{\textit{act}}
\newcommand{\envact}{\textit{env}}

\newcommand{\envmap}[2]{#1\mapsto #2}
\newcommand{\mapupdate}[2]{#1[#2]}
\newcommand{\proj}[2]{{#1}|_{#2}}

\newcommand{\ruleNameSize}[1]{{\scriptsize #1}}

\newcommand{\domofNAME}{\textbf{dom}}
\newcommand{\domof}[1]{\domofNAME(#1)}

\newcommand{\vtree}{\theta}
\newcommand{\mkvt}[2]{#1 \langle #2 \rangle}
\newcommand{\piB}[1]{\pi^{#1}}
\newcommand{\piBof}[2]{\piB{#1}(#2)}
\newcommand{\piI}[1]{\pi_{#1}}
\newcommand{\piIof}[2]{\piI{#1}(#2)}
\newcommand{\piIofOv}[1]{\overline{\pi}(#1)}

\newcommand{\bsopsem}[4]{#1;#2\vdash #3\Downarrow #4}
\newcommand{\deviceId}{\delta}
\newcommand{\vroot}{\mathbf{\rho}}
\newcommand{\vrootOf}[1]{\vroot(#1)}
\newcommand{\substitution}[2]{#1:=#2}
\newcommand{\applySubstitution}[2]{#1[#2]}

\newcommand{\skiptransition}{\\[10pt]}

\newcommand{\netopsemRule}[3]{\surfaceTyping{#1}{#2}{#3}}

\newcommand{\builtindenot}[2]{\mathcal{#1}\llbracket #2 \rrbracket}

\newcommand\pto{\mathrel{\ooalign{\hfil$\mapstochar$\hfil\cr$\to$\cr}}}

\newcommand{\denottype}[1]{\mathcal{T}\llbracket{#1}\rrbracket}
\newcommand{\denotval}[1]{\mathcal{V}\llbracket {#1} \rrbracket}

\newcommand{\dvalue}[0]{\Phi}

\definecolor{dark-gray}{gray}{0}
\newcommand{\km}[1]{{\bf \textcolor{red}{#1}}} 
\newcommand{\pr}[1]{\texttt{\textcolor{blue}{#1}}} 

\newcommand{\cp}[1]{\left( #1 \right)}
\newcommand{\qp}[1]{\left[ #1 \right]}
\newcommand{\ap}[1]{\langle #1 \rangle}
\newcommand{\bp}[1]{\left\lbrace #1 \right\rbrace}
\newcommand{\vp}[1]{\left\lvert #1 \right\rvert}

\DeclareMathOperator{\dist}{dist}

\usepackage[ruled]{algorithm2e}

\SetAlFnt{\small}
\SetAlCapFnt{\small}
\SetAlCapNameFnt{\small}
\SetAlCapHSkip{0pt}
\IncMargin{-\parindent}

\usepackage{hyperref}

\usepackage{cleveref}

\begin{document}

\markboth{M. Viroli et al.}{Engineering Resilient Collective Adaptive Systems by Self-Stabilisation}

\title{Engineering Resilient Collective Adaptive Systems by Self-Stabilisation}
\author{Mirko Viroli, Giorgio Audrito, Jacob Beal,\\Ferruccio Damiani, Danilo Pianini}
\date{}

\maketitle

\begin{abstract}
Collective adaptive systems are an emerging class of networked computational systems, particularly suited in application domains such as smart cities, complex sensor networks, and the Internet of Things.
These systems tend to feature large scale, heterogeneity of communication model (including opportunistic peer-to-peer wireless interaction), and require inherent self-adaptiveness properties to address unforeseen changes in operating conditions.
In this context, it is extremely difficult (if not seemingly intractable) to engineer reusable pieces of distributed behaviour so as to make them provably correct and smoothly composable.

Building on the field calculus, a computational model (and associated toolchain) capturing the notion of aggregate network-level computation, we address this problem with an engineering methodology coupling formal theory and computer simulation.
On the one hand, functional properties are addressed by identifying the largest-to-date field calculus fragment generating self-stabilising behaviour, guaranteed to eventually attain a correct and stable final state despite any transient perturbation in state or topology, and including highly reusable building blocks for information spreading, aggregation, and time evolution.
On the other hand, \TYtask{dynamical} properties are addressed by simulation, empirically evaluating the different performances that can be obtained by switching \TYtask{between} implementations of building blocks with provably equivalent functional properties.
Overall, our methodology sheds light on how to identify core building blocks of collective behaviour, and how to select implementations that improve system performance while leaving overall system function and resiliency properties unchanged.
\end{abstract}

%
%
 
%
%


\section{Introduction}
\label{sec:introduction}

Collective adaptive systems are an emerging class of networked computational systems situated in the real-world, finding extensive application in domains such as smart cities, complex sensor networks, and the Internet of Things.
The pervasive nature of these systems can potentially fulfill the vision of a fully integrated digital and physical world.
With collective adaptive systems, in the near future one may easily envision ``enhanced'' living and working environments, thanks to computing devices connected to every physical object that provide increasingly powerful capabilities of computing, storage of local data, communication with neighbours, physical sensing, and actuation.
\TYtask{Such environments} pave the way towards implementing any non-trivial pervasive computing service through the inherent distributed cooperation of a large set of devices, so as to address by self-adaptation the unforeseen changes in working conditions that necessarily happen---much in the same way adaptivity and resilience are addressed in complex natural systems at all levels, from molecules and cells to animals, species, and entire ecosystems \cite{ZV-JPCC2011}.

A long-standing aim in computer science has indeed been to find effective engineering methods for exploiting mechanisms for adaptation and resilience in complex, large-scale applications.
Practical adoption, however, poses serious challenges, since such mechanisms need to carefully trade efficiency for resilience, and are often difficult to predictably compose to meet more complex specifications.
Despite much prior work, e.g., in macroprogramming, spatial computing, pattern languages, etc. (as surveyed in~\cite{SpatialIGI2013}), to date no such approach has provided a comprehensive workflow for efficient engineering of complex
self-organising systems.

\MVtask{Recently, however, among the many related works (see Section 2), two key ingredients have been provided toward
such an engineering workflow.}
First, the \emph{computational field calculus}~\cite{VDB-FOCLASA-CIC2013,forte2015} provides \FDtask{a language} for specifying large-scale distributed computations and, critically, a functional programming model for their encapsulation and safely-scoped composition. \MVtask{This framework assumes that the system is composed of a discrete set of devices deployed in a space equipped with a notion of locality: each device works in asynchronous computational rounds producing a result data that is sent to local neighbours\footnote{Hence, we do not specifically deal with continuous functions and with virtual nodes that do not host computation---though they are mechanisms that might be mimicked: e.g., approximation of continuous functions can be developed along the lines of \cite{BVPD-SASO2016}.}.}
Second, \MVtask{a set of sufficient conditions for ``self-stabilisation'' have been identified~\cite{VD-COORD2014-LNCS2014,BV-FOCAS2014,LMCS:selfsabilisation}, guaranteeing that a large class of programs are all self-adaptive systems resilient to changes in their environment}\MVtask{---more precisely, after some period without changes in the computational environment, such a distributed computation reaches a stable state that only depends on inputs and network topology (i.e., the converged state is independent of computational history).}
\MVtask{As an example, such conditions reveal the non-resiliency of gossiping to find the minimum of a given value across a network: since each node continuously exposes that minimum of the values received from neighbours, the system can't recover from the temporary decrease of a value below the minimum \cite{LMCS:selfsabilisation}.}

This paper \TYtask{combines} these two advances \MVtask{with an approach to optimisation of self-organising systems via substitution of equivalent
coordination mechanisms}, guaranteed to result in the same functional behaviour though with different performance characteristics.
\TYtask{Together,} they combine into a workflow for efficient
engineering of complex self-organising systems
in which, once a distributed system is framed as a computation over fields, then:
\begin{enumerate}
 \item a minimal resilient implementation is created, by composition of building blocks extracted from a library of reusable self-stabilising components or designed ad-hoc;
 \item optimisation of performance is achieved by selective substitution of building block instances with alternate implementations, whose performance is checked by simulation.
\end{enumerate}

\TYtask{This workflow is backed} by pairing formal modelling and simulation of complex distributed systems.
On the one hand, functional properties are addressed by building systems on top of a formally proved language of self-stabilising specifications, which \TYtask{also establishes the functional equivalence} of certain building blocks.
On the other hand, dynamical properties are addressed by simulation, empirically evaluating the different performance of systems in which building blocks are selectively substituted by provably equivalent implementations.
\MVtask{In particular, the use of empirical analysis for the large-scale systems we consider (even though it may result in sub-optimality), is motivated by the fact that finding an optimal combination of alternative distributed systems implementations easily becomes a computationally hard problem \cite{DBLP:conf/emsoft/DarulovaKMS13}, that -- to the best of our knowledge -- has never been addressed in literature.}

The technical contributions of this paper with respect to previous work are the following:
\begin{itemize}
 \item we provide a simplified operational semantics of the first-order field calculus, reducing the formalisation in \cite{DVB-SCP2016} along the lines of the approach proposed in \cite{forte2015} for the higher-order version of the calculus;
 \item by developing on \cite{VBDP-SASO2015}, we provide the largest to date fragment of the calculus that is provably self-stabilising, using a proof methodology showing inevitable reachability of a unique stable state \cite{LMCS:selfsabilisation}---the calculus is shown to include \MVtask{the basic self-organisation building blocks defined in \cite{BV-FOCAS2014}};
 \item we provide alternative implementations of these building blocks (some new and some consolidating existing algorithms), still in the \TYtask{self-stabilising} fragment, and proved equivalent to the original \TYtask{versions};
 \item we empirically evaluate and compare the performance of the alternate versions of the building blocks, characterising the contexts in which a given implementation can be favoured.
\end{itemize}

The remainder of this paper is organised as follows:
Section~\ref{sec:background} reviews related work and discusses the background and motivation for this paper, presenting the methodological workflow in the context of the field calculus;
Section~\ref{sec:calculus} then formalises syntax, semantics and properties of the field calculus, providing building block examples showcasing its expressiveness;
Section~\ref{sec:stabilisation} presents our self-stabilisation framework, \MVtask{with formal definition and methodological implications};
Section~\ref{sec:fragment} provides the self-stabilising fragment, proof of self-stabilisation, proof of membership for the building blocks, \MVtask{and several motivating examples};
Section~\ref{sec:alternatives} defines alternative implementations of the building blocks, and empirically evaluates their performance;
Section~\ref{sec:casestudy} \TYtask{presents} two case studies illustrating the methodology;
%
%
and Section \ref{sec:conclusions} summarises and concludes.


\section{Related Work, Background and Motivation}
\label{sec:background}

The approach we propose in this paper falls under the umbrella of \emph{aggregate computing} \cite{BPV-COMPUTER2015}, a framework for designing resilient distributed systems based on the idea of abstracting away from \TYtask{the behaviour of individual devices}: system design focusses instead on the global, aggregate behaviour of the collection of all (or a subset of) the devices.
Put in other words, aggregate computing considers as ``abstract computing machine'' the whole set of devices seen as a single ``body.''
Coupled with a formal computational model, this approach has the goal of providing smooth composition of distributed behaviour, and trading off expressiveness with the ability to control the outcome of system design. This is done 
by relying on a formal model guaranteeing functional properties and relying on other means (such as simulation) for assessing \TYtask{dynamical} properties.

\subsection{Relationship to Prior Work}

The work presented in this paper builds on two well-developed areas of prior work: 
aggregate programming languages, which address the challenges of programming collectives of devices,
and self-stabilisation, which formalises a useful class of resilient system behaviours.

\subsubsection*{Aggregate programming}

Aggregate programming methods of many sorts have been developed across a wide variety of applications and fields.
A thorough review may be found in~\cite{SpatialIGI2013}, in which four main approaches to aggregate programming are identified.
First, many ``bottom-up'' methods aim to simplify aggregate programming by focusing on abstracting and simplifying the programming of individual networked devices.
These methods include TOTA~\cite{tota}, the Hood sensor network abstraction~\cite{hood}, 
the chemical models by~\citeN{VPMSZ-SCP2015}, Butera's ``paintable computing'' hardware model~\cite{butera}, 
and Meld~\cite{Meld}.
\MVtask{In the context of parallel computing, the Bulk Synchronous Parallel (BSP) model~\cite{valiant1990bsp} facilitates programming by enacting synchronisation barriers that allow multiple processors to synchronise, e.g., up to allow for computation to proceed on system-wide rounds.}
\JBtask{Similarly, a number of cloud computing models (e.g., MapReduce~\cite{dean2008mapreduce}) provide bulk programming models that abstract away network structure completely or nearly so.}

Three families of ``top-down'' approaches are complementary to these bottom-up methods.
These higher-level approaches specify tasks for aggregates and then translate, by means of compilers or similar software, from aggregate specifications into a set of individual local actions that can implement the desired aggregate behaviour.
These approaches also tend to build in at least some notion of implicit resilience, though the specifics vary wildly from approach to approach.

One of these families focusses on creating geometric and topological patterns, such as the topological networks of Growing Point Language~\cite{coorephd}, the geometric patterns of Origami Shape Language~\cite{nagpalphd}, the self-healing geometries by~\citeN{clement2003self} and~\citeN{kondacs}, or Yamins' universal patterns~\cite{yamins}.
Another largely disjoint family instead aims at summarisation and streaming of information over regions of space and time.
Examples include sensor-network query languages like TinyDB~\cite{Madden:SIGOPS-2002}, Cougar~\cite{Yao02thecougar},
TinyLime~\cite{Curino05mobiledata}, and Regiment~\cite{regiment}.

The third family, generalising over all of the prior approaches, are general purpose space-time computing models.
Some of these are spatial parallel computing models, such as StarLisp~\cite{starlisp} and systolic computing
(e.g., the works by \citeN{SDEF} and \citeN{ReLaCS}), which use parallel shifting of data on a structured network.
Others, such as MGS~\cite{GiavittoMGS02,GiavittoMGS05}, are more topological in nature.
Because of their generality, this class of computing models can form the basis of a layered approach to the construction of distributed adaptive systems, as in our previous work on field calculus~\cite{DVB-SCP2016,forte2015} and the generalised framework of aggregate programming~\cite{BPV-COMPUTER2015,VBDP-SASO2015}.

\subsubsection*{Self-stabilisation} The primary focus of the work in this paper is to find sufficient conditions for identifying a large class of complex network computations whose outcome is predictable \TYtask{despite} transient changes in \TYtask{their} environment \TYtask{or} inputs, and to express this class in terms of a language of resilient programs that can be used to create such systems by construction.
The notion we focus on requires a unique global state (being reached in finite time) depending only on the state of the environment (topology and sensors), that is,  independent of the initial state.
We speak of this property as \emph{self-stabilisation} since it is contained within the notion of self-stabilisation to \emph{correct} states for distributed systems \cite{dolev}, defined in terms of a set $C$ of correct states into which the system eventually enters in finite time, and then never escapes from---in our case, $C$ is made up of only the single state \MVtask{eventually reached} and corresponding to the intended result of computation, obtained as a function from inputs and environment.

Several versions of the notion of self-stabilisation can be found in literature, surveyed by \citeN{S93c}, from works by \citeN{D73b} to more  abstract ones \cite{AG93}, usually depending on the reference model for the system to study---protocols, state machines, and so on.
In our case, self-stabilisation is studied for computational fields, which can be considered as data structures distributed over space.
\TYtask{However,} since previous work trying to identify general conditions for self-stabilisation (e.g., by \citeN{HP01}) only considers very specific models (e.g. heap-like data structures in a non-distributed settings), it is difficult to make a precise connection with those prior results.

Some variations of the definition of self-stabilisation also deal with different levels of quality (e.g., fairness, performance).
For instance, the notion of superstabilisation \cite{Dolev:1997:SPD:866056} extends the standard self-stabilisation definition by adding a requirement on a ``passage predicate'' that should hold while a system recovers from a specific topological change.
Our work does not address this particular issue, since we  completely equate the treatment of topological changes and changes to the inputs (\TYtask{e.g.}, sensors), and do not address specific performance requirements formally.
Performance is also affected by the fairness assumption adopted: we relied on a notion abstracting from more concrete ones typically used \cite{KC98}\TYtask{---these more concrete models could be applied with our work as well, but would reduce the generality of our results}.
\TYtask{Instead}, we address performance issues in a rather different way: we \TYtask{allow for multiple} different implementations of given building block functions\TYtask{, trading} off reactiveness to different kinds of changes in different ways\TYtask{,} proved equivalent in their final result, and selected based on empirical evaluation.

Concerning specific results on self-stabilisation, some approaches have achieved results \TYtask{that} more closely relate to ours.
\citeN{dolev} introduced  a hop-count gradient (\MVtask{computing minimum hop-by-hop distance from a source node}) that is known to self-stabilise and used it as a preliminary step in the creation of the spanning tree of a graph.
Other authors attempt to devise general methodologies.
\citeN{AV91} depicted a compiler turning any protocol into a self-stabilising one. Though this is technically unrelated to our solution, it shares the philosophy of hiding the details of how self-stabilisation is achieved under the hood of the execution platform: 
in our case in fact, the designer \TYtask{is intended} to focus on the macro-level specification, 
trusting that components behave and interact so as to achieve the global outcome in a self-stabilising way.
\TYtask{Similarly,} \citeN{GH91} suggested that hierarchical composition of self-stabilising programs is self-stabilising\TYtask{---}an idea that is key here to construct a whole functional language of self-stabilising programs.

Concerning the specific technical result achieved here in the context of the field calculus, and apart form the work by \citeN{VBDP-SASO2015} that we extend here, the closest prior work appears to be the work by \citeN{LMCS:selfsabilisation} which, to the best of out knowledge, is the first attempt \TYtask{at} providing a notion of self-stabilisation directly connected to the problem of engineering self-organisation.
As in the present work, self-stabilisation is not proved  for a specific algorithm or system, but is proved for all fields inductively obtained by functional composition of fixed fields (sensors, values) and by a \MVtask{spanning-tree}-inspired spreading process. In this paper we consider a more liberal programming language and also address dynamical properties by simulation.
Finally, an alternative approach to prove self-stabilisation for computational fields is developed in \cite{DBLP:journals/corr/Lluch-LafuenteL16}, in which it is seen in terms of a fix-point semantics, and currently includes only structures \MVtask{based on spanning trees}.

\subsection{Computing with Fields}

The computational model of aggregate computing uses as basic unit of data a dynamically changing
{\em computational field} (or field for short) of values held across many devices in the network.
More precisely, a \emph{field value} $\phi$ is a function
\mbox{$\phi: D \rightarrow \lvalueSet$} that maps each device
$\deviceId$ in domain $D$ to a local value $\lvalue$
in range $\lvalueSet$.
\MVtask{Similarly, a \emph{field evolution} is a dynamically changing field value, and a field computation can be seen as taking field evolutions as input (e.g., from sensors or user inputs) and producing a field evolution as output, from which field values are (distributed) snapshots.}
%
%
%
For example, given an input of a Boolean field mapping certain devices
of interest to {\em true}, an output field of estimated distances to
the nearest such device can be constructed by iterative
aggregation and spreading of information, such that as the input
changes the output changes to match\MVtask{---this computation is referred to in this paper as \texttt{distanceTo}; it is also sometimes elsewhere referred to as a \emph{gradient computation} (e.g., \cite{crf,Beal:FLEX})}.
Note that while \TYtask{the computational field} model maps most intuitively onto
spatially-embedded systems, it can be used for
any distributed computation (though it tends to be best suited for
sparse networks, of which spatially-embedded systems are an example).

Critically to the approach, any field computation can be properly turned into an equivalent single device behaviour, to be iteratively executed by all devices in the network.
Namely, this is carried on in \MVtask{(per-device)} \emph{computation rounds}: sense-eval-broadcast iterations, in which information coming from neighbours and from local sensors are collected \MVtask{in a device}, the computation is evaluated against the \TYtask{device's} local state, and a result of computation is broadcast to neighbours (which will \TYtask{collect and use that state} in their own \TYtask{future} rounds of computation).

\subsection{Proposed workflow}\label{sec:workflow}
   
Our proposed workflow is 
based on computational field calculus~\cite{forte2015} (or field calculus for short), a tiny \FDtask{functional}
language, in which any distributed computation can be expressed,
encapsulated, and safely composed.
\FDtask{Field calculus is a general-purpose language in which it is possible to express both resilient and non-resilient computations.}
\JBtask{For example, field calculus can express computing the minimum value in a network by gossip or by directed aggregation: the gossip implementation is non-resilient, because it cannot track a rising minimum, while the directed aggregation implementation is resilient and can track both rising and falling minimum values.}
Field calculus can, however, be restricted to a sub-language in which
all programs are guaranteed resilient in the \TYtask{sense} of
self-stabilisation, as discussed in the following.

The succinctness of field calculus that makes formal proofs tractable, however, 
is not well suited for the practical engineering of self-organising systems, especially when one needs to scale to complex designs.
This can be mitigated by highly reusable ``building block''
operators capturing common coordination
patterns~\cite{BV-FOCAS2014}, thus raising the abstraction level
and allowing programmers to work with general-purpose functionalities or user-friendly APIs capturing
common use patterns.

These building blocks, despite their desirable resilience properties,
may not be particularly efficient or have desirable dynamical
properties in the specific application at hand.
We thus incorporate a new insight: due to the functional composition model
and modular proof used in establishing the self-stabilising calculus,
any coordination mechanism that is guaranteed to self-stabilise to the
same result as a building block can be substituted for that building
block without affecting the self-stabilisation of the overall program, 
including its final output.
This allows us to include \TYtask{alternative implementations} in our ``library of self-stabilising blocks:'' blocks \TYtask{that} are functionally equivalent but trade off performance
in different ways or have more desirable dynamics
(typically specialised for particular applications of the building
blocks, as the base operators are extremely generic).

Together, these insights enable a two-stage engineering workflow
  \TYtask{that} progressively treats complex
specification, resilience, and efficiency. The starting point for the workflow is a specification of the
desired aggregate behaviour to be implemented by the self-organising
system.  Following this:

\begin{enumerate}
\item The specification is expressed \MVtask{as a} composition of coordination
  patterns (e.g., information spreading, information collection, state
  tracking) that can be mapped onto building block operators.  The
  result is a ``minimal resilient implementation'' 
  guaranteed to be self-stabilising but possibly far
  from optimal.

\item Each building block is then considered for replacement with
  a mechanism from the
  substitution library expected to provide better performance,
  confirming the improvement by analysis or simulation, then iterating, 
  until a satisfying level of performance is achieved. 
\end{enumerate}

Finally, given the intrinsic extensibility of our approach, our library of building blocks can be naturally extended with new blocks and/or alternative block implementations, as will likely be needed when addressing some novel application scenarios. 


\section{Field Calculus}
\label{sec:calculus}

\FDtask{In this section}, we present the first-order Field Calculus~\cite{DVB-SCP2016}, with a syntax inspired by recent DSLs implementing the constructs of the calculus \cite{Casadei:PMLDC16} 
(in Section~\ref{sec:syntax}),\footnote{The original formulation of the Field Calculus~\cite{VDB-FOCLASA-CIC2013,DVB-SCP2016} uses a Scheme-like syntax reflecting earlier implementations~\cite{ProtoSemantics12}.} 
  its operational semantics  (in Section~\ref{sec:semantics}), 
 a convenient  minimal extension allowing for functional parameters  (in Section~\ref{sec:parametrisation})   and examples including key building blocks for the paper (in Section~\ref{sec:blocks}).

\FDtask{Our formulation assumes a denumerable set of device identifiers, ranged over by $\deviceId$, such that each device has a distinguished identifier.
In the rest of the paper each device is represented by its identifier---our formalisation does not provide (and does not need) a syntax for  devices.}

\subsection{Syntax}\label{sec:syntax}

\begin{figure}[!t]
\centering
\centerline{\framebox[\linewidth]{$
	\begin{array}{lcl@{\hspace{5pt}}r}
		\PROGRAM & \BNFcce & \overline{\FUNCTION}  \; \e
																																		&{\footnotesize \mbox{program}} \\[3pt]
		\FUNCTION & \BNFcce &  \defK \; \fname (\overline{\xname}) \; \eqSymK{\e}
																																		&{\footnotesize \mbox{function declaration}} \\[3pt]
		\e & \BNFcce &  \xname \; \BNFmid \; \anyvalue 
                \; \BNFmid \; \letK \; \xname = \e \; \inK \; \e
               \; \BNFmid \; \funvalue(\overline\e) 
&{\footnotesize \mbox{expression}} \\[3pt]
&&  \; \BNFmid \;
		\ifK (\e) \{ \e \} \{ \e \}  \; \BNFmid \; \nbrK\{\e\} \; \BNFmid \; \repK(\e)\{ (\xname) \toSymK{} \e \} 
 \\[3pt]																												
		\anyvalue & \BNFcce &  \lvalue \; \BNFmid \; \fvalue
																																		&{\footnotesize \mbox{value}} \\[3pt]
		\lvalue & \BNFcce &  \dcOf{\dc}{\overline\lvalue}
																																		&{\footnotesize \mbox{local value}} \\[3pt]
		\fvalue & \BNFcce &  \envmap{\overline\deviceId}{\overline\lvalue}
																																		&{\footnotesize \mbox{neighbouring field value}} \\[3pt]
		\funvalue & \BNFcce &  \fname \; \BNFmid \; \oname
																																		&{\footnotesize \mbox{function name}} \\[3pt]
	\end{array}
	$}
}
\caption{Syntax of field calculus.}
\label{fig:syntax}
\end{figure}

Figure~\ref{fig:syntax} presents the syntax of the field calculus. Following~\cite{FJ}, the overbar notation denotes metavariables over sequences and the empty sequence is denoted  by $\emptyseq$. E.g., for expressions, we let $\overline\e$ range over sequences of expressions, written $\e_1,\,\e_2,\,\ldots\,,\e_n$ $(n\ge 0)$. Similarly, formulas containing one or more sequences in overbar notation are supposed to be duplicated for each element of the sequences (which are assumed to share the same length): e.g., $\overline\funvalue(\e) = \overline\anyvalue$ is a shorthand for $\funvalue_i(\e) = \anyvalue_i$ for $i=1 \ldots |\anyvalue|$.

A program  $\PROGRAM$ consists of a sequence of function declarations and of a main expression $\e$. A function declaration $\FUNCTION$ defines a (possibly recursive) function, where $\fname$ is the function name, $\overline\xname$ are the parameters and $\e$ is the body. 
Expressions $\e$ are the main entities of the calculus, modelling a whole field (that is, an expression $\e$ evaluates to a value on every device in the network, thus producing a computational field). An expression can be:
\begin{itemize}
	\item a variable $\xname$, declared either as function formal parameter or as local to a  $\letK$- or $\repK$-expression;
	\item a value, which in turn could be either a \emph{local value} (associating each device to a computational value -- numbers, literals, and so on -- defined through data constructors $\dc$) or a \emph{neighbouring field value} $\fvalue$ (associating each device to a map from neighbours to local values---note \TYtask{that such values are} allowed to appear in intermediate computations but not in source programs);
       \item a $\letK$-expression $\letK \, \xname = \e_0 \, \inK \, \e$, which is evaluated by first computing the value $\anyvalue_0$ of $\e_0$ and then yelding as result the value of the expression obtained from $\e$ by replacing all the occurrences of the variable $\xname$ with the value $\anyvalue_0$;
	\item a function call $\funvalue(\overline\e)$, where $\funvalue$ can be either a \emph{declared function} $\fname$ or a \emph{built-in function} $\oname$ (\TYtask{such as accessing} \MVtask{sensors,} mathematical and logical operators\TYtask{, or data structure operations});
	\item a conditional $\ifK (\e_1) \{ \e_2\} \{ \e_3\}$, which is evaluated by splitting the computation into two sub-networks working in isolation: the devices that evaluated $\e_1$ to $\truevalue$ altogether compute expression $\e_2$, the devices that evaluated $\e_1$ to $\falsevalue$ compute $\e_3$;
	\item a $\nbrK$-expression $\nbrK\{\e\}$, modelling neighbourhood interaction and producing a neighbouring field value $\fvalue$ that represents an ``observation map'' of neighbour's values for expression $\e$, namely, associating each device to a map from neighbours to their latest evaluation of $\e$;
	\item or a $\repK$-expression $\repK(\e_1)\{(\xname) \toSymK{} \e_2\}$, \FDtask{evolving a local state through time by 
evaluating} an expression $\e_2$, substituting the variable $\xname$ with the value calculated for the whole $\repK$-expression at the previous computational round (in the first computation round $\xname$ is substituted with the value of $\e_1$). 
Although the calculus does not model anonymous functions, the syntax  $(\xname) \toSymK{} \e_2$ can be understood as defining an anonymous function with parameter $\xname$ and body $\e_2$.
\end{itemize}
The set of free variables in an expression $\e$ is denoted by 
$\FV{\e}$. As usual, we say that an expression
$\e$ is \emph{closed} iff $\FV{\e}$ is empty.

Values associated to data constructors $\dc$  of arity zero are written by omitting the empty parentheses, i.e.,  we write $\dc$ instead of $\dcOf{\dc}{}$. We assume a constructor for each literal value
 (e.g., $\falsevalue$, $\truevalue$, $0$, $1$, $-1$,...) and a built-in function $\ocname$ for every data constructor $\dc$ of arity $n\ge 1$, i.e., such that $\ocname(\e_1,...,\e_n)$ evaluates to  $\dcOf{\dc}{\lvalue_1,...,\lvalue_n}$ where each $\lvalue_i$ is the value of $\e_i$.
In case $\oname$ is a binary built-in operator, we allow infix notation to enhance readability: i.e., we shall \TYtask{sometimes} write $1+2$ for $+(1, 2)$.
To simplify notation (and following features present in concrete implementations of field calculus \cite{ProtoSemantics12}, \cite{Protelis15}), we shall also overload each (user-defined or built-in) function with local arguments to accept any combination of local and neighbouring field values: the intended meaning is then to apply the given function \emph{pointwise} to its arguments. For example, let $\fvalue$ be the neighbouring field $\envmap{\deviceId_1}{1},\envmap{\deviceId_2}{2},\envmap{\deviceId_3}{3}$ and $\fvaluealt$ be $\envmap{\deviceId_1}{10},\envmap{\deviceId_2}{20},\envmap{\deviceId_3}{30}$, we shall use $\fvalue + \fvaluealt$ for the pointwise sum of the two numerical fields giving the neighbouring field $\envmap{\deviceId_1}{11},\envmap{\deviceId_2}{22},\envmap{\deviceId_3}{33}$, or $1 + \fvalue$ for the field obtained incrementing by $1$ each value in $\fvalue$, namely, $\envmap{\deviceId_1}{2},\envmap{\deviceId_2}{3},\envmap{\deviceId_3}{4}$.

\FDtask{In the following we assume that the calculus is equipped with the type system defined by~\citeN{DVB-SCP2016}, which is  variant of the Hindley-Milner type system~\cite{Damas-Milner:POPL-1982} that has two kinds of types: local types (for local values) and field types (for neighbouring field values). This system associates to each local value a type $\ltype$, and type $\ftypeOf{\ltype}$ to a neighbouring field of elements of type $\ltype$, and correspondingly a type $\type$ to any expression.}

\begin{figure}[t]{
\small
\centerline{\framebox[\textwidth]{ $
\begin{array}{l@{\hspace{-0.2cm}}c@{\hspace{-0.6cm}}r}
\textrm{Built-in Function} & \textrm{Type Signature} & \textrm{Meaning} \\
\hline
\selfK() & () \to \ntype  
& \textit{\footnotesize  device identifier} \\
\texttt{+}, \texttt{-}, \texttt{*}, \texttt{/} & (\ntype,\ntype)\to\ntype
& \textit{\footnotesize  arithmetical operators} \\
\texttt{<}, \texttt{<=}, \texttt{=}, \texttt{>=}, \texttt{>} & (\ntype,\ntype)\to\btype
& \textit{\footnotesize  comparison operators} \\
\texttt{\&\&}, \texttt{||} & (\btype,\btype)\to\btype
& \textit{\footnotesize  boolean operators} \\
\texttt{mux}(\texttt{b}, \lvalue, \lvalue) & \forall\tvar. (\btype,\tvar,\tvar) \to \tvar  
& \textit{\footnotesize  multiplex selection} \\
\texttt{pair}(\lvalue, \lvalue) & \forall\tvar_1\tvar_2.(\tvar_1,\tvar_2)\rightarrow \texttt{tuple}(\tvar_1, \tvar_2)
& \textit{\footnotesize  pair construction} \\
\qp{\,\overline\lvalue\,} & \forall \overline\tvar. (\overline\tvar)\rightarrow \texttt{tuple}(\overline\tvar)
& \textit{\footnotesize  tuple construction} \\
\texttt{1st}(\lvalue), \texttt{2nd}(\lvalue), \texttt{3rd}(\lvalue) & \forall \overline\tvar. (\texttt{tuple}(\overline\tvar)) \rightarrow \tvar_i ~ (i=1,2,3)
& \textit{\footnotesize  tuple element access} \\
\texttt{pickHood}(\fvalue) & \forall\tvar.(\ftypeOf{\tvar})\to\tvar
& \textit{\footnotesize  value in current device} \\
\texttt{foldHood}(\fvalue, \lvalue)(\funvalue) & \forall\tvar. \cp{\ftypeOf{\tvar}, \tvar, (\tvar, \tvar) \to \tvar} \to \tvar
& \textit{\footnotesize  general neighbour aggregation} \\
\texttt{meanHood}(\fvalue) & \forall\tvar.(\ftypeOf{\tvar})\to\tvar
& \textit{\footnotesize  average of neighbour values} \\
\texttt{maxHood}(\fvalue), \texttt{maxHood+}(\fvalue) & \forall\tvar.(\ftypeOf{\tvar})\to\tvar
& \textit{\footnotesize  maximum of neighbour values} \\
\texttt{minHood}(\fvalue), \texttt{minHood+}(\fvalue) & \forall\tvar.(\ftypeOf{\tvar})\to\tvar
& \textit{\footnotesize  minimum of neighbour values} \\
\texttt{minHoodLoc}(\fvalue, \lvalue) & \forall\tvar.(\ftypeOf{\tvar}, \tvar)\to\tvar
& \textit{\footnotesize  minimum of neighbor \& local values} \\
\texttt{nbrRange}(), \texttt{nbrLag}() & ()\to\ftypeOf{\ntype}
& \textit{\footnotesize  space-time distance from neighbours} \\
\texttt{snsNum}() & ()\to \ntype
& \textit{\footnotesize  generic numeric sensor} \\
\texttt{sns\_interval}() & ()\to \ntype
& \textit{\footnotesize  interval with previous round} \\
\end{array}
$}}}
\caption{\FDtask{Built-in functions used throughout this paper, with types and meaning.}}
\label{fig:builtin}
\end{figure}

\FDtask{\begin{remark}
Figure \ref{fig:builtin} presents the collection of built-in functions and operators used in this paper (a small subset of possible built-in functions covered by this calculus).
A few notes regarding these functions:
\begin{itemize}
\item Recall that each built-in function with local arguments is overloaded to work on fields on a pointwise basis. 
\item The multiplex operator $\muxK$ selects between its second and third arguments based on the value of the first one. This is similar to the $\ifK$ keyword but not equivalent: $\muxK$ evaluates both of these arguments everywhere, whereas $\ifK$ only evaluates each on the subspace with the matching Boolean value.
\item A special role is played by the second-order operator \texttt{foldHood} and its specialisations for different aggregation functions (\texttt{minHood}, \texttt{maxHood} and so on) that collapse a field value into a local value (reminiscent of ``reduce'' functions common in parallel programming frameworks like MPI). 
The versions of these operators ending in \texttt{+} also aggregate the value corresponding to the current device (which is otherwise ignored), while the versions ending in \texttt{Loc} also aggregate a given local value in place of the value corresponding to the current device. 
\end{itemize}
\end{remark}
}

\begin{example}\label{exa:distanceTo-informal}
As an example showcasing all \TYtask{classes of construct} at work, consider the following definition of a \texttt{distanceToWithObs} function, mapping each device to an estimated distance to a \texttt{source} area, computed as length of a minimum path that circumvents an \texttt{obstacle} area:

\begin{Verbatim}[fontsize=\fontsize{7pt}{8pt}, frame=single, commandchars=\\\{\}, codes={\catcode`$=3\catcode`^=7\catcode`_=8}]
\km{def} distanceToWithObs(source, obstacle) \{
  \km{if}(obstacle) \{ \pr{infinity} \} \{ distanceTo(source) \}
\}

\km{def} distanceTo(source) \{
   \pr{mux}( source, 0, 
      \km{rep} (\pr{infinity}) \{ (x) => \pr{minHood}(\km{nbr}\{x\} + \pr{nbrRange}())\}
   )
\}
\end{Verbatim}

In the body of function  \texttt{distanceToWithObs}, construct \texttt{if} divides the space in two regions, where \texttt{obstacle} is $\truevalue$ and where it is $\falsevalue$: in the former the output is \texttt{infinity}, in the latter we compute---isolated from the devices in the former area, hence ``circumventing it''---distance estimation by  calling function \texttt{distanceTo}.
 
In the body of function \texttt{distanceTo}, via a purely functional \texttt{mux} built-in operator, we give $0$ on sources (i.e., on devices evaluating \texttt{source} to $\truevalue$).
On other devices, we compute the estimated \TYtask{as} distance being \texttt{infinity} at the beginning, \TYtask{then evolving the distance esimate} by taking the minimum value (\texttt{minHood(field)} is a built-in which returns the minimum value in \texttt{field} or \FDtask{$\infty$} if the field is empty\TYtask{) across neighbour estimates added pointwise to the estimated distance to each neighbor} (obtained by built-in \texttt{nbrRange} modeling a local range sensor).
\end{example}


\subsection{Operational Semantics} \label{sec:semantics}

We now present a formal semantics that can serve both as a specification for implementation of programming languages based on the calculus and for reasoning about its properties.
\MVtask{Differently from models like BSP \cite{valiant1990bsp} that can enact system-wide synchronous rounds in which each device computes exactly once, in our model~}\FDtask{individual devices undergo computation in (local) rounds,}~\MVtask{which are sequential for each device, and interleaved among different devices.}
In each round, a device sleeps for some time, wakes up, gathers information about messages received from neighbours while sleeping, performs an evaluation of the program, and finally emits a message to all neighbours with information about the outcome of computation before going back to sleep.
The scheduling of such rounds across the network is \FDtask{fair and asynchronous---the considered notion of fairness is explained in Section~\ref{ssec:ss_def},}~\MVtask{and basically amounts to the eventual existence of another round for each device and for each moment of time.}
To simplify the notation, we shall assume a fixed program $\PROGRAM$.
We say that ``device $\deviceId$ \emph{fires}'', to mean that the main expression $\emain$ of $\PROGRAM$ is evaluated on $\deviceId$ at a particular round. 

Network evolution is modelled (in Section~\ref{sec:small-step}) by a small-step semantics, given as a transition system $\nettran{}{\act}{}$ on network configurations $\Cfg$, where actions can either be firings of a device or network configuration changes. The semantics of a firing action is defined in terms of the computation that takes place on an individual device,  which is modelled  (in Section~\ref{sec:big-step}) by a big-step semantics. \MVtask{Note that we use small-step semantics in network transitions to capture the step-by-step evolution of a network, while the more abstract big-step semantics is used in individual devices since in that case only the final result of round computation matters---and is in fact unique.}

\subsubsection{Device Semantics}\label{sec:big-step}

The computation that takes place on a single device is formalised  by a big-step semantics, expressed by the judgement $\bsopsem{\deviceId}{\Trees}{\emain}{\vtree}$, to be read ``expression $\emain$ evaluates to $\vtree$ on device $\deviceId$  with respect to environment $\Trees$''.
The result of evaluation is a \emph{value-tree} $\vtree$, which is an ordered tree of values that tracks the results of all evaluated subexpressions of $\emain$. Such a result is made available to $\deviceId$'s neighbours for their subsequent firing (including $\deviceId$ itself, so as to support a form of state across computation rounds). The recently-received value-trees of neighbours are then collected into a \emph{value-tree environment} $\Trees$, implemented as a map from device identifiers to value-trees (written $\envmap{\overline\deviceId}{\overline\vtree}$ as short for $\envmap{\deviceId_1}{\vtree_1},\ldots,\envmap{\deviceId_n}{\vtree_n}$).
Intuitively, the outcome of the evaluation will depend on those value-trees. Figure~\ref{fig:deviceSemantics} (top) defines value-trees  and  value-tree \FDtask{environments---the syntax of values $\anyvalue$ is given in Fig.~\ref{fig:syntax}.}

\begin{example}\label{exa:value-trees}
\FDtask{The graphical representation of the  value trees $\mkvt{5}{\mkvt{2}{},\mkvt{3}{}}$ and $\mkvt{5}{\mkvt{2}{},\mkvt{3}{\mkvt{7}{},\mkvt{1}{},\mkvt{4}{}}}$ is as follows:}
\begin{verbatim}
      5                   5
     / \                 / \
    2   3               2   3
                           /|\
                          7 1 4
\end{verbatim}
\end{example}

\FDtask{In the following, for sake of readability, we sometimes write the value $\anyvalue$ as short for the value-tree $\mkvt{\anyvalue}{}$.
 Following this convention, the value-tree 
$\mkvt{5}{\mkvt{2}{},\mkvt{3}{}}$
is shortened to $\mkvt{5}{2,3}$,  and the value-tree $\mkvt{5}{\mkvt{2}{},\mkvt{3}{\mkvt{7}{},\mkvt{4}{},\mkvt{4}{}}}$ is shortened to $\mkvt{5}{2,\mkvt{3}{7,1,4}}$.
}

Figure~\ref{fig:deviceSemantics} (bottom) defines the judgement $\bsopsem{\deviceId}{\Trees}{\e}{\vtree}$, where:
\emph{(i)} $\deviceId$ is the identifier of the current device;
\emph{(ii)} $\Trees$ is the neighbouring field of the value-trees produced by the most recent evaluation of (an expression corresponding to) $\e$ on $\deviceId$'s neighbours;
\emph{(iii)} $\e$ is a closed run-time expression (i.e., a closed  expression that may contain neighbouring field values);
\emph{(iv)} the value-tree $\vtree$  represents the values computed for all the expressions encountered during the evaluation of $\e$---in particular the root of the value tree $\vtree$, denoted by $\vrootOf{\vtree}$, is the  value computed for expression $\e$. \FDtask{The auxiliary function $\vroot$ is defined in Figure~\ref{fig:deviceSemantics} (second frame).}

The operational semantics rules are based on rather standard rules for functional languages, extended so as to be able to evaluate a subexpression $\e'$ of $\e$ with respect to\ the value-tree environment $\Trees'$ obtained from $\Trees$ by extracting the corresponding subtree (when present) in the value-trees in the range of $\Trees$. This process, called \emph{alignment}, is modelled by the auxiliary function $\pi$ \FDtask{defined in Figure~\ref{fig:deviceSemantics} (second frame). This function} has two different behaviours (specified by its subscript or superscript): $\piIof{i}{\vtree}$ extracts the $i$-th subtree of $\vtree$; while $\piBof{\lvalue}{\vtree}$ extracts the last subtree of $\vtree$, \emph{if} the root of the first subtree of $\vtree$ is equal to the local (boolean) value $\lvalue$ (thus implementing a filter specifically designed for the $\ifK$ construct). 
Auxiliary functions $\vroot$ and $\pi$  \FDtask{apply pointwise on value-tree environments, as defined in Figure~\ref{fig:deviceSemantics} (second frame).}

\begin{figure}[!t]{
\small
 \framebox[1\textwidth]{
 $\begin{array}{l}
 \textbf{Value-trees and value-tree environments:}\\
\begin{array}{lcl@{\hspace{6cm}}r}
\vtree & \BNFcce &  \mkvt{\anyvalue}{\overline{\vtree}}    &   {\footnotesize \mbox{value-tree}} \\
\Trees & \BNFcce & \envmap{\overline{\deviceId}}{\overline{\vtree}}   &   {\footnotesize \mbox{value-tree environment}}
\end{array}\\[10pt]
\hline\\[-8pt]
\FDtask{\textbf{Auxiliary functions:}}\\
\begin{array}{l}
\begin{array}{l@{\hspace{0.4cm}}l}
\vrootOf{\mkvt{\anyvalue}{\overline{\vtree}}}  =   \anyvalue
&
\\
\piIof{i}{\mkvt{\anyvalue}{\vtree_1,\ldots,\vtree_n}}  =   \vtree_i
\quad \mbox{if} \; 1\le i \le n
&
\piBof{\lvalue}{\mkvt{\anyvalue}{\vtree_1,\vtree_2}}  =   \vtree_2
\quad \mbox{if} \;  \vrootOf{\vtree_1} = \lvalue
\\
\piIof{i}{\vtree}  =   \emptyseq \quad \mbox{otherwise} 
&
 \piBof{\lvalue}{\vtree}  =   \emptyseq \quad \mbox{otherwise}
\\  
\end{array}
\\
\mbox{For } \auxNAME\in\rho,\piI{i},\piB{\lvalue}:
\quad 
\left\{\begin{array}{lcll}
 \aux{\envmap{\deviceId}{\vtree}}  & =  & \envmap{\deviceId}{\aux{\vtree}} & \quad \mbox{if} \; \aux{\vtree} \not=\emptyseq  
\\
\aux{\envmap{\deviceId}{\vtree}}  & =   & \emptyseq  & \quad \mbox{if} \; \aux{\vtree}=\emptyseq  
\\
\aux{\Trees,\Trees'}  & =  &  \aux{\Trees},\aux{\Trees'}
\end{array}\right.   
\\
\begin{array}{l@{\hspace{1.28cm}}l}
\args{\fname} = \overline{\xname} \quad \mbox{if } \, \defK \; \fname (\overline{\xname}) \; \{\e\}
&
\body{\fname} = \e  \quad \mbox{if } \, \defK \; \fname (\overline{\xname}) \; \{\e\}
\end{array}
\end{array}\\
\hline\\[-10pt]
\FDtask{\textbf{Syntactic shorthands:}}\\
\begin{array}{l@{\hspace{2pt}}l@{\hspace{2pt}}l}
\bsopsem{\deviceId}{\piIofOv{\Trees}}{\overline{\e}}{\overline{\vtree}}
&
  \textrm{where~~} |\overline{\e}|=n
&
  \textrm{for~~}
  \bsopsem{\deviceId}{\piIof{1}{\Trees}}{\e_1}{\vtree_1}
  \;
    \cdots
    \;
    \bsopsem{\deviceId}{\piIof{n}{\Trees}}{\e_n}{\vtree_n}\\
\vrootOf{\overline{\vtree}}
&
  \textrm{where~~} |\overline{\vtree}|=n
  & \textrm{for~~}
\vrootOf{\vtree_1},\ldots,\vrootOf{\vtree_n}\\
\substitution{\overline{\xname}}{\vrootOf{\overline{\vtree}}}
&   \textrm{where~~} |\overline{\xname}|=n
  &
  \textrm{for~~}
\substitution{\xname_1}{\vrootOf{\vtree_1}}~\ldots\quad\substitution{\xname_n}{\vrootOf{\vtree_n}}
\end{array}\\
\hline\\[-10pt]
\textbf{Rules for expression evaluation:} \hspace{4.5cm} 
  \boxed{\bsopsem{\deviceId}{\Trees}{\e}{\vtree}}
\skiptransition
\begin{array}{c}
\nullsurfaceTyping{E-LOC}{
\bsopsem{\deviceId}{\Trees}{\lvalue}{\mkvt{\lvalue}{}}
}
\qquad\qquad
\surfaceTyping{E-FLD}{\qquad \fvalue' = \proj{\fvalue}{\domof{\Trees}\cup\{\deviceId\}}}{
\bsopsem{\deviceId}{\Trees}{\fvalue}{\mkvt{\fvalue'}{}}
}
\skiptransition\\[-8pt]
\surfaceTyping{E-LET}{ \quad
\begin{array}{c}
  \bsopsem{\deviceId}{\piIof{1}{\Trees}}{\e_1}{\vtree_1} \qquad 
  \bsopsem{\deviceId}{\piIof{2}{\Trees}}{\applySubstitution{\e_2}{\substitution{\xname}{\vrootOf{\vtree_1}}}}{\vtree_2}
\end{array}
 }{
\bsopsem{\deviceId}{\Trees}{\letK \; \xname = \e_1 \; \inK \; \e_2}{\mkvt{\vrootOf{\vtree_2}}{\vtree_1,\vtree_2}}
}
\skiptransition\\[-6pt]
\surfaceTyping{E-B-APP}{  \quad
\begin{array}{c}
  \bsopsem{\deviceId}{\piIofOv{\Trees}}{\overline{\e}}{\overline{\vtree}}
  \qquad \anyvalue=\builtinop{\oname}{\deviceId}{\Trees}(\vrootOf{\overline{\vtree}})
\end{array}
 }{
\bsopsem{\deviceId}{\Trees}{\oname(\overline{\e})}{\mkvt{\anyvalue}{\overline{\vtree}}}
}
\skiptransition\\[-6pt]
\surfaceTyping{E-D-APP}{ \quad
\begin{array}{c}
  \bsopsem{\deviceId}{\piIofOv{\Trees}}{\overline{\e}}{\overline{\vtree}} \qquad 
  \bsopsem{\deviceId}{\Trees}{\applySubstitution{\body{\fname}}{\substitution{\args{\fname}}{\vrootOf{\overline{\vtree}}}}}{\vtree'}
\end{array}
 }{
\bsopsem{\deviceId}{\Trees}{\fname(\overline{\e})}{\mkvt{\vrootOf{\vtree'}}{\overline{\vtree},\vtree'}}
}
\skiptransition\\[-5pt]
\surfaceTyping{E-NBR}{
         \qquad
     \bsopsem{\deviceId}{\piIof{1}{\Trees}}{\e}{\vtree}
\qquad
 \fvalue=\mapupdate{\vrootOf{\piIof{1}{\Trees}}}{\envmap{\deviceId}{\vrootOf{\vtree}}}
 }{
\bsopsem{\deviceId}{\Trees}{\nbrK\{\e\}}{\mkvt{\fvalue}{\vtree}}
}
%
%
\skiptransition\\[-6pt]
\surfaceTyping{E-REP}{
	\;
	\begin{array}{l}
     \bsopsem{\deviceId}{\piIof{1}{\Trees}}{\e_1}{\vtree_1} \\
     \bsopsem{\deviceId}{\piIof{2}{\Trees}}{\applySubstitution{\e_2}{\substitution{\xname}{\lvalue_0}}}{\vtree_2}~~
	\end{array}
	\lvalue_0 \! = \!\left\{\begin{array}{ll}
                             \vrootOf{\piIof{2}{\Trees}}(\deviceId) & \mbox{if} \;  \deviceId \in \domof{\Trees} \\
                             \vrootOf{\vtree_{1}} & \mbox{otherwise}
                           \end{array}\right.
 }{
\bsopsem{\deviceId}{\Trees}{\repK(\e_1)\{(\xname) \; \toSymK \; \e_2\}}{\mkvt{\vrootOf{\vtree_{2}}}{\vtree_1,\vtree_2}}
}
%
\skiptransition\\[-4pt]
\surfaceTyping{E-IF}{
\quad
     \bsopsem{\deviceId}{\piIof{1}{\Trees}}{\e}{\vtree_1}
\quad
\vrootOf{\vtree_{1}}\in\{\truevalue,\falsevalue\}
\quad
     \bsopsem{\deviceId}{\piBof{\vrootOf{\vtree_{1}}}{\Trees}}{\e_{\vrootOf{\vtree_{1}}}}{\vtree}
 }{
\bsopsem{\deviceId}{\Trees}{\ifK (\e) \{\e_\truevalue\} \{\e_\falsevalue\}}{\mkvt{\vrootOf{\vtree}}{\vtree_1,\vtree}}
}
%
%
\end{array}
\end{array}$}
}
\vspace{-0.1cm}
 \caption{Big-step operational semantics for expression evaluation.} \label{fig:deviceSemantics}
\end{figure}

Rules \ruleNameSize{[E-LOC]} and \ruleNameSize{[E-FLD]} model the evaluation of expressions that are either a local value or a neighbouring field value, respectively: note that in \ruleNameSize{[E-FLD]} we take care of restructing the domain of a neighbouring field value to the only set of neighbour devices as reported in $\Trees$.

Rule \ruleNameSize{[E-LET]} is fairly standard: it first evaluates  $\e_1$ and then evaluates  the expression obtained from $\e_2$ by replacing all the occurrences of the variable $\xname$ with the value of $\e_1$.

\FDtask{Rule \ruleNameSize{[E-B-APP]} models the application of built-in functions.
It is used to evaluate expressions of the form $\oname(\e_1 \cdots \e_n)$, where $n\ge 0$. It produces the value-tree $\mkvt{\anyvalue}{\vtree_{1},\ldots,\vtree_{n}}$, where  $\vtree_{1},\ldots,\vtree_{n}$ are the value-trees produced by the evaluation of the actual parameters  $\e_{1},\ldots,\e_{n}$  and $\anyvalue$ is the value returned by the function.}
The rule exploits the special auxiliary function $\builtinop{\oname}{\deviceId}{\Trees}$, whose actual definition is abstracted away. This is such that $\builtinop{\oname}{\deviceId}{\Trees}(\overline\anyvalue)$ computes the result of applying built-in function $\oname$ to values $\overline\anyvalue$ in the current environment of the device $\deviceId$.
\FDtask{In particular: the built-in 0-ary function $\selfK$ gets evaluated to the current device identifier (i.e.,  $\builtinop{\selfK}{\deviceId}{\Trees}() =\deviceId$), and  mathematical operators have their standard meaning, which is independent from $\deviceId$ and $\Trees$ (e.g., $\builtinop{+}{\deviceId}{\Trees}(2,3)=5$).
\begin{example}
Evaluating the expression $\mathtt{+(2, 3)}$ produces the value-tree 
$\mkvt{5}{2,3}$.
The value of the whole expression, $5$, has been computed by using rule \ruleNameSize{[E-B-APP]} to evaluate the application of the sum operator $+$ to the values $2$  (the root of the first subtree of the
 value-tree) and $3$  (the root of the second subtree of the value-tree). 
\end{example}}

\FDtask{The $\builtinop{\oname}{\deviceId}{\Trees}$ function also encapsulates measurement variables such as \texttt{nbrRange} and interactions with the external world via sensors and actuators.}

Rule \ruleNameSize{[E-D-APP]} models the application of a user-defined function. \FDtask{It is used to evaluate expressions of the form $\fname(\e_1 \ldots \e_n)$, where $n\ge 0$.} It resembles rule \ruleNameSize{[E-B-APP]} while producing a value-tree with one more subtree $\vtree'$, which is produced by evaluating the body of the function $\fname$ (denoted by $\body{\fname}$)  substituting the formal  parameters of the function (denoted by $\args{\fname}$) with the values obtained evaluating $\e_1, \ldots \e_n$.

Rule \ruleNameSize{[E-REP]} implements internal state evolution through computational rounds: $\repK(\e_1)\{(\xname) \toSymK{} \e_2\}$ evaluates to $\applySubstitution{\e_2}{\substitution{\xname}{\anyvalue}}$ where $\anyvalue$ is obtained from $\e_1$ on the first firing of a device, from the previous value of the whole $\repK$-expression otherwise.

\FDtask{\begin{example}
To illustrate rule \ruleNameSize{[E-REP]}, as well as computational rounds, we consider program \mbox{\texttt{rep(0)\{(x) => +(x, 1)\}}}.
The first firing of a device $\deviceId$  is performed against the empty tree environment. Therefore, according to rule  \ruleNameSize{[E-REP]}, 
to evaluate  \mbox{\texttt{rep(0)\{(x) => +(x, 1)\}}} means to evaluate the subexpression \mbox{\texttt{+(0, 1)}}, obtained from \mbox{\texttt{+(x, 1)}} 
by replacing \mbox{\texttt{x}} with \mbox{\texttt{0}}. This produces the  value-tree $\vtree=\mkvt{1}{0, \mkvt{1}{0,1}}$, where root $1$ is the overall result as usual,
 while its sub-trees are the result of evaluating the first and second argument respectively.
Any subsequent firing of the device $\deviceId$ is performed with respect to\ a tree environment $\Trees$ that associates to $\deviceId$ the outcome $\vtree$ of the most recent firing of $\deviceId$. 
Therefore, evaluating    \mbox{\texttt{rep(0)\{(x) => +(x, 1)\}}} at the second firing means to evaluate the subexpression \mbox{\texttt{+(1, 1)}},
 obtained from \mbox{\texttt{+(x, 1)}} by replacing \mbox{\texttt{x}} with \mbox{\texttt{1}}, which is the root of $\vtree$.
 Hence the results of computation are $1$, $2$, $3$, and so on.
\end{example}}

Rule \ruleNameSize{[E-NBR]} models device interaction. It first collects neighbour's values for expressions $\e$ as $\TYtask{\fvalue} = \vrootOf{\piIof{1}{\Trees}}$, then evaluates $\e$ in $\deviceId$ and updates the corresponding entry in $\TYtask{\fvalue}$.

\FDtask{\begin{example}\label{exa-NBR-device-semantics}
To illustrate rule \ruleNameSize{[E-NBR]}, consider 
the program:
\[\e'=\mathtt{minHood}(\nbrK\{\snsNumK()\})\]
where the 1-ary built-in function $\minHoodK$ returns the lower limit of values in the range of its neighbouring field argument, and  the 0-ary built-in
 function $\snsNumK$ returns the numeric value measured by a sensor. Suppose that the program runs  on a network of three devices
 $\deviceId_A$,  $\deviceId_B$, and $\deviceId_C$ where:
 \begin{itemize}
 \item
  $\deviceId_B$ and  $\deviceId_A$ are mutually connected,  $\deviceId_B$ and $\deviceId_C$ are mutualy connected, while  $\deviceId_A$ and  $\deviceId_C$ are not connected;
 \item
  $\snsNumK$ returns  \texttt{1} on $\deviceId_A$,  \texttt{2} on $\deviceId_B$, and \texttt{3} on $\deviceId_C$; and
  \item
   all devices have an initial empty tree-environment $\emptyset$.
  \end{itemize}
 Suppose that device $\deviceId_A$ is the first device that fires:
the evaluation of  $\snsNumK()$ on $\deviceId_A$ yields $1$ (by rules  \ruleNameSize{[E-LOC]} and \ruleNameSize{[E-B-APP]}, since
 $\builtinop{\snsNumK}{\deviceId_A}{\emptyset}()=1$); the evaluation of  $\nbrK\{\snsNumK()\}$ on $\deviceId_A$ yields $\mkvt{(\envmap{\deviceId_A}{1})}{2}$ 
(by rule \ruleNameSize{[E-NBR]}); and the evaluation of  $\e'$ on $\deviceId_A$ yields
\[
\begin{array}{l@{\quad}c@{\quad}l}
    \vtree_A & = & \mkvt{1}{\mkvt{(\envmap{\deviceId_A}{1})}{1}} 
\end{array}
\] 
(by rule \ruleNameSize{[E-B-APP]}, since $\builtinop{\minHoodK}{\deviceId_A}{\emptyset}(\envmap{\deviceId_A}{1})=1$). Therefore, at its first fire, device $\deviceId_A$ produces the value-tree $\vtree_A$.
 Similarly, if device $\deviceId_C$ is the second device that fires, it produces the value-tree
\[
\begin{array}{l@{\quad}c@{\quad}l}
	\vtree_C & = & \mkvt{3}{\mkvt{(\envmap{\deviceId_C}{3})}{3}}
\end{array}
\]
Suppose that device $\deviceId_B$ is the third device that  fires. Then the  evaluation of $\e'$ on $\deviceId_B$ is performed with respect to the value tree 
environment \mbox{$\Trees_{B} = (\envmap{\deviceId_A}{\vtree_A},\;\envmap{\deviceId_C}{\vtree_C})$} and the evaluation of its subexpressions  
$\nbrK\{\snsNumK()\}$  and $\snsNumK()$ is performed, respectively, with respect to the  following  value-tree environments  obtained from $\Trees_{B}$ by alignment:
\[
\begin{array}{l}
	\Trees'_{B}  \; = \;  \piIof{1}{\Trees_{B}} \; = \; 
(\envmap{\deviceId_A}{\mkvt{(\envmap{\deviceId_A}{1})}{1}},\;\;\envmap{\deviceId_C}{\mkvt{(\envmap{\deviceId_C}{3})}{3}})
	\\
	\Trees''_{B}     \; = \;  \piIof{1}{\Trees'_{B}}  \; = \;  
 (\envmap{\deviceId_A}{1},\;\;\envmap{\deviceId_C}{3})
\end{array}
\] 
We thus have that $\builtinop{\snsNumK}{\deviceId_B}{\Trees''_B}()=2$; the evaluation of  $\nbrK\{\snsNumK()\}$ on $\deviceId_B$ with respect to\ $\Trees'_B$ produces the value-tree
$\mkvt{\phi}{2}$ where $\phi = (\envmap{\deviceId_A}{1},\envmap{\deviceId_B}{2},\envmap{\deviceId_C}{3})$; and $\builtinop{\minHoodK}{\deviceId_B}{\Trees_B}(\phi)=1$. 
Therefore the  evaluation of $\e'$ on $\deviceId_B$ produces the value-tree $\vtree_B  = \mkvt{1}{\mkvt{\phi}{2}}$.
Note that, if the network topology and the values of the sensors will not change, then: any subsequent fire of device $\deviceId_B$ will yield a value-tree with root $1$ (which is the minimum of $\snsNumK$ across $\deviceId_A$,  $\deviceId_B$ and $\deviceId_C$); any subsequent fire of device $\deviceId_A$ will yield a value-tree with root $1$ (which is the minimum of $\snsNumK$ across $\deviceId_A$ and  $\deviceId_B$); and any subsequent fire of device $\deviceId_C$ will yield a value-tree with root $2$ (which is the minimum of $\snsNumK$ across $\deviceId_B$ and  $\deviceId_C$).
\end{example}}

Rule \ruleNameSize{[E-IF]} is almost standard, except that it performs domain restriction $\piBof{\truevalue}{\Trees}$ (resp. $\piBof{\falsevalue}{\Trees}$) in order to guarantee that subexpression $\e_\truevalue$ is not matched against value-trees obtained from $\e_\falsevalue$ (and vice-versa).

\subsubsection{Network Semantics}\label{sec:small-step}

The overall network evolution is formalised by the small-step operational semantics given in Figure~\ref{fig:networkSemantics} as a transition system on network configurations $\Cfg$.
Figure \ref{fig:networkSemantics} (top) defines key syntactic elements to this end.
$\Field$ models the overall status of the devices in the network at a given time, as a map from device identifiers to value-tree environments. From it, we can define the state of the field at that time by summarising the current values held by devices.
$\Topo$ models \emph{network topology}, namely, a directed neighbouring graph, as a map from device identifiers to set of identifiers (denoted as $I$).
$\Sens$ models \emph{sensor (distributed) state}, as a map from device identifiers to (local) sensors (i.e., sensor name/value maps denoted as $\senstate$).
Then, $\Envi$ (a couple of topology and sensor state) models the system's environment.
So, a whole network configuration $\Cfg$ is a couple of a status field and environment.

\begin{figure}[!t]{
\small
 \framebox[1\textwidth]{
 $\begin{array}{l}
 \textbf{System configurations and action labels:}\\
\begin{array}{lcl@{\hspace{5.5cm}}r}
\Field & \BNFcce &  \envmap{\overline\deviceId}{\overline\Trees}    &   {\footnotesize \mbox{status field}} \\
\Topo & \BNFcce &  \envmap{\overline\deviceId}{\overline\devset}    &   {\footnotesize \mbox{topology}} \\
\Sens & \BNFcce &  \envmap{\overline\deviceId}{\overline\senstate}    &   {\footnotesize \mbox{sensors-map}} \\
\Envi & \BNFcce &  \EnviS{\Topo}{\Sens}    &   {\footnotesize \mbox{environment}} \\
\Cfg & \BNFcce &  \SystS{\Envi}{\Field}    &   {\footnotesize \mbox{network configuration}} \\
\act & \BNFcce &  \deviceId \;\BNFmid\; \envact    &   {\footnotesize \mbox{action label}} 
\\[8pt]
\end{array}\\
\hline\\[-4pt]
\textbf{Environment well-formedness:}\\
\begin{array}{l}
\wfn{\EnviS{\Topo}{\Sens}} \textrm{~~holds if $\Topo,\Sens$ have same domain, and $\Topo$'s values do not escape it.}
\\
\end{array}\\[10pt]
\hline\\[-9pt]
\textbf{Transition rules for network evolution:} \hfill
  \boxed{\nettran{\Cfg}{\act}{\Cfg}}
  \\[0.2cm]
\vspace{0.5cm}
\begin{array}{c}
\netopsemRule{N-FIR}
                 {
                 \begin{array}{l}
	                 \Envi=\EnviS{\Topo}{\Sens} \\
    	              \Topo(\deviceId)= \overline\deviceId 
                 \end{array}
                  ~ \bsopsem{\deviceId}{\filter(\Field)(\deviceId)}{\emain}{\vtree} \; (\mbox{w.r.t.} \; \Sens(\deviceId))
                  \quad
                 \Field_1=\envmap{\overline\deviceId}{\{\envmap{\deviceId}{\vtree}\}}}
                 {\nettran{\SystS{\Envi}{\Field}}{\deviceId}{\SystS{\Envi}{\mapupdate{\filter(\Field)}{\Field_1}}}
                 }
\skiptransition
\netopsemRule{N-ENV}
                 {\qquad \wfn{\Envi'}\qquad \Envi'=\EnviS{\Topo}{\envmap{\overline\deviceId}{\overline\senstate}} \qquad
                  \Field_0=\envmap{\overline\deviceId}{\emptyset}
                 }
                 {\nettran{\SystS{\Envi}{\Field}}{\envact}{\SystS{\Envi'}{\mapupdate{\Field_0}{\Field}}}
                 }\\[-10pt]
\end{array}\\
\end{array}$}
}
\caption{Small-step operational semantics for network evolution.} \label{fig:networkSemantics}
\end{figure}

We use the following notation for status fields. Let $\envmap{\overline\deviceId}{\Trees}$ denote a map from device identifiers $\overline\deviceId$ to the same value-tree environment $\Trees$. Let $\mapupdate{\Trees_0}{\Trees_1}$ denote the value-tree environment with domain $\domof{\Trees_0} \cup \domof{\Trees_1}$ coinciding with $\Trees_1$ in the domain of $\Trees_1$ and with $\Trees_0$ otherwise. Let $\mapupdate{\Field_0}{\Field_1}$ denote the status field with the \emph{same domain} as $\Field_0$ made of $\envmap{\deviceId}{\mapupdate{\Field_0(\deviceId)}{\Field_1(\deviceId)}}$ for all $\deviceId$ in the domain of $\Field_1$, $\envmap{\deviceId}{\Field_0(\deviceId)}$ otherwise.

We consider transitions $\nettran{\Cfg}{\act}{\Cfg'}$ of two kinds: firings, where $\act$ is the corresponding device identifier, and environment changes, where $\act$ is the special label $\envact$. This is formalised in Figure \ref{fig:networkSemantics} (bottom).
Rule \ruleNameSize{[N-FIR]} models a computation round (firing) at device $\deviceId$: it takes the local value-tree environment filtered out of old values $\filter(\Field)(\deviceId)$;\footnote{Function $\filter(\Field)$ in rule \ruleNameSize{[N-FIR]} models a filtering operation that clears out old stored values from the value-tree environments in $\Field$, implicitly based on space/time tags.} then by the single
 device semantics it obtains the device's value-tree $\vtree$,\footnote{\FDtask{We shall assume that any device firing is guaranteed to terminate in any environmental condition. 
Termination of a device firing is clearly not decidable, but we shall assume---without loss of generality for the results of this paper---that a decidable subset of the termination fragment can be identified 
(e.g., by ruling out recursive user-defined functions or by applying standard static analysis techniques for termination).}} which is used to update the system configuration of  $\deviceId$ and of $\deviceId$'s neighbours.

Rule \ruleNameSize{[N-ENV]} takes into account the change of the environment to a new \emph{well-formed} environment $\Envi'$---environment well-formedness is specified by the predicate $\wfn{\Envi}$  in Figure~\ref{fig:networkSemantics} (middle). Let $\overline\deviceId$ be the domain of $\Envi'$. We first construct a status field $\Field_0$ associating to all the devices of $\Envi'$ the empty context $\emptyset$. Then, we adapt the existing status field $\Field$ to the new set of devices: $\mapupdate{\Field_0}{\Field}$ automatically handles removal of devices, map of new devices to the empty context, and retention of existing contexts in the other devices.

\FDtask{
\begin{example}\label{exa:network-semantics}
Consider a network of devices with the program \[\e'=\minHoodK(\nbrK\{\snsNumK()\})\] introduced in Example~\ref{exa-NBR-device-semantics}. The network configuration illustrated at the beginning of Example~\ref{exa-NBR-device-semantics} can be generated by applying rule \ruleNameSize{[N-ENV]} to the empty network configuration. I.e., we have
\[
\nettran{\SystS{\EnviS{\emptyset}{\emptyset}}{\emptyset}}{\envact}{\SystS{\Envi_0}{\Field_0}}
\]
where
\begin{itemize}
\item
$\Envi_0=\EnviS{\Topo_0}{\Sens_0}$,
\item
$\Topo_0=(\envmap{\deviceId_A}{\{\deviceId_B\}},\envmap{\deviceId_B}{\{\deviceId_A,\deviceId_C\}},\envmap{\deviceId_C}{\{\deviceId_B\}})$,
\item
$\Sens_0=(\envmap{\deviceId_A}{(\envmap{\snsNumK}{1})},\envmap{\deviceId_B}{(\envmap{\snsNumK}{2})},\envmap{\deviceId_C}{(\envmap{\snsNumK}{3})})$, and
\item
$\Field_0=(\envmap{\deviceId_A}{\emptyset},\envmap{\deviceId_B}{\emptyset},\envmap{\deviceId_C}{\emptyset})$.
\end{itemize}
Then, the tree fires of devices $\deviceId_A$, $\deviceId_C$ and $\deviceId_B$ illustrated in Example~\ref{exa-NBR-device-semantics}  correspond to  the following transitions, respectively.
\begin{enumerate}
\item
$
\nettran{\SystS{\Envi_0}{\Field_0}}{\deviceId_A}{\SystS{\Envi_0}{\Field'}}
$,
where
\begin{itemize}
\item
$\Field'=(\envmap{\deviceId_A}{(\envmap{\deviceId_A}{\vtree_A})},\;\envmap{\deviceId_B}{(\envmap{\deviceId_A}{\vtree_A})},\;\envmap{\deviceId_C}{\emptyset})$, and
\item
$\vtree_A  =  \mkvt{1}{\mkvt{(\envmap{\deviceId_A}{1})}{1}}$;
\end{itemize}
\item
$
\nettran{\SystS{\Envi_0}{\Field'}}{\deviceId_C}{\SystS{\Envi_0}{\Field''}}
$,
where
\begin{itemize}
\item
$\Field''=(\envmap{\deviceId_A}{(\envmap{\deviceId_A}{\vtree_A})},\;\envmap{\deviceId_B}{(\envmap{\deviceId_A}{\vtree_A},\envmap{\deviceId_C}{\vtree_C})},\;\envmap{\deviceId_C}{(\envmap{\deviceId_C}{\vtree_C})})$, and
\item
$\vtree_C  =  \mkvt{1}{\mkvt{(\envmap{\deviceId_C}{3})}{3}}$; 
\end{itemize}
\item
$
\nettran{\SystS{\Envi_0}{\Field''}}{\deviceId_B}{\SystS{\Envi_0}{\Field'''}}
$,
where
\begin{itemize}
\item
$\Field'''=(\envmap{\deviceId_A}{(\envmap{\deviceId_A}{\vtree_A},\envmap{\deviceId_B}{\vtree_B})},$
\\
$\mbox{$\qquad\quad$}$ $\envmap{\deviceId_B}{(\envmap{\deviceId_A}{\vtree_A},\envmap{\deviceId_B}{\vtree_B},\envmap{\deviceId_C}{\vtree_C})},$
\\
$\mbox{$\qquad\quad$}$ $\envmap{\deviceId_C}{(\envmap{\deviceId_B}{\vtree_B},\envmap{\deviceId_C}{\vtree_C})})$,
\item
$\vtree_B  = \mkvt{1}{\mkvt{\phi}{2}}$, and
\item
$\phi = (\envmap{\deviceId_A}{1},\envmap{\deviceId_B}{2},\envmap{\deviceId_C}{3})$.
\end{itemize}
\end{enumerate}
\end{example}
}

\subsection{A Minimal Convenient Extension: Functional Parametrisation}\label{sec:parametrisation}

The pragmatic convenience of the calculus defined so far can be improved to express general-purpose \emph{building blocks}, which are parametric algorithms designed to be applied to a broad class of problems, and necessarily make use of functional parameters to tune their behaviour. 
%
%

 \FDtask{To this end,} we extend the syntax of user-defined functions to admit  \emph{functional parameters}, ranged over by $\zname$. Such \emph{extended functions} can be defined as $\defK \; \fname(\overline\xname)(\overline\zname) \{ \e \}$ and called by $\fname(\overline\e)(\overline\funvalue)$ where the arguments $\overline\funvalue$ can be either names of \emph{plain} (i.e., non-extended)  functions or functional parameters---names  of extended functions are not allowed to  be passed as arguments.
 By convention, we omit the second parentheses whenever no functional parameters are present; so that functions without functional parameters can be defined and called as usual. We also allow the presence of built-in functions admitting functional parameters (e.g., the field aggregator $\foldK(\xname,\yname)(\zname)$ which combines values in a field $\xname$ through an initial value $\yname$ and  a binary function $\zname$ given as functional parameter).

We remark that a functional parameter $\zname$ (like any other function name) is not an expression by itself, and it only constitutes one when provided with appropriate arguments  or passed as argument to a function. This implies for instance that $(\ifK (\e) \{ \zname_1 \} \{ \zname_2 \})(\overline\e)$ is \emph{not} a valid expression.

A program in the extended syntax can be converted to a program in plain first-order syntax by systematically substituting each  call  $\fname(\overline\e)(\overline\funvalue)$ to an extended function  $\defK \; \fname(\overline\xname)(\overline\zname)\{\e\}$ (where the arguments $\overline\funvalue$ do not contain functional parameters) by a call  $\fname_{\overline\funvalue}(\overline\e)$ to a plain  function $\fname_{\overline\funvalue}$ defined as $\defK \; \fname_{\overline\funvalue}(\overline\xname) \{ \applySubstitution{\e}{\substitution{\overline\zname}{\overline\funvalue}} \}$---thus interpreting  functional parameters as  macro parameters.\footnote{This rewriting process always terminates. Consider $F$ as the set of distinct plain function names that are passed as parameters to extended functions in any point of the program. Then an extended function with $k$ functional parameters can be instantiated at most once for each combination of functions in $F$, that is, at most $n^k$ times where $n$ is the cardinality of $F$.}
For example, the following program (comparing minimum temperature and maximum threshold across a network):

\begin{Verbatim}[fontsize=\fontsize{7pt}{8pt}, frame=single, commandchars=\\\{\}, codes={\catcode`$=3\catcode`^=7\catcode`_=8}]
\km{def} foldwithlocal(field, local, initial)(aggregate) \{
  aggregate(\pr{foldHood}(field, initial)(aggregate), local)
\}

\km{def} gossip(null)(aggregate, sensor) \{
  \km{rep} (initial) \{ (x) => foldwithlocal(\km{nbr}\{x\}, sensor(), initial)(aggregate) \}
\}

gossip(\pr{infinity})(\pr{min}, \pr{sns\_temp}) < gossip(\pr{-infinity})(\pr{max}, \pr{sns\_threshold})
\end{Verbatim}
can be rewritten eliminating functional parameters in the following way:
\begin{Verbatim}[fontsize=\fontsize{7pt}{8pt}, frame=single, commandchars=\\\{\}, codes={\catcode`$=3\catcode`^=7\catcode`_=8}]
\km{def} foldwithlocal\_min(field, local, initial) \{
  \pr{min}(\pr{foldHood\_min}(field, initial), local)
\}

\km{def} gossip\_min\_temp(initial) \{
  \km{rep} (initial) \{ (x) => foldwithlocal\_min(\km{nbr}\{x\}, \pr{sns\_temp}(), initial) \}
\}

\km{def} foldwithlocal\_max(field, local, initial) \{
  \pr{max}(\pr{foldHood\_max}(field, initial), local)
\}

\km{def} gossip\_max\_threshold(initial) \{
  \km{rep} (initial) \{ (x) => foldwithlocal\_max(\km{nbr}\{x\}, \pr{sns\_threshold}(),initial) \}
\}

gossip\_min\_temp(\pr{infinity}) < gossip\_max\_threshold(\pr{-infinity})
\end{Verbatim}
\FDtask{where \texttt{foldHood\_min} and \texttt{foldHood\_max} can then be substituted with their equivalent versions \texttt{minHood}, \texttt{maxHood}.}

\subsection{Examples: Building Blocks} \label{sec:blocks}

Through functional parametrisation we are now able to express the main ``building blocks'' used in field calculus (as reported in \cite{beal:SFM}), a set of highly general and guaranteed composable operators for the construction of resilient coordination applications. Each of these building blocks captures a family of frequently used strategies for achieving flexible and resilient decentralised behaviour, hiding the complexity of using the low-level constructs of field calculus. Despite their small number, these operators are so general as to cover, individually or in combination, a large number of the common coordination patterns used in design of resilient systems. The three building blocks, whose behaviour will be thoroughly evaluated in next section along that of alternative implementations, are defined as follows:

\subsubsection{Block G}

\texttt{G(source,initial)(metric,accumulate)} is a ``spreading'' operation generalising distance measurement, broadcast, and projection, which takes two fields and two functions as inputs: \texttt{source} (a float indicator field, which is $0$ for sources and $\infty$ for other devices), \texttt{initial} (initial values for the output field), \texttt{metric} (a function providing a map from each neighbour to a distance), and \texttt{accumulate} (a commutative and associative two-input function over values). It may be thought of as executing two tasks: first, computing a field of shortest-path distances from the source region according to the supplied function \texttt{metric}, and second, propagating values up the gradient of the distance field away from source, beginning with value \texttt{initial} and accumulating along the gradient with \texttt{accumulate}. 
%
%
This is accomplished through built-in $\minHoodLoc(\fvalue, \lvalue)$, which selects the minimum \TYtask{of the neighbours'} values in $\fvalue$ and the local value $\lvalue$ (i.e. the minimum in $\applySubstitution{\fvalue}{\envmap{\deviceId}{\lvalue}}$) according to the lexicographical order on pairs.

\begin{Verbatim}[fontsize=\fontsize{7pt}{8pt}, frame=single, commandchars=\\\{\}, codes={\catcode`$=3\catcode`^=7\catcode`_=8}]
\km{def} G(source, initial)(metric, accumulate) \{
  \km{rep} ( \pr{pair}(source, initial) ) \{ (x) $\toSymK{}$
    \pr{minHoodLoc}(\pr{pair}(\km{nbr}\{\pr{1st}(x)\} + metric(), accumulate(\km{nbr}\{\pr{2nd}(x)\})),
      \pr{pair}(source, initial))
  \}
\}
\end{Verbatim}

As an example, \texttt{G\_distanceTo} function (equivalent to the function \texttt{distanceTo}  shown in Section~\ref{sec:syntax}), and a \texttt{G\_broadcast} function to spread values from a source, can be simply implemented with G as:

\begin{Verbatim}[fontsize=\fontsize{7pt}{8pt}, frame=single, commandchars=\\\{\}, codes={\catcode`$=3\catcode`^=7\catcode`_=8}]
\km{def} addRange(x) \{ x + \pr{nbrRange}() \}

\km{def} identity(x) \{ x \}

\km{def} G\_distanceTo(source) \{ \pr{2\TYtask{nd}}( G(source, 0)(\pr{nbrRange}, addRange)) \}

\km{def} G\_broadcast(source, value) \{ \pr{2\TYtask{nd}}( G(source, value)(\pr{nbrRange}, identity)) \}
\end{Verbatim}

\subsubsection{Block C}

\texttt{C(potential,local,null)(accumulate)} is an operation that is complementary to \texttt{G}: it accumulates information down the gradient of a supplied potential field. This operator takes three fields and a function as inputs: \texttt{potential} (a numerical field), \texttt{local} (values to be accumulated), \texttt{null} (an idempotent value for \texttt{accumulate}) and \texttt{accumulate} (a commutative and associative two-input function over values). At each device, the idempotent \texttt{null} is combined with the \texttt{local} value and any values from neighbours with higher values of the \texttt{potential} field, using function \texttt{accumulate} to produce a cumulative value at each device. 
For instance, if \texttt{potential} is a distance gradient computing with \texttt{G} in a given region $R$, \texttt{accumulate} is addition, and \texttt{null} is 0, then \texttt{C} collects the sum of values of \texttt{local} in region $R$.
	
\begin{Verbatim}[fontsize=\fontsize{7pt}{8pt}, frame=single, commandchars=\\\{\}, codes={\catcode`$=3\catcode`^=7\catcode`_=8}]
\km{def} C(potential, local, null)(accumulate) \{
  \km{rep} ( \pr{pair}(local, \pr{uid}()) ) \{ (x) $\toSymK{}$
    \pr{pair}(
      accumulate(
        \pr{mux}(\km{nbr}\{potential\} < potential \pr{&&} \km{nbr}\{\pr{2nd}(x)\} = \pr{uid}(), \km{nbr}\{\pr{1st}(x)\}, null), 
        local
      ), 
      \pr{2nd}(\pr{maxHood+}(\km{nbr}\{\pr{pair}(potential, \pr{uid}())\})) 
    )
  \}
\}
\end{Verbatim}

As an example, a \texttt{C\_sum} function summing all the values of a field down a potential, and a \texttt{C\_any} function checking if any value of a boolean field is true and reporting the result down a potential, can be simply implemented with C as:

\begin{Verbatim}[fontsize=\fontsize{7pt}{8pt}, frame=single, commandchars=\\\{\}, codes={\catcode`$=3\catcode`^=7\catcode`_=8}]
\km{def} sum\_aux(field, local) \{ \pr{sumhood}(field) \pr{+} local \}
\km{def} C\_sum(potential, value) \{ \pr{1st}( C(potential, value, 0)(\FDtask{sum\_aux})) \}

\km{def} or\_aux(field, local) \{ \pr{anyhood}(field) \pr{||} local \}
\km{def} C\_any(potential, value) \{ \pr{1st}( C(potential, value, false)(\FDtask{or\_aux})) \}
\end{Verbatim}

\subsubsection{Block T}

\texttt{T(initial,zero)(decay)} deals with time, whereas \texttt{G} and \texttt{C} deal with space. Since time is one-dimensional, however, there is no distinction between spreading and collecting, and thus only a single operator. This operator takes two fields and a function as inputs: \texttt{initial} (initial values for the resulting field), \texttt{zero} (corresponding final values), and \texttt{decay} (a one-input strictly decreasing function over values). Starting with \texttt{initial} at each node, that value gets decreased by function \texttt{decay} until eventually reaching the \texttt{zero} value, thus implementing a flexible count-down, where the rate of the count-down may change over time. For instance, if \texttt{initial} is a pair of a value $\anyvalue$ and a timeout $t$, \texttt{zero} is a pair of the blank value \texttt{null} and $0$, and \texttt{decay} takes a pair, removing the elapsed time since previous computation from \TYtask{the} second component of the pair and turning the first component to \texttt{null} if the second reached 0, then \texttt{T} implements a limited-time memory of $\anyvalue$.

\begin{Verbatim}[fontsize=\fontsize{7pt}{8pt}, frame=single, commandchars=\\\{\}, codes={\catcode`$=3\catcode`^=7\catcode`_=8}]
\km{def} T(initial, zero)(decay) \{
  \km{rep} ( initial ) \{ (x) $\toSymK{}$
    \pr{min}(\pr{max}(decay(x), zero), initial)
  \}
\}
\end{Verbatim}

As an example, a \texttt{T\_track} function simply tracking an input value over time, and a \texttt{T\_memory} function holding a value for a given amount of time (and the showing a null value), can be simply implemented with T as:

\begin{Verbatim}[fontsize=\fontsize{7pt}{8pt}, frame=single, commandchars=\\\{\}, codes={\catcode`$=3\catcode`^=7\catcode`_=8}]
\km{def} T\_track(value) \{ T(value, value)(identity) \}

\km{def} memory\_evolve(x) \{ 
  \km{if} ( \pr{1st}(x) < \pr{sns\_interval}() ) \{ \pr{pair}(0,null) \} \{ \pr{pair}(\pr{1st}(x)-\pr{sns\_interval}(), \pr{2nd}(x)) \}
\}

\km{def} T\_memory(value, time, null) \{ \pr{2nd}(T(\pr{pair}(time,value), \pr{pair}(0,null))(memory\_evolve)) \}
\end{Verbatim}

\TYtask{with} the built-in operator (sensor) \texttt{sns$\_$interval} \TYtask{returning} the time elapsed since the last execution round.


\section{Self-Stabilisation and Eventual Behaviour}
\label{sec:stabilisation}

In the dynamic environments typically considered by self-organising systems, a key resilience property is {\em self-stabilisation}: the ability of a system to recover from arbitrary changes in state. In particular, of the various notions of self-stabilisation (see the survey in~\cite{S93c}), we use the definition from~\cite{dolev} as further restricted by~\cite{LMCS:selfsabilisation}: a self-stabilising computation is one that, from any initial state, after some period without changes in the computational environment, reaches a single ``correct'' final configuration, intended as the output of computation.

\JBtask{Self-stabilisation (formalised in Section~\ref{ssec:ss_def}) focuses on a computation's eventual behaviour (formalised in Section~\ref{ssec:ss_eventual}), rather than its transient behaviour, which also enables optimisation by substitution of alternate coordination mechanisms (cf.\ Section~\ref{sec:workflow}).
As we will see, this definition covers a broad and useful class of self-organisation mechanisms, though some are excluded, such as continuously changing fields like self-synchronising pulse-coupled oscillators~\cite{MirolloStrogatz90} and computations that converge only in the limit like Laplacian-based approximate consensus~\cite{OlfatiSaber07}.
Incorporating such mechanisms into a framework such as we present here will require bounding the dynamical behaviours of computations (e.g., by identification of an appropriate Lyapunov function~\cite{DasguptaCDC16}).
Preliminary investigations in this area have produced positive results (e.g.,~\cite{DasguptaCDC16,MoECAS17}), but integration with the framework presented in this paper is a major project that remains as future work.}

\subsection{Self-Stabilisation} \label{ssec:ss_def}

\MVtask{Our notion of self-stabilisation considers resilience to changes in the computational system's state or external environment.}
Hence, assume a program $\PROGRAM$ and fixed environmental conditions $\Envi$ (i.e., fixed network topology and inputs of sensors). 
%
%
According to the operational semantics defined in Section \ref{sec:semantics},
\FDtask{for each network configuration $\Cfg$ with environment $\Envi$ that is reachable from the empty network configuration, we can define a transition system $\ap{\System, \nettran{}{\act}{}}$ where:
\begin{itemize}
\item
the only possible action labels $\act$ are device identifiers $\deviceId$ representing firings of an individual device of the network;  and
\item 
 the set of the states $\System$ is  the smallest set of the network configurations such that:
\begin{enumerate}
\item
$\Cfg\in\System$, and
\item
for each $\Cfg' \in \System$ and $\deviceId$ in the network there is an $\Cfg'' \in \System$ such that $\nettran{\Cfg'}{\deviceId}{\Cfg''}$.
\end{enumerate}
\end{itemize}}
We say that a configuration $\Cfg$ is \emph{stable} iff it is not changed by firings, i.e., $\nettran{\Cfg}{\deviceId}{\Cfg}$ for each $\deviceId$.
Let  $\Cfg_0 \nettran{}{\deviceId_0}{} \Cfg_1 \nettran{}{\deviceId_1}{} \ldots$ be an infinite sequence of transitions in $\System$. We say that the sequence is \emph{fair} iff each configuration $\Cfg_t$ is followed by firings of every possible device, i.e., for each $t \geq 0$ and $\deviceId$ there exists a $t' > t$ such that $\deviceId_{t'} = \deviceId$.
We say that the sequence stabilises to state $\Cfg$ iff $\Cfg_i = \Cfg$ for each $i$ after a certain $t \geq 0$.

Given a program $\PROGRAM$ and fixed environmental conditions, a transition system like the one considered above can be defined for any closed expression $\e$ that may call the user-defined functions defined in $\PROGRAM$: just consider $\e$ as the main expression of $\PROGRAM$. In the following, for convenience of the presentation,  we focus on computations associated to such an expression $\e$.

\begin{defn}[Stabilisation and Self-Stabilisation]
 A closed expression $\e$ is: 
\begin{itemize}
	\item
	\emph{stabilising} iff every fair sequence stabilises given fixed environmental conditions $\Envi$;
	\item
	\emph{self-stabilising} to state $\Cfg$ iff every fair sequence stabilises to the same state $\Cfg$ given fixed environmental conditions $\Envi$. 
\end{itemize}
A function $\funvalue(\overline\xname)$ is \emph{self-stabilising} iff given any self-stabilising 
 expressions $\overline\e$ of \FDtask{the type of the inputs of $\funvalue$} the expression $\funvalue(\overline\e)$ is self-stabilising.
\end{defn}

Note that if an expression $\e$ self-stabilises, then it does so to a state that is unequivocally determined by the environmental conditions $\Envi$ (i.e., it does not depend on the initial configuration $\Cfg_0$) and can hence be interpreted as the output of a computation on $\Envi$. Furthermore, this final state $\Cfg$ must be stable.
Note that this definition implies that field computations recover from any change on environmental conditions, since they react to them by forgetting their current state and reaching the stable state implied by such a change.
Complementarily, computation can reach a stable state only when environmental changes are transitory. 

\subsection{Eventual Behaviour} \label{ssec:ss_eventual}

Define a \emph{computational field} $\dvalue$ as a map $\envmap{\overline\deviceId}{\overline\anyvalue}$,\footnote{Even though the definition resembles that of a \emph{neighbouring field value}, it differs both in purpose and in content, since $\anyvalue$ is allowed to be a neighbouring field value itself, and $\overline\deviceId$ spans the whole network and not just a device's neighbourhood.} such that if $\overline\anyvalue$ have field type their domains are \emph{coherent} with the environment $\ap{\Topo,\Sens}$, that is, $\domof{\dvalue(\deviceId)} = \Topo(\deviceId) \cap \domof{\dvalue}$.
Let $\denotval{\type}$ be the set of values of type $\type$ and $\denottype{\type} = \deviceIdSet \pto \denotval{\type}$ be\footnote{By $A \pto B$ we denote the set of all partial functions from $A$ to $B$.} the set of all computational fields $\dvalue$ of the same type. Each such $\dvalue$ is computable by at least one self-stabilising expression $\e$  (defined by cases, and executed in the restricted environment corresponding to $\domof{\dvalue}$)---in which case we say that  $\e$ is a \emph{self-stabilising expression for} $\dvalue$.

Note that a network status $\Field$ induces uniquely a computational field $\dvalue$ defined by $\dvalue(\deviceId) = \vrootOf{\Field(\deviceId)(\deviceId)}$, while conversely each $\dvalue$ is coherent with multiple network statuses $\Field$. Thus a computational field is not sufficient to capture the whole status of a computation of a program $\PROGRAM$. However, for self-stabilising programs $\PROGRAM$ and self-stabilising functions $\funvalue$, it suffices to define the \emph{eventual output} of a computation: given computational fields $\overline\dvalue$, let $\Cfg_0 \nettran{}{\deviceId_0}{} \Cfg_1 \nettran{}{\deviceId_1}{} \ldots$ be any fair evolution of a network computing $\funvalue(\overline\e)$ where $\overline\e$ are self-stabilising expressions for $\overline\dvalue$. Since $\funvalue$ is self-stabilising, the fair evolution eventually stabilises to a uniquely determined state $\Cfg = \SystS{\Envi}{\Field}$, independently from the chosen evolution and initial state. This final status field $\Field$ in turn determines a unique computational field $\dvalue$, which we can think of as the eventual output of the computation\footnote{\MVtask{Note this eventual state is reached independently of the fair sequence of firing that occurs; hence, it would be the same also with firings following fully-synchronous concurrency models like BSP \cite{valiant1990bsp}.}}.

\begin{defn}[Eventual behaviour]
	Let $\e$ be a self-stabilising closed expression.
	We write $\builtindenot{}{\e}$ for the computational field $\dvalue$ eventually produced by the computation of $\e$.

	Let $\funvalue$ be a self-stabilising function of type $\overline\type \to \type'$, where $\overline\type = \type_1\times\cdots\times\type_n$ $(n\ge 0)$. We write $\builtindenot{}{\funvalue}$ for the mathematical function in $(\denottype{\type_1}\times\cdots\times\denottype{\type_n}) \to \denottype{\type'}$,\footnote{Here we assume that all input computational fields share the same domain, which is to be intended as the domain of the overall computation.} such that $\builtindenot{}{\funvalue}(\overline\dvalue) = \builtindenot{}{\funvalue(\overline\e)}$ where $\overline\e$ are self-stabilising expressions for $\overline\dvalue$.
\end{defn}

Using the above definition of eventual behaviour, it is possible to precisely state and investigate the equivalence of different self-stabilising programs.

\begin{prop}[Eventual behaviour preserving equivalences]\label{prop:ebpe}
\begin{enumerate}
	\item
	Let $\e_1$, $\e_2$ be self-stabilising expressions with the same eventual behaviour. Then given a self-stabilising expression $\e$, swapping $\e_1$ for $\e_2$ in $\e$ does not change the eventual outcome of its computation.
	\item
	Let $\funvalue_1$, $\funvalue_2$ be self-stabilising functions with the same eventual behaviour. Then given a self-stabilising expression $\e$, swapping $\funvalue_1$ for $\funvalue_2$ in $\e$ does not change the eventual outcome of its computation.
	\item \label{prop:inlining}
	Let $\e$ be a self-stabilising expression calling a user-defined self-stabilising function $\fname$ such that in $\body{\funvalue}$ no $\xname\in\args{\funvalue}$ occurs in the branch of an $\ifK$. Let $\e'$ be the expression obtained from $\e$ by substituting each function application of the kind $\funvalue(\overline\e)$ with $\applySubstitution{\body{\funvalue}\,}{\substitution{\args{\funvalue}}{\overline\e}}$. Then $\e'$ is self-stabilising and has the same eventual behaviour as $\e$ (i.e. $\builtindenot{}{\e} = \builtindenot{}{\e'}$).
\end{enumerate}
\end{prop}
\begin{proof}
\begin{enumerate}
	\item
	By straightforward induction on the structure of an expression. The base case is given by expressions without occurrences of $\e_1$ and $\e_2$, and by expressions $\e_i$ for $i=1,2$. The inductive step follows by compositionality of the operational semantics.
	
	\item
	For the same reasoning as in point (1), where the base case is given by expressions without occurrences of $\funvalue_1$ and $\funvalue_2$ and by expressions $\funvalue_i(\overline\e)$ for $i=1,2$.
	
	\item
	Recall that no expressions with side effects are contemplated in the present calculus. Since no argument of $\funvalue$ occurs in the branch of an $\ifK$, each of those arguments is evaluated in the same environment as the whole function application $\funvalue(\overline\e)$. It follows that $\e_1 = \funvalue(\overline\e)$ and $\e_2 = \applySubstitution{\body{\funvalue}\,}{\substitution{\args{\funvalue}}{\overline\e}}$ have the same behaviour (hence the same eventual behaviour). The thesis follows then by applying point (1) to expressions $\e_1$ and $\e_2$.
\end{enumerate}
\end{proof}

\section{Self-Stabilising Fragment}
\label{sec:fragment}

By exploiting the definition of self-stabilisation given in previous section, and its implication in considering eventual behaviour as a valid characterisation of the functional property of a field computation, it is possible to identify sufficient conditions for self-stabilisation in terms of a fragment of the Field Calculus, inductively defined by:

\begin{enumerate}
\item
identifying a ``base'' fragment of the Field Calculus that contains only self-stabilising programs; 
\item
identifying a set of eventual behaviour preserving equivalences (cf.\ Proposition~\ref{prop:ebpe});
\item
relying on the fact that the least fragment of the Field Calculus  that contains  the fragment (1) and is closed under the equivalences (2) is self-stabilising.
\end{enumerate}

Accordingly, in this section we first present \FDtask{some motivating examples of non self-stabilising Field Calculus programs (in Section~\ref{s:ss-non-ss-examples}), then present} the syntax of the identified ``base'' self-stabilising fragment (in Section \ref{s:ss-syntax}), then state the self-stabilisation result for the fragment along with equivalence results further extending the fragment (in Section \ref{sec:fragment_is_ss}), and finally discuss its expressiveness 
(in Section \ref{ss-expr}).

\subsection{Examples of non self-stabilising programs}\label{s:ss-non-ss-examples}

Let us begin by considering some examples of Field Calculus programs that are not self-stabilising, illustrating key classes of program behavior that need to be excluded from our self-stabilising fragment.

\begin{example} \label{ex:nonss1}
	Consider the following function:
\begin{Verbatim}[fontsize=\fontsize{7pt}{9pt}, frame=single, commandchars=\\\{\}, codes={\catcode`$=3\catcode`^=7\catcode`_=8}]
\km{def} f1(v) \{ \km{rep} (v) \{ (x) => v-x \} \}
\end{Verbatim}
	This function does not self-stabilise, since given a fixed input $\anyvalue$ its output loops through a series of different values. For example, if $\anyvalue$ is constantly equal to $1$ the outputs are $0, 1, 0, 1, \ldots$ Thus in this case self-stabilisation is prevented by an \emph{oscillating} behaviour.
\end{example}

\begin{example} \label{ex:nonss2}
	Consider the following function (a classical gossip implementation):
\begin{Verbatim}[fontsize=\fontsize{7pt}{9pt}, frame=single, commandchars=\\\{\}, codes={\catcode`$=3\catcode`^=7\catcode`_=8}]
\km{def} f2(v) \{ \km{rep} (v) \{ (x) => \pr{max}(\pr{maxHood+}(\km{nbr}\{x\}), v) \} \}
\end{Verbatim}
	This function does not self-stabilise, since its output depends on the whole history of values $\anyvalue$ given to it in the network. For example, if at some point a highest value $k$ was given in some device, the eventual output of the function upon a fixed input $\anyvalue < k$ is $k$, thus it is not a function of the constant input $\anyvalue$. Thus, in this case self-stabilisation is prevented by an indefinite ``state preservation''.
\end{example}

\begin{example} \label{ex:nonss3}
	Consider the following function, with input $\anyvalue$ of an unbounded integer type (big integer):
\begin{Verbatim}[fontsize=\fontsize{7pt}{9pt}, frame=single, commandchars=\\\{\}, codes={\catcode`$=3\catcode`^=7\catcode`_=8}]
\km{def} f3(v) \{ \km{rep} (v) \{ (x) => \pr{min}(\pr{minHood}(\km{nbr}\{x\}) - 1, v) \} \}
\end{Verbatim}
	This function does not self-stabilise, since given any fixed input $\anyvalue$ and at least one neighbour, its output keeps decreasing without a bound. Thus, in this case self-stabilisation is prevented by a \emph{divergent} behaviour.
\end{example}

\MVtask{Oscillation, state preservation, and divergence are three key causes of non self-stabilisation that the proposed fragment will address.}

\subsection{Syntax}\label{s:ss-syntax}

The ``base'' self-stabilising fragment of field calculus is obtained by replacing each occurrence of the expression token $\e$ in the first two lines of Figure \ref{fig:syntax} (i.e., in the productions for $\PROGRAM$ and $\FUNCTION$) with the self-stabilising expression token $\s$, defined in Figure~\ref{fig:fragment}. This fragment includes:
\begin{itemize}
	\item all expressions not containing a $\repK$ construct, hence comprising built-in functions, which are therefore assumed to be self-stabilising;
	\item three special forms of $\repK$-constructs, defined with a specific syntax coupled with semantic restrictions on relevant functional parameters.
\end{itemize}

\begin{figure}[!t]
\centering
\centerline{\framebox[\linewidth]{$
	\begin{array}{lcl@{\hspace{-5pt}}r}
		\s & \BNFcce &  \xname \BNFmid \anyvalue  \BNFmid \letK \; \xname = \s \; \inK \; \s \BNFmid \funvalue(\overline\s) \BNFmid \ifK (\s) \{ \s \} \{ \s \} \BNFmid \nbrK\{\s\}
		&   {\footnotesize \mbox{self-stabilising}} \\[3pt]
		&&  \; \BNFmid \;  \repK(\e)\{ (\xname) \toSymK{} \funvalue^\mathsf{C}(\nbrK\{\xname\}, \nbrK\{\s\}, \overline\e) \}  &  {\footnotesize \mbox{expression}}\\[3pt]
		&&  \; \BNFmid \;  \repK(\e)\{ (\xname) \toSymK{} \funvalue(\muxK(\nbrlt(\s), \nbrK\{\xname\}, \s), \overline\s) \}\\[3pt]
		&&  \; \BNFmid \;\repK(\e)\{ (\xname) \toSymK{} \funvalue^\mathsf{R}(\minHoodLoc(\funvalue^\mathsf{MP}(\nbrK\{\xname\}, \overline\s), \s), \xname, \overline\e) \}
	\end{array}
	$}
}
\caption{Syntax of a self-stabilising fragment of field calculus expressions, where self-stabilising expressions $\s$ occurring inside a $\repK$ statement cannot contain free occurrences of the $\repK$-bound variable $\xname$.}
\label{fig:fragment}
\end{figure}

\subsubsection{\FDtask{The $\mathsf{C}, \mathsf{M}, \mathsf{P}, \mathsf{R}$ function properties}}

The properties that these functional parameters are required to satisfy are among the following, visually annotated in the figure through superscripts on function names. Notice that properties $\mathsf{M}$, $\mathsf{P}$, and $\mathsf{R}$ require some of their argument types to be equipped with a \emph{partial order} relation, while property $\mathsf{C}$ requires its argument types to be equipped with a \emph{metric}. In order to obtain the self-stabilisation property for the fragment, we shall also need some further assumptions, discussed later in the description of each pattern.


 \paragraph{$\mathsf{C}$ (Converging)}
A function $\funvalue(\fvalue, \fvaluealt, \overline\anyvalue)$ is \MVtask{said} converging iff, for every device $\deviceId$, its return value is closer to $\fvaluealt(\deviceId)$ than the maximal distance of $\fvalue$ to $\fvaluealt$.
\FDtask{To be precise, given any environment $\Trees$, device $\deviceId \in \domof{\Trees}$, values $\fvalue, \fvaluealt, \overline\anyvalue$ coherent with the domain of $\Trees$, and assuming that $\bsopsem{\deviceId}{\Trees}{\funvalue(\fvalue, \fvaluealt, \overline\anyvalue)}{\mkvt{\lvalue}{\overline \vtree}}$:
	\[
	\dist\cp{\lvalue, \fvaluealt(\deviceId)} = 0 \text{ or }
	\dist\cp{\lvalue, \fvaluealt(\deviceId)} < \max\bp{\dist(\fvalue(\deviceId'), \fvaluealt(\deviceId')): ~ \deviceId' \in \domof{\Trees}}
	\]
	where $\dist$ is any metric.}

\FDtask{
\begin{example} \label{ex:propC}
Function $\funvalue_1(\fvalue, \fvaluealt) = \texttt{pickHood}(\fvaluealt - \fvalue) = (\fvaluealt - \fvalue)(\deviceId)$ is not converging, for example when $\fvalue, \fvaluealt$ are constant fields equal to $2, 3$ respectively so that $\lvalue = 1$ (\texttt{pickHood} selects the value on the current device from a field). 
On the other hand, functions $\funvalue_2(\fvalue, \fvaluealt) = \texttt{pickHood}((\fvaluealt + \fvalue)/2)$ and $\funvalue_3(\fvalue, \fvaluealt) = \texttt{pickHood}(\fvaluealt) + \texttt{meanHood}(\fvalue - \fvaluealt)/2$ are converging.
\end{example}
}
	
 \paragraph{$\mathsf{M}$ (Monotonic non-decreasing)}
A stateless\footnote{A function $\funvalue(\overline\xname)$ is \emph{stateless} iff given fixed inputs $\overline\anyvalue$ always produces the same output, independently from the environment or specific firing event. In other words, its behaviour corresponds to that of a mathematical function.} function $\funvalue(\xname, \overline\xname)$ with arguments of local type is monotonic non-decreasing in its first argument iff whenever $\lvalue_1 \leq \lvalue_2$, also $\funvalue(\lvalue_1, \overline\lvalue) \leq \funvalue(\lvalue_2, \overline\lvalue)$.

\FDtask{
\begin{example} \label{ex:propM}
Function $\funvalue_1(\lvalue) = \lvalue-1$ is monotonic non-decreasing, while function $\funvalue_2(\lvalue) = \lvalue^2$ is not.
\end{example}
}
	
 \paragraph{$\mathsf{P}$ (Progressive)}
A stateless function $\funvalue(\xname, \overline\xname)$ with local arguments is progressive in its first argument iff $\funvalue(\lvalue, \overline\lvalue) > \lvalue$ or $\funvalue(\lvalue, \overline\lvalue) = \top$ (where $\top$ denotes the unique maximal element of the relevant type).

\FDtask{
\begin{example} \label{ex:propP}
Function $\funvalue_1(\lvalue) = \lvalue+1$ is progressive, while functions $\funvalue_2(\lvalue) = \lvalue-1$, $\funvalue_3(\lvalue) = \lvalue ^2$ are not.
\end{example}
}
	
 \paragraph{$\mathsf{R}$ (Raising)}
A function $\funvalue(\lvalue_1, \lvalue_2, \overline\anyvalue)$ is raising with respect to partial orders $<$, $\vartriangleleft$ iff:
	\begin{itemize}
		\item $\funvalue(\lvalue, \lvalue, \overline\anyvalue) = \lvalue$;
		\item $\funvalue(\lvalue_1, \lvalue_2, \overline\anyvalue) \geq \min(\lvalue_1, \lvalue_2)$;
		\item either $\funvalue(\lvalue_1, \lvalue_2, \overline\anyvalue) \vartriangleright \lvalue_2$ or $\funvalue(\lvalue_1, \lvalue_2, \overline\anyvalue) = \lvalue_1$.
	\end{itemize}

\FDtask{
\begin{example} \label{ex:propR}
	Function $\funvalue_1(\lvalue_1, \lvalue_2) = \lvalue_1$ is raising with respect to any partial orders. Function $\funvalue_2(\lvalue_1, \lvalue_2) = \lvalue_1 - \lvalue_2$ is not raising since it violates both the first two clauses. Function $\funvalue_3(\lvalue_1, \lvalue_2) = (\lvalue_1+\lvalue_2)/2$ respects the first two clauses for $\vartriangleleft = <$, but it violates the last one whenever $\lvalue_2 > \lvalue_1$.
\end{example}
}


\subsubsection{\FDtask{The  three $\repK$ patterns}}

We are now able to analyse the three $\repK$ patterns.


\paragraph{Converging $\repK$}
In this pattern, variable $\xname$ is repeatedly updated through function $\funvalue^\mathsf{C}$ and a self-stabilising value $\s$. The function $\funvalue^\mathsf{C}$ may also have additional (not necessarily self-stabilising) inputs $\overline\e$. If the range of the metric granting convergence is a well-founded set\footnote{An ordered set is \emph{well-founded} iff it does not contain any infinite descending chain.} of real numbers, the pattern self-stabilises since it gradually approaches the value given by $\s$.

\begin{example}
Function \texttt{f1} in Example \ref{ex:nonss1} does not respect the converging $\repK$ pattern, as shown in Example \ref{ex:propC}. However, if we change \texttt{f1} to the following and assume that its input and output are finite-precision numeric values (e.g., Java's \texttt{double}):
\begin{Verbatim}[fontsize=\fontsize{7pt}{9pt}, frame=single, commandchars=\\\{\}, codes={\catcode`$=3\catcode`^=7\catcode`_=8}]
\km{def} filter(v) \{ \km{rep} (v) \{ (x) => (v+x)/2 \} \}
\end{Verbatim}
we obtain a \emph{low-pass filter} that is self-stabilising and complies with the converging $\repK$ pattern.
\end{example}

\paragraph{Acyclic $\repK$}
In this pattern, \TYtask{the} neighbourhood's values \TYtask{for} $\xname$ are first filtered through a self-stabilising partially ordered ``potential'', keeping only values held in devices with lower potential (thus in particular discarding the device's own value of $\xname$). This is accomplished by the built-in function $\nbrlt$, which returns a field of booleans selecting the neighbours with lower argument values, and could be defined as $\defK \; \nbrlt(\xname) \, \{ \nbrK\{\xname\} < \xname\}$.
	
	The filtered values are then combined by a function $\funvalue$ (possibly together with other values obtained from self-stabilising expressions) to form the new value for $\xname$. No semantic restrictions are posed in this pattern, and intuitively it self-stabilises since there are no cyclic dependencies between devices.

\begin{example}
Function \texttt{f2} in Example \ref{ex:nonss2} does not respect the acyclic $\repK$ pattern, since it aggregates all neighbours without any ``acyclic filtering''. However, if we change \texttt{f2} to the following:
\begin{Verbatim}[fontsize=\fontsize{7pt}{9pt}, frame=single, commandchars=\\\{\}, codes={\catcode`$=3\catcode`^=7\catcode`_=8}]
\km{def} f2C(v, p) \{ \km{rep} (v) \{ (x) => \pr{max}(\pr{maxHood+}(\pr{mux}(\pr{nbrlt}(p), \km{nbr}\{x\}, 0), v) \} \}
\end{Verbatim}
we obtain a particular usage of the \emph{C} block, which is self-stabilising and complies with the acyclic $\repK$ pattern.
\end{example}

\paragraph{Minimising $\repK$}	
In this pattern, \TYtask{the} neighbourhood's values \TYtask{for} $\xname$ are first increased by a monotonic progressive function $\funvalue^\mathsf{MP}$ (possibly depending also on other self-stabilising inputs). As specified above, $\funvalue^\mathsf{MP}$ needs to operate on local values: in this pattern it is therefore implicitly promoted to operate (pointwise) on fields.
	
	Afterwards, the minimum among those values and a local self-stabilising value is then selected by function $\minHoodLoc(\fvalue, \lvalue)$ (which selects the ``minimum'' in $\applySubstitution{\fvalue}{\envmap{\deviceId}{\lvalue}}$). In order to be able to define such a minimum, we thus require the partial order $\leq$ to constitute a \emph{lower semilattice}.\footnote{A \emph{lower semilattice} is a partial order such that \FDtask{greatest lower bounds are defined for any finite set of values in the partial order. In the examples used in this paper we shall treat \emph{greatest lower bounds} as \emph{minima}, since the only examples of such partial orders we consider are in fact total orders.}}
	
	 Finally, this minimum is fed to the \emph{raising} function $\funvalue^\mathsf{R}$ together with the old value for $\xname$ (and possibly any other inputs $\overline\e$), obtaining a result \TYtask{that} is higher than at least one of the two parameters. We assume that the second partial order $\vartriangleleft$ is \emph{noetherian},\footnote{A partial order is \emph{noetherian} iff it does not contain any infinite ascending chains.} so that the raising function is required to eventually conform to the given minimum.
	
	Intuitively, this pattern self-stabilises since it computes the minimum among the local values $\lvalue$ after being increased by $\funvalue^\mathsf{MP}$ along every possible path (and the effect of the raising function can be proved to be negligible).

\begin{example}
Function \texttt{f3} in Example \ref{ex:nonss3} does not respect the minimising $\repK$ pattern, since its internal function is monotonic (see Example \ref{ex:propM}) but not progressive (see Example \ref{ex:propP}). However, if we change \texttt{f3} to the following:
\begin{Verbatim}[fontsize=\fontsize{7pt}{9pt}, frame=single, commandchars=\\\{\}, codes={\catcode`$=3\catcode`^=7\catcode`_=8}]
\km{def} hopcount(v) \{ \km{rep} (v) \{ (x) => \pr{min}(\pr{minHood}(\km{nbr}\{x\}) + 1, v) \} \}
\end{Verbatim}
we obtain a \emph{hop-count gradient}, a particular case of the \emph{G} block which is self-stabilising and complies with the minimising $\repK$ pattern.
\end{example}

Note that the well-foundedness and noetherianity properties are trivially verified whenever the underlying data set is finite.

\subsection{Self-Stabilisation and Equivalence} \label{sec:fragment_is_ss}

Under reasonable conditions, we are able to prove that the proposed fragment is indeed self-stabilising. The proofs of all the results in this section are given in Appendix \ref{sec:proofs}, while here we only report the full statements.

\begin{thm}[Fragment Stabilisation] \label{thm:stabilisation}
	Let $\s$ be a closed expression in the self-stabilising fragment, and assume that every built-in operator is self-stabilising.\footnote{Most built-in operators are stateless, thus trivially self-stabilising in one round.} Then $\s$ is self-stabilising.
\end{thm}

Since the fragment is closed under function application, the result is immediately extended to whole programs.

In Section \ref{ssec:ss_eventual} we introduced a notion of \emph{equivalence} for self-stabilising programs. 
Therefore, although the $\repK$ patterns are defined through \emph{functions} with certain properties, we are allowed to inline them (which is a transformation preserving self-stabilisation, as shown in Proposition \ref{prop:ebpe}).
Moreover, \TYtask{a} few noteworthy equivalence properties hold for the given patterns, as shown by the following theorem.
\begin{thm}[Substitutability] \label{thm:substitutability}
	The following three equivalences hold:
	\begin{itemize}
		\item Each $\repK$ in a self-stabilising fragment self-stabilises to the same value 
under arbitrary substitution of the initial condition.
		\item The \emph{converging $\repK$} pattern
  self-stabilises to the same value \TYtask{as}
the single expression $\s$ occurring in it.
		\item The \emph{minimising $\repK$} pattern 
  self-stabilises to the same value \TYtask{as}
 the analogous pattern where $\funvalue^\mathsf{R}$ is the identity on its first argument.
	\end{itemize}
\end{thm}

In other words, the function $\funvalue^\mathsf{R}$ does not influence the eventual behaviour of a function, and can instead be used to fine-tune the transient behaviour of an algorithm. The same holds for the initial conditions of all patterns and function $\funvalue^\mathsf{C}$ in the converging $\repK$ pattern (which in fact is only meant to fine-tune the transient behaviour of the given expression $\s$). No relevant equivalences can be stated for the acyclic $\repK$ pattern, since it is parametrised by a single aggregating function which in general heavily influences the final outcome of the computation.

\subsection{Expressiveness}\label{ss-expr}

\subsubsection{\FDtask{Programs captured by the fragment}}\label{sec:newGCT}

Even though at a first glance the fragment could seem rather specific, it encompasses (equivalent versions of) many relevant algorithms. In particular, all of the three building blocks introduced in Section \ref{sec:blocks} are easily shown to belong to the fragment. This effectively constitutes a new and simpler proof of self-stabilisation for them.
	
	Operator $G$ is the following instance of the minimising $\repK$ pattern:
\begin{Verbatim}[fontsize=\fontsize{7pt}{8pt}, frame=single, commandchars=\\\{\}, codes={\catcode`$=3\catcode`^=7\catcode`_=8}]
\km{def} fr(new, old) \{ new \}

\km{def} fmp(field, dist)(accumulate) \{
  \pr{pair}(\pr{1st}(field) + dist, accumulate(\pr{2nd}(field)))
\}

\km{def} G(source, initial)(metric, accumulate) \{
  \km{rep}(\pr{pair}(source, initial))\{ (x) =>
    fr(\pr{minHoodLoc}(fmp(\km{nbr}\{x\}, metric())(accumulate), \pr{pair}(source, initial)), x)
  \}
\}
\end{Verbatim}
	Function \texttt{fr} is trivially raising (with respect to any pair of partial orders), and function \texttt{fmp} is monotonic progressive provided that pairs are ordered lexicographically (since $\texttt{dist}$ is a positive field).

	Operator $C$ is the following instance of the acyclic $\repK$ pattern:
\begin{Verbatim}[fontsize=\fontsize{7pt}{8pt}, frame=single, commandchars=\\\{\}, codes={\catcode`$=3\catcode`^=7\catcode`_=8}]
\km{def} f(field, local, null, potential)(accumulate) \{
  \pr{pair}(accumulate(\pr{mux}(\pr{2nd}(field) = \pr{uid}(), \pr{1st}(field), null), local),
    \pr{2nd}(\pr{maxHood+}(\km{nbr}\{\pr{pair}(potential, \pr{uid}())\})) )
\}

\km{def} C(potential, local, null)(accumulate) \{
  \km{rep}(\pr{pair}(local, \pr{uid}()))\{ (x) =>
    f(\pr{mux}(\pr{nbrlt}(potential), \km{nbr}\{x\}, null), local, null, potential)(accumulate)
  \}
\}
\end{Verbatim}

	Operator $T$ is the following instance of the converging $\repK$ pattern:
\begin{Verbatim}[fontsize=\fontsize{7pt}{8pt}, frame=single, commandchars=\\\{\}, codes={\catcode`$=3\catcode`^=7\catcode`_=8}]
\km{def} fc(cur, lim, initial)(decay) \{
  \pr{min}(\pr{max}(decay(\pr{pickHood}(cur)), \pr{pickHood}(lim)), initial)
\}

\km{def} T(initial, zero)(decay) \{
  \km{rep}(initial)\{ (x) => fc(\km{nbr}\{x\}, \km{nbr}\{zero\}, initial)(decay) \}
\}
\end{Verbatim}
	Function \texttt{fc} is converging since $\texttt{decay}(\texttt{pickHood}(\texttt{cur}))$ is granted to be closer to \texttt{zero} than its argument, hence:
	\[
	\vp{\texttt{fc}(\fvalue, \nbrK\{\texttt{zero}\}, \anyvalue) - \texttt{zero}} < \vp{\fvalue(\deviceId) - \texttt{zero}} \leq \max(\vp{\fvalue - \nbrK\{\texttt{zero}\}})
	\]

Furthermore, the present fragment strictly includes the one defined in \cite{VBDP-SASO2015}. Both fragments include all expressions without the $\repK$ construct. The first and third $\repK$ pattern in \cite{VBDP-SASO2015} are special cases of \emph{converging $\repK$} (the first converges to the value $\anyvalue_0$ in the \emph{bounded} condition and the third to the value $\lvalue$ in the \emph{double bounded} condition). The second pattern is almost exactly equivalent to the \emph{acyclic $\repK$}.

In the following Section \ref{sec:alternatives} we shall show further examples of algorithms still belonging to the fragment, which are alternative implementations of G, C and T.

\subsubsection{Programs not captured by the fragment}

Unfortunately, many self-stabilising programs are not captured by the fragment. In most cases this is due to syntactical reasons, so that the critical program $P$ can in fact be rewritten into an equivalent program $P'$, which instead belongs to the fragment. An example of this issue is given by the three building blocks $G$, $C$ and $T$, which we needed to rewrite in order to make them fit inside the self-stabilising fragment (see Section~\ref{sec:newGCT}).

Furthermore, self-stabilising programs exist which cannot be rewritten to fit inside the fragment. As an example, one such program could be obtained by the \emph{replicated gossip} \cite{viroli:replicatedgossip} algorithm, which does not fit inside the fragment. In particular, replicated gossip is ``self-stabilising'' \emph{provided} that a certain parameter $p$ (refresh period) is set to a large enough value with respect to certain network characteristics---and as such, it would require a slight modification of our definition of self-stabilisation as well.


\section{Alternative Building Blocks}
\label{sec:alternatives}

\MVtask{Even though} the G, C, and T building blocks define a useful and versatile base of operators, in practice better performing alternatives are often preferred in some specific conditions (see for example the work \TYtask{in} \citeN{VBDP-SASO2015}). 
We can \TYtask{also} use the fragment itself to get inspiration for new alternatives or interesting variations of existing ones.
\JBtask{Importantly, the self-stabilisation framework allows alternatives to be assessed on empirical grounds even when the dynamics of their operation are imperfectly understood, allowing engineering decisions to be made even when analytical solutions are not available.}

In the exploration to follow, we compare the performance of each operator and an alternative via simulation.
We evaluate each proposed alternative by simulating a network of 100 devices placed uniformly randomly in a $200m \times 20m$ rectangular arena, with a $30m$ communication radius.
The dynamics of self-stabilisation are studied by introducing perturbations in ``space'' or ``time''.
In the space perturbation experiments, devices run asynchronously at 1 Hz frequency, moving at 1 m/s in a direction randomly chosen at every round.
We shall conside\TYtask{r} ``small spatial perturbation'' where this is the entirety of the perturbation, and ``large spatial perturbation'' where the source for the spreading / aggregation of the information also switches from the original device to an alternate device every 200 seconds.
On the other hand, in the ``time perturbation'', devices remain still, but their operating frequency is randomly chosen between 0.9 Hz and 1.1 Hz (small perturbation) or 0.5 Hz and 2 Hz (large perturbation).
We performed 200 simulations per configuration, letting both the control and alternate building blocks run at the same time.
Experiments are performed using the Alchemist simulator \cite{PianiniJOS2013}.\footnote{For the sake of reproducibility, the actual experiments are made available at \url{https://bitbucket.org/danysk/experiment-2017-tomacs}}
	
\subsection{Alternative G}

The G operator can be understood as the computation of a distance measure w.r.t. a given metric, while also propagating values according to an accumulating function. However, naive computation of distance suffers from the \emph{rising value problem}: the rising rate of distance values is bounded by the shortest distance in the network, possibly enforcing a very slow convergence rate. Some algorithms avoiding this problem have been developed, such as the CRF-gradient algorithm \cite{crf}. It is possible to rewrite a CRF-gradient distance calculuation to fit the present fragment, as in the following (adapted from the code implemented in the Protelis library \cite{Protelis15}):
\begin{Verbatim}[fontsize=\fontsize{7pt}{9pt}, frame=single, commandchars=\\\{\}, codes={\catcode`$=3\catcode`^=7\catcode`_=8}]
\km{def} raise(new, old, speed, dist) \{
  \km{let} constraint = \pr{minHood}(\km{nbr}\{\pr{1st}(old)\} + dist + (\pr{nbrLag}()+\pr{sns\_interval}())*\pr{2nd}(old)) \km{in}
  \km{if} (new = old \pr{||} \pr{1st}(new) = 0 \pr{||} constraint <= \pr{1st}(old)) \{
    new
  \} \{
    \pr{pair}(\pr{1st}(old)+speed, speed/\pr{sns\_interval}())
  \}
\}

\km{def} combine(x, dist) \{
  \pr{pair}(\pr{1st}(x) + dist, 0)
\}

\km{def} CRF(source, speed)(metric) \{
  \km{rep} ( \pr{pair}(source, 0) ) \{ (x) $\toSymK{}$
    raise(\pr{minHoodLoc}(combine(\km{nbr}\{x\}, metric()), \pr{pair}(source, 0)), x, speed)
  \}
\}
\end{Verbatim}
where \texttt{nbrLag} returns a field of communication lags from neighbours.

It is easy to see that \texttt{raise} is raising with respect to the two identical partial orders $\leq$, $\leq$ (the output either increases the \texttt{old} value \TYtask{or} conforms to the \texttt{new} value). Notice that this rewriting effectively constitutes an alternative proof of self-stabilisation for the algorithm.

If it is acceptable to \TYtask{lose} some degree of accuracy, another possibility for avoiding the rising value problem is to introduce a \emph{distortion} into the metric. 
This is the approach chosen by the Flex-Gradient algorithm \cite{Beal:FLEX} (which we will abbreviate FLEX).
This algorithm allows for a better response to transitory changes while reducing the amount of communication needed between devices. In this case also, we can equivalently rewrite the algorithm in order to make it fit into the self-stabilising fragment.
\begin{Verbatim}[fontsize=\fontsize{7pt}{9pt}, frame=single, commandchars=\\\{\}, codes={\catcode`$=3\catcode`^=7\catcode`_=8}]
\km{def} raise(new, old, dist, eps, freq, rad) \{
  \km{let} slopeinfo = \pr{maxHood}(\pr{triple}((\pr{1st}(old) - \km{nbr}\{\pr{1st}(old)\})/dist, \km{nbr}\{\pr{1st}(old)\}, dist)) \km{in}
  \km{if} (new = old \pr{||} \pr{1st}(new) = 0 \pr{||} \pr{2nd}(old) = freq \pr{||} \pr{1st}(old) > \pr{max}(2*\pr{1st}(new), rad)) \{
    new
  \} \{
    \km{if} (\pr{1st}(slopeinfo) > 1+eps) \{
      \pr{pair}(\pr{2nd}(slopeinfo) + (1+eps)*\pr{3rd}(slopeinfo), \pr{2nd}(old)+1)
    \} \{
      \km{if} (\pr{1st}(slopeinfo) < 1-eps)) \{
        \pr{pair}(\pr{2nd}(slopeinfo) + (1-eps)*\pr{3rd}(slopeinfo), \pr{2nd}(old)+1)
      \} \{
        \pr{pair}(\pr{1st}(old), \pr{2nd}(old)+1)
  \} \} \}
\}

\km{def} combine(x, dist) \{
  \pr{pair}(\pr{1st}(x) + dist, 0)
\}

\km{def} FLEX(source, epsilon, frequency, distortion, radius)(metric) \{
  \km{rep} ( \pr{pair}(source, 0) ) \{ (x) $\toSymK{}$
    \km{let} dist = \pr{max}(metric(), distortion*radius) \km{in}
    raise(\pr{minHoodLoc}(combine(\km{nbr}\{x\}, dist), \pr{pair}(source, 0)),
          x, dist, epsilon, frequency, radius)
  \}
\}
\end{Verbatim}

In this case, \texttt{raise} is raising with respect to the two partial orders $\leq_1$ (ordering w.r.t. the first component of the pair) and $\leq_2$ (ordering w.r.t. the second component).

We evaluate these new building blocks when applied to distance estimation, using the two following variations of \texttt{G\_distance} (parameter \texttt{r} in the body of \texttt{G'\_flex\_distance} stands for the communication radius of devices):

\begin{Verbatim}[fontsize=\fontsize{7pt}{8pt}, frame=single, commandchars=\\\{\}, codes={\catcode`$=3\catcode`^=7\catcode`_=8}]
\km{def} G'\_crf\_distance(source) \{  CRF(source, 1/12)(\pr{nbrRange}) \}

\km{def} G'\_flex\_distance(source) \{  FLEX(source, 0.3, 10, 0.2, r)(\pr{nbrRange}) \}
\end{Verbatim}

\Cref{fig:g} shows the evaluation of G and its proposed replacements: FLEX has a good performance all-around, while CRF suffers poor performance with small spatial disruptions and G suffers poor performance with large spatial disruptions.

\begin{figure}[tb]
        \includegraphics[width=.5\textwidth]{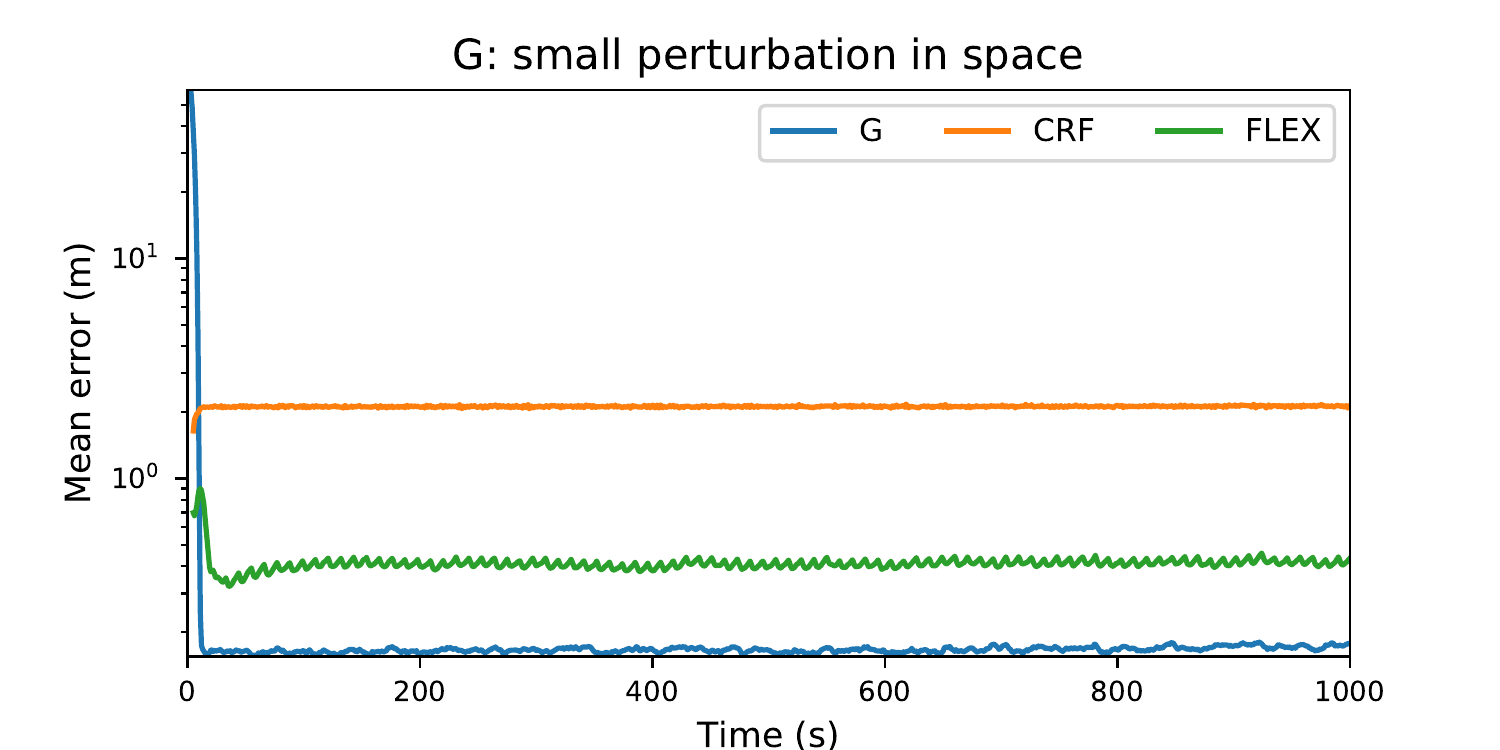}
        \includegraphics[width=.5\textwidth]{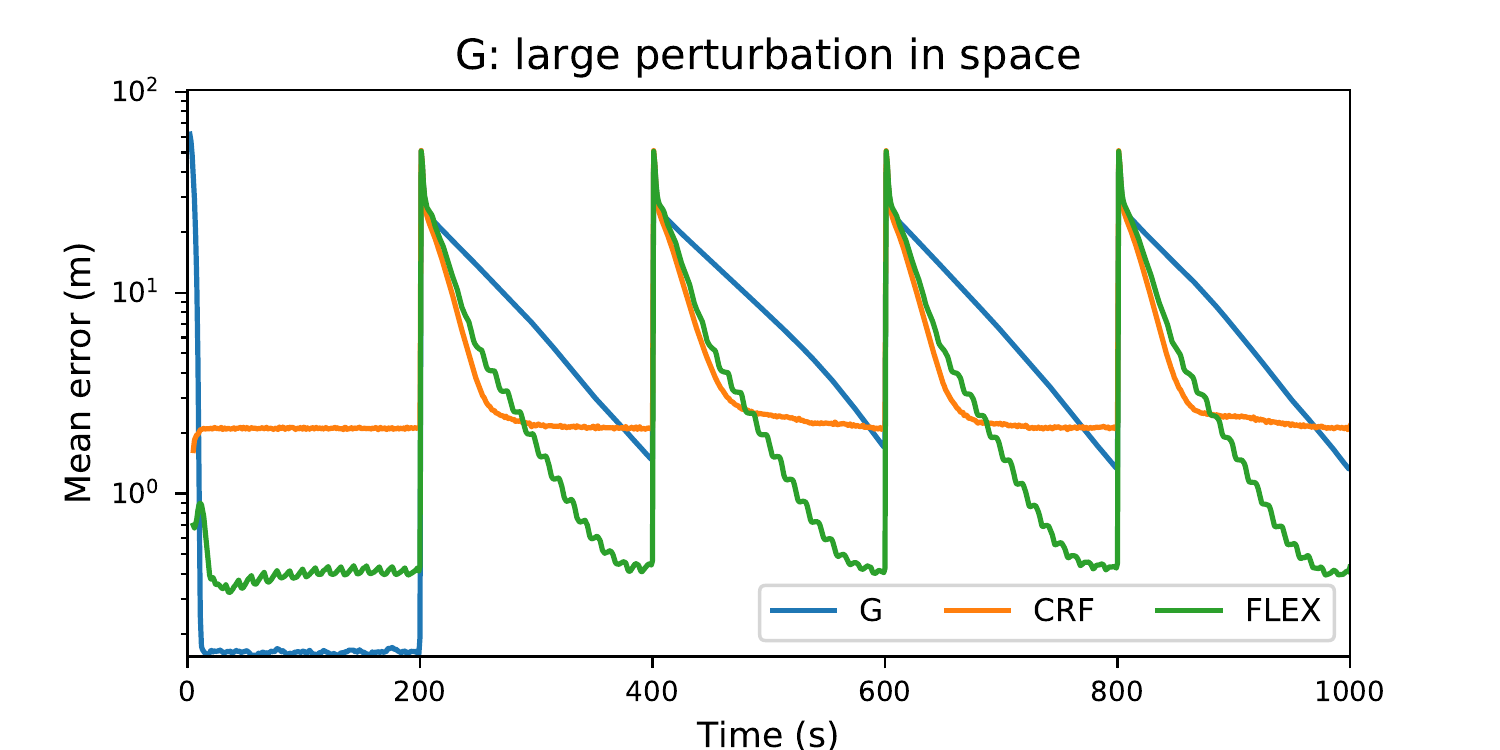}
        \includegraphics[width=.5\textwidth]{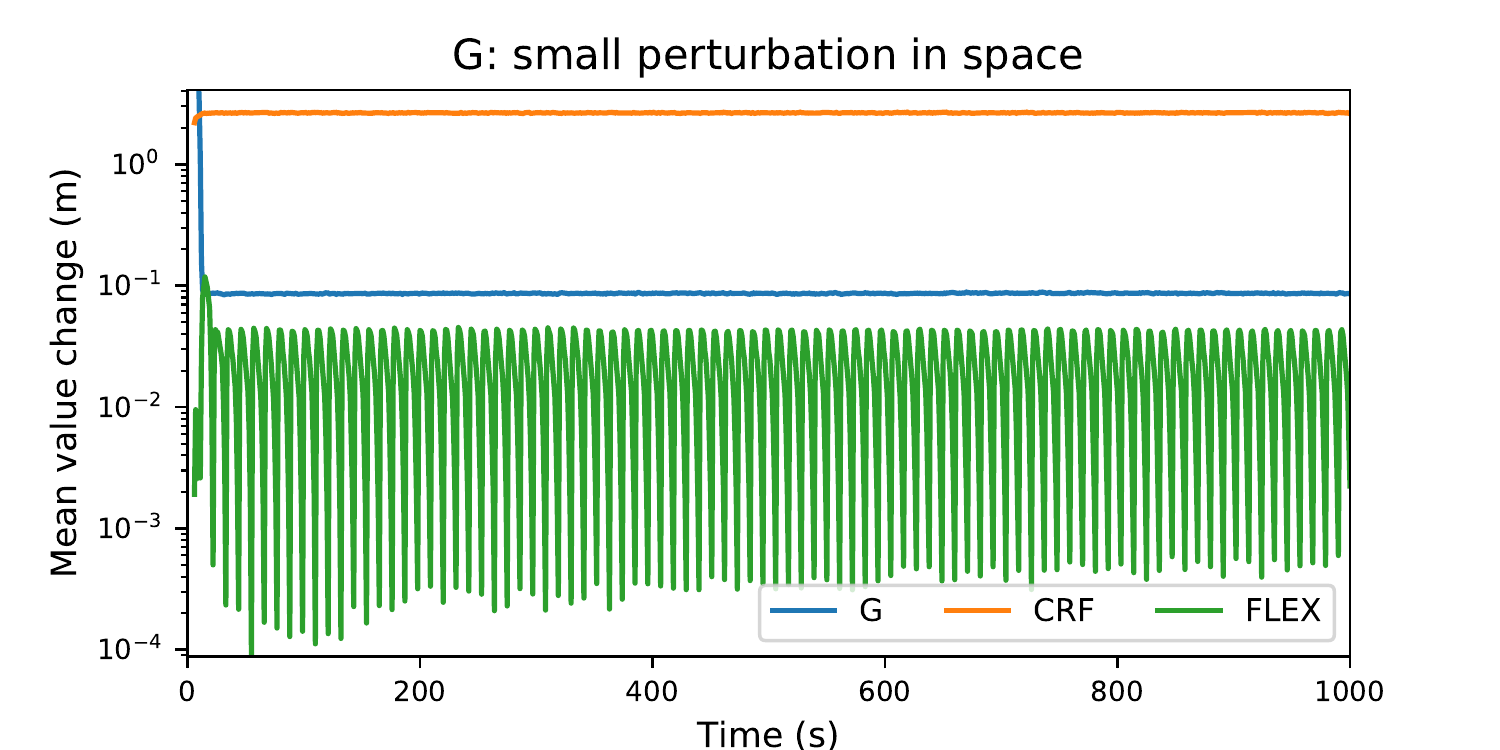}
        \includegraphics[width=.5\textwidth]{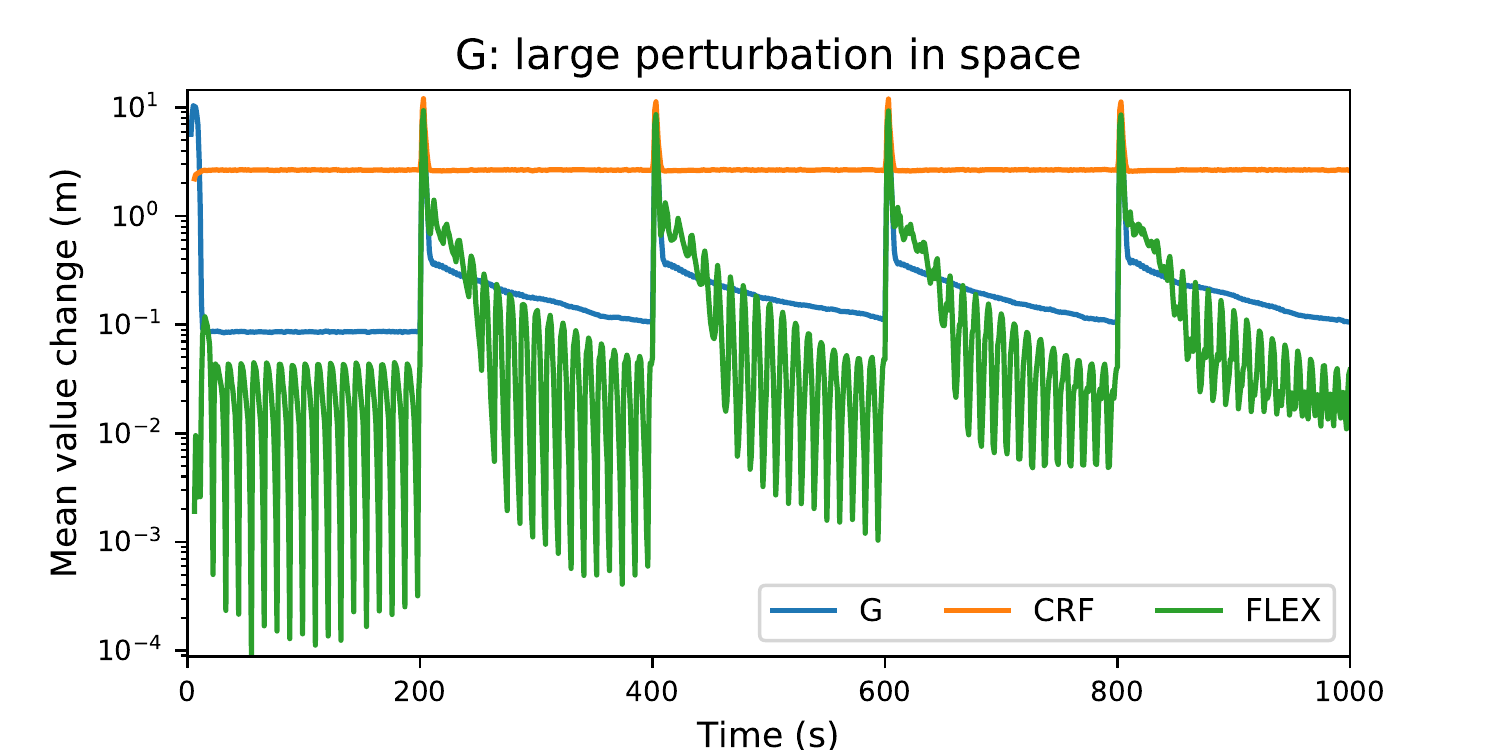}
        \includegraphics[width=.5\textwidth]{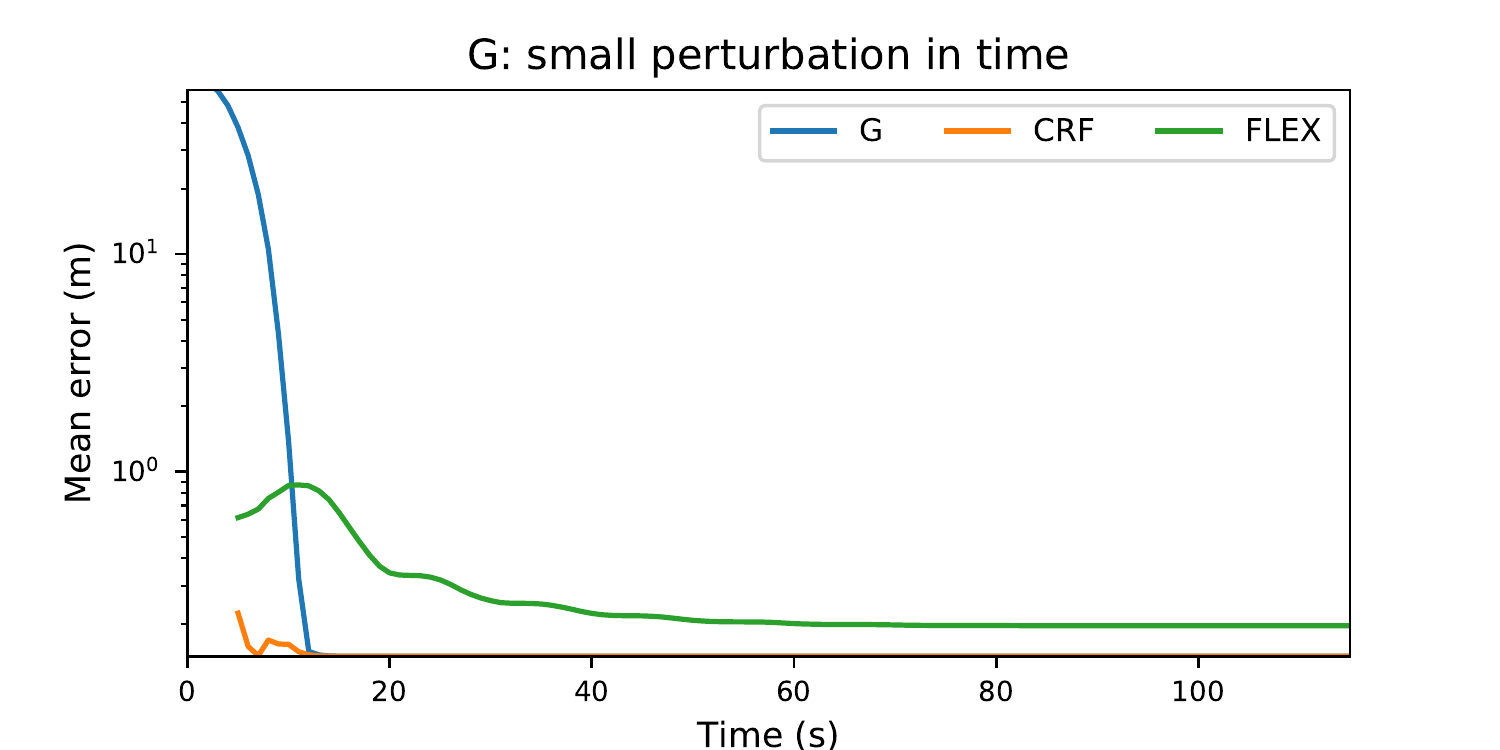}
        \includegraphics[width=.5\textwidth]{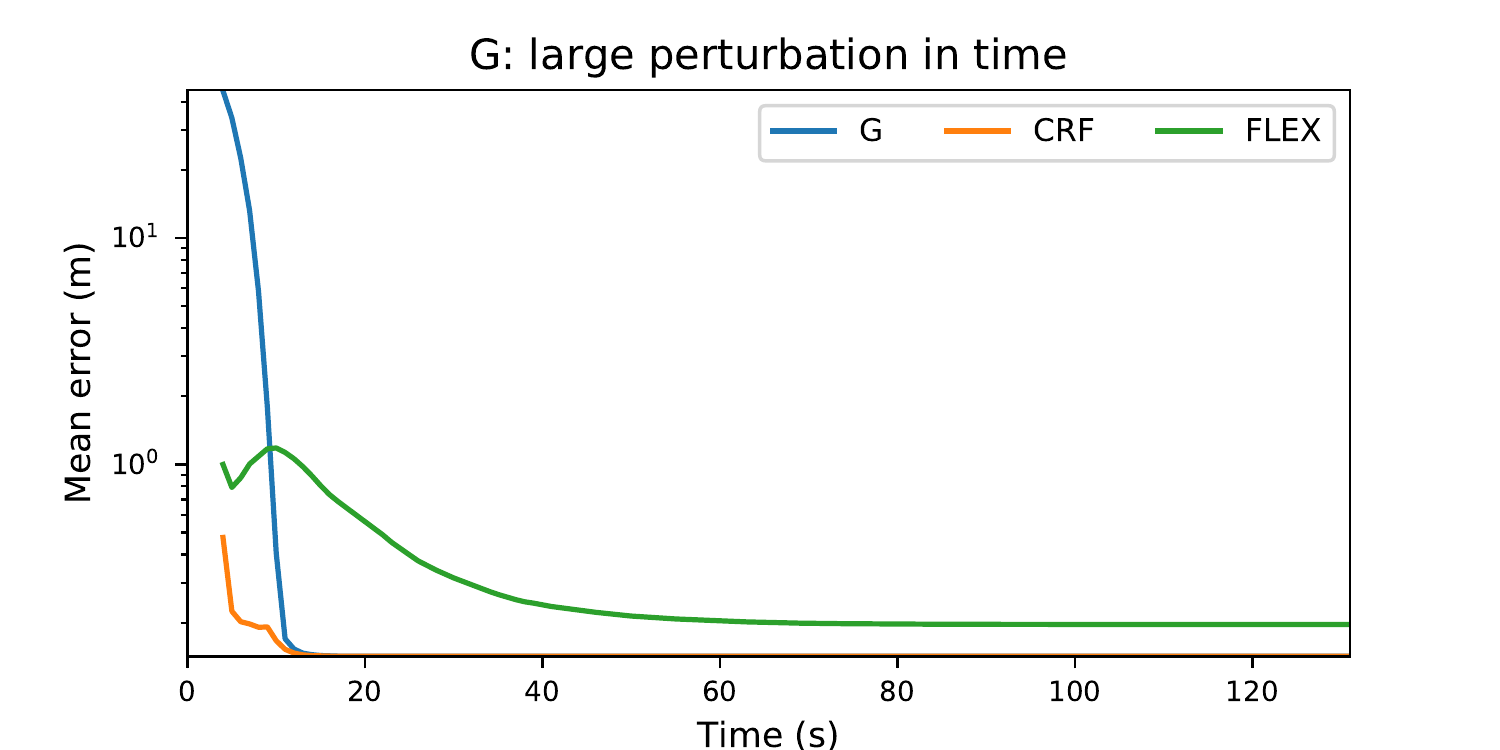}
        \caption{Evaluation of G building blocks: plain G (blue), CRF (green) and FLEX (red).
        We measure the average error across all devices (first and last row) and the stability of the value, namely, the average value change between subsequent rounds (middle row).
        With small spatial perturbations, G provides the lowest average error, while FLEX provides the highest local value stability. With large spatial changes, CRF is the quickest to adapt, but stabilises with a higher error than FLEX.
        The classic G suffers from the \TYtask{rising} value problem.
        All the algorithms stabilise in time with little sensitivity to device asynchrony.}
        \label{fig:g}
\end{figure}

\subsection{Alternative C}

The C operator aggregates a computational field of \texttt{local} values with the function \texttt{accumulate} towards the device with highest potential, each device feeding its value to the neighbour with highest potential. This process, however, is fragile since the ``neighbour with highest potential'' changes often and abruptly over time. In order to overcome this shortcomings, it is sometimes possible to use a \emph{multipath C}.

Assume that the aggregating operator defines an abelian monoid\footnote{A structure $\ap{X, \circ}$ is an abelian monoid if $\circ$ is an associative and commutative operator with identity.} on its domain. Assume in addition that each $\lvalue$ in the domain has \TYtask{an} $n$-th root $\lvalue_n$, that is, an element which aggregated with itself $n$ times produces the original value $\lvalue$. Then the value computed by a device can be ``\TYtask{split}'' and sent to \emph{every} neighbour device with higher potential than the current device, by taking its $n$-th root where $n$ is the number of devices with higher potential.

\begin{Verbatim}[fontsize=\fontsize{7pt}{8pt}, frame=single, commandchars=\\\{\}, codes={\catcode`$=3\catcode`^=7\catcode`_=8}]
\km{def} extract(val, num)(root) \{
  \pr{pair}(val, root(val, num))
\}

\km{def} aggregate(field, local, potential)(accumulate, root) \{
  extract( accumulate(\pr{foldHood}(\pr{2nd}(field), accumulate), local),
    \pr{counthood}(\km{nbr}\{potential\} > potential) )(root)
\}

\km{def} C'(potential, local, null)(accumulate, root) \{
  \km{rep} ( \pr{pair}(local,local) ) \{ (x) $\toSymK{}$
    aggregate(\pr{mux}(\km{nbr}\{potential\} < potential, \km{nbr}\{x\}, null),
      local, potential)(accumulate, root)
  \}
\}
\end{Verbatim}

We evaluate the multi-path alternative of C when used to sum values of a field, using the following \TYtask{variation} of \texttt{C\_sum}\footnote{Operator \texttt{/} is used as root for C' since a value gets equally divided by $n$ and spread in the $n$ neighbour nodes ascending potential.}:

\begin{Verbatim}[fontsize=\fontsize{7pt}{8pt}, frame=single, commandchars=\\\{\}, codes={\catcode`$=3\catcode`^=7\catcode`_=8}]
\km{def} C'\_sum(potential, field) \{  \pr{1st}( C'(potential, value, 0)(\pr{+}, \pr{/})) \}
\end{Verbatim}

\TYtask{Specifically}, we compare \texttt{C\_sum} and \texttt{C'\_sum} used to aggregate the summation of ``1'' along the gradient of a distance estimate produced by the FLEX algorithm.
As a consequence, we expect to get the count of devices participating to the system in the source of the distance estimate.
Since the source switches in case of large perturbation, the counting device switches as well.
\Cref{fig:c} shows the evaluation of C and its proposed replacement: the multi-path version performs better with small spatial changes, but may return higher errors during transients that require a whole network reconfiguration.

\begin{figure}[tb]
        \includegraphics[width=.5\textwidth]{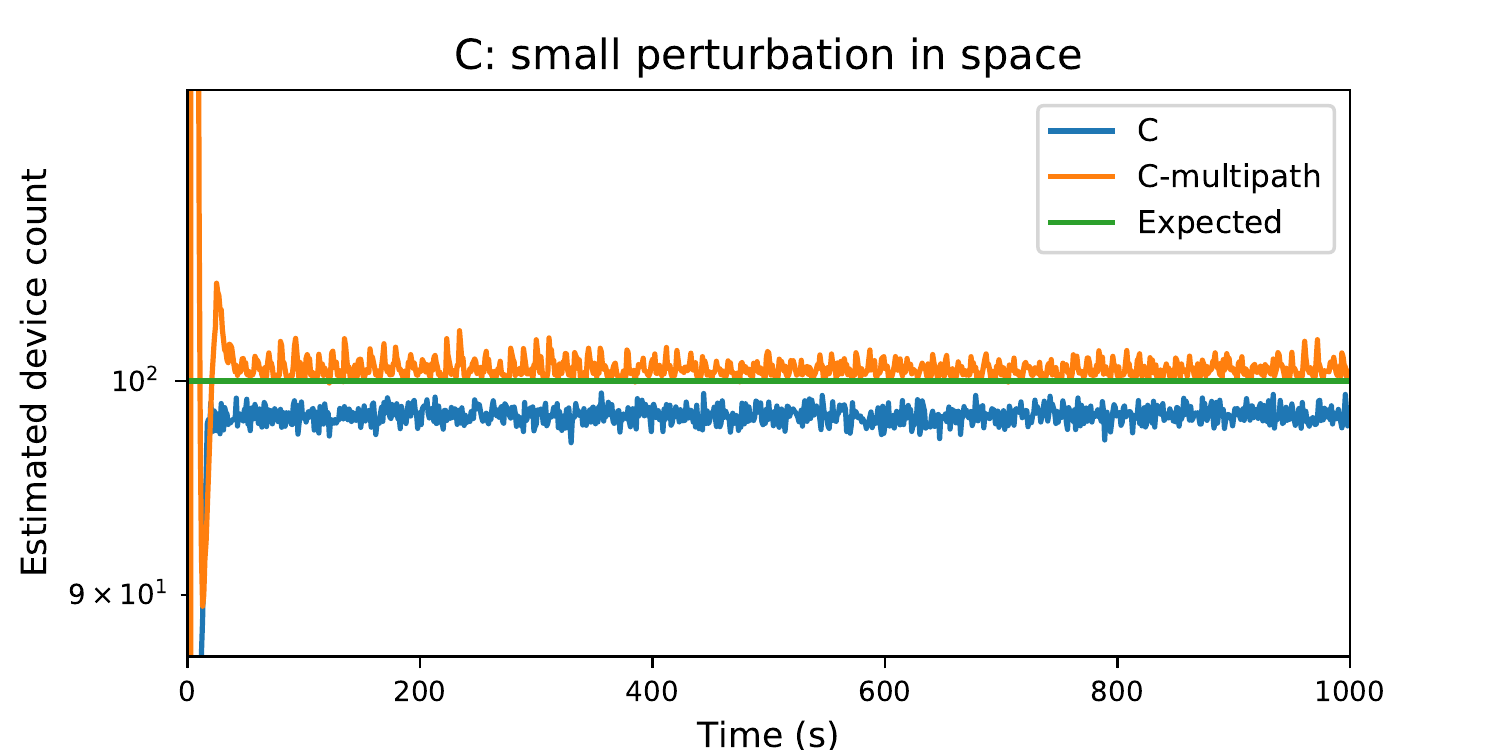}
        \includegraphics[width=.5\textwidth]{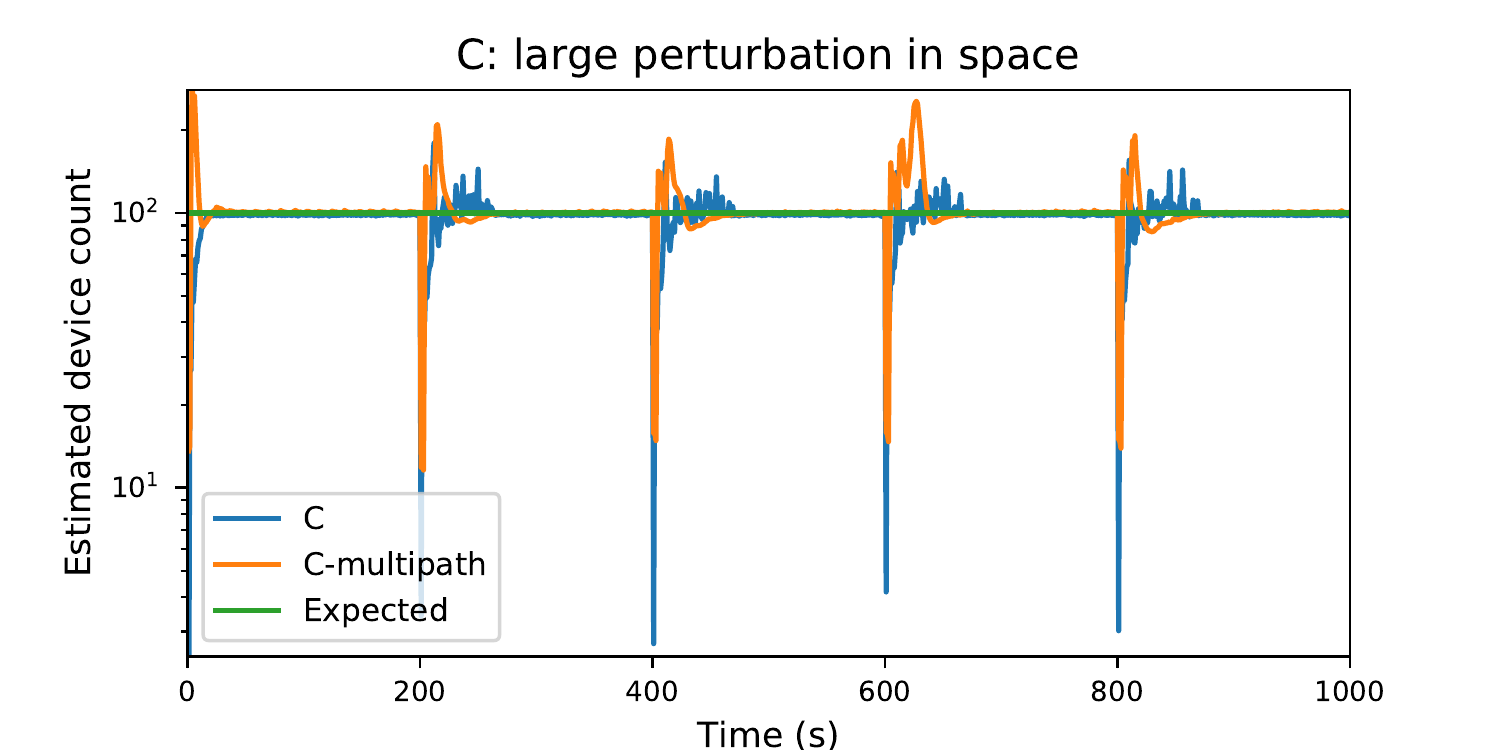}
        \includegraphics[width=.5\textwidth]{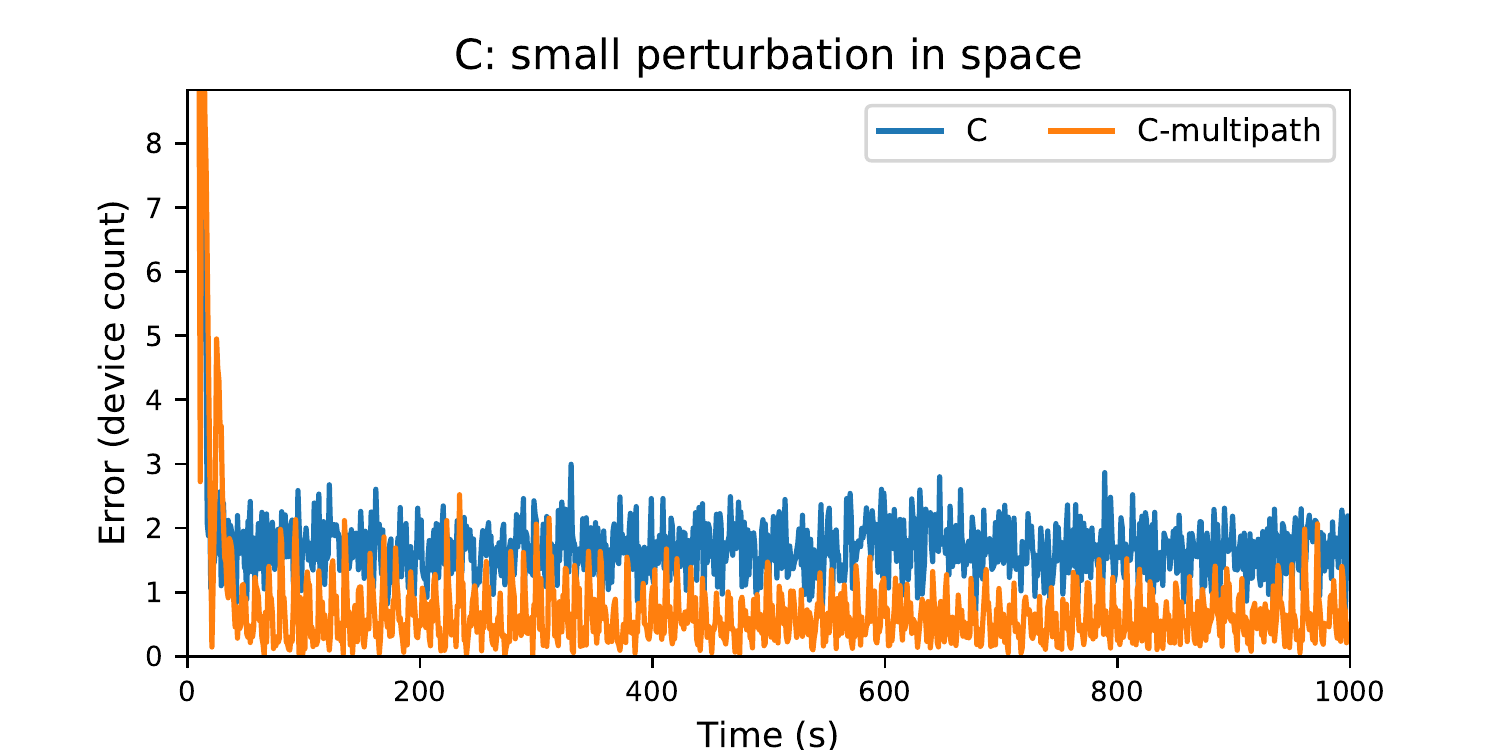}
        \includegraphics[width=.5\textwidth]{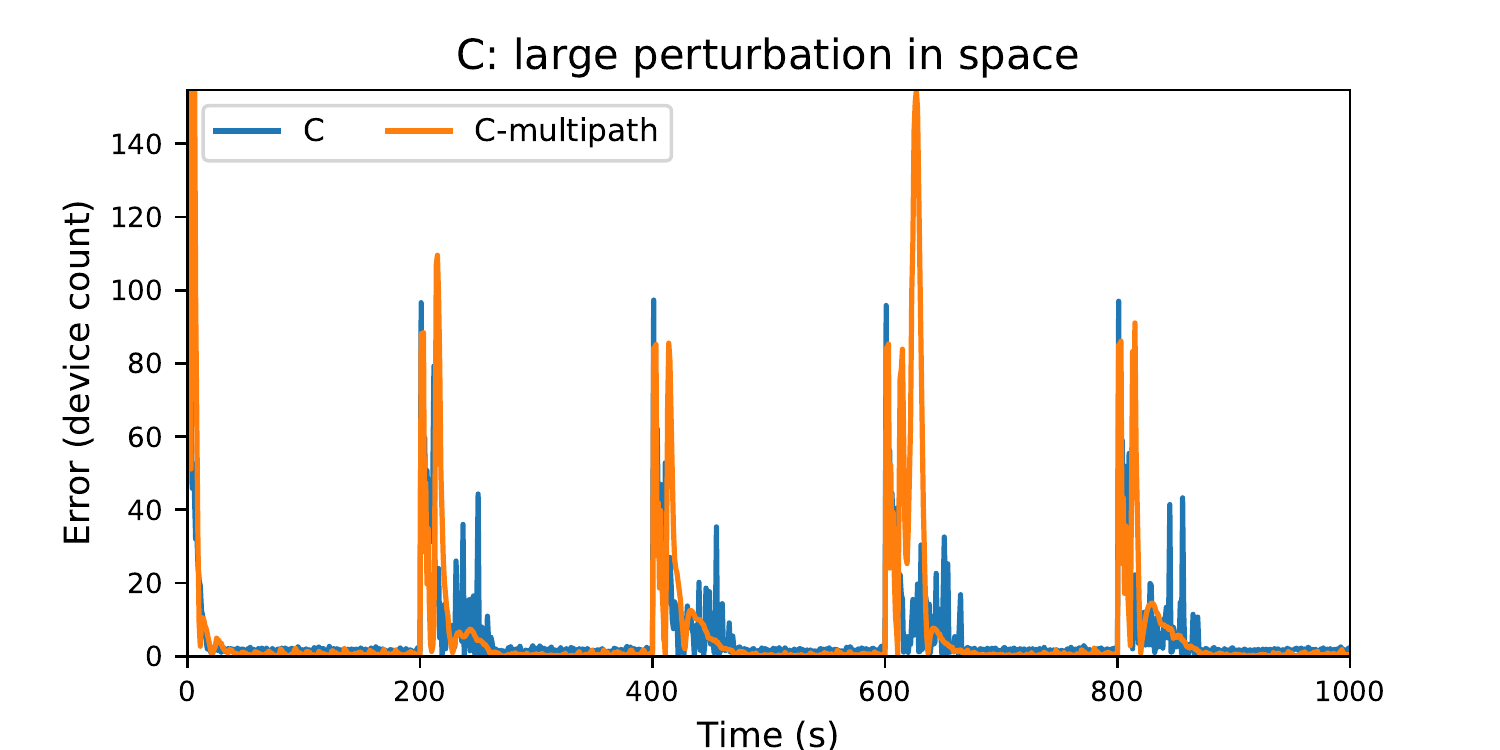}
        \includegraphics[width=.5\textwidth]{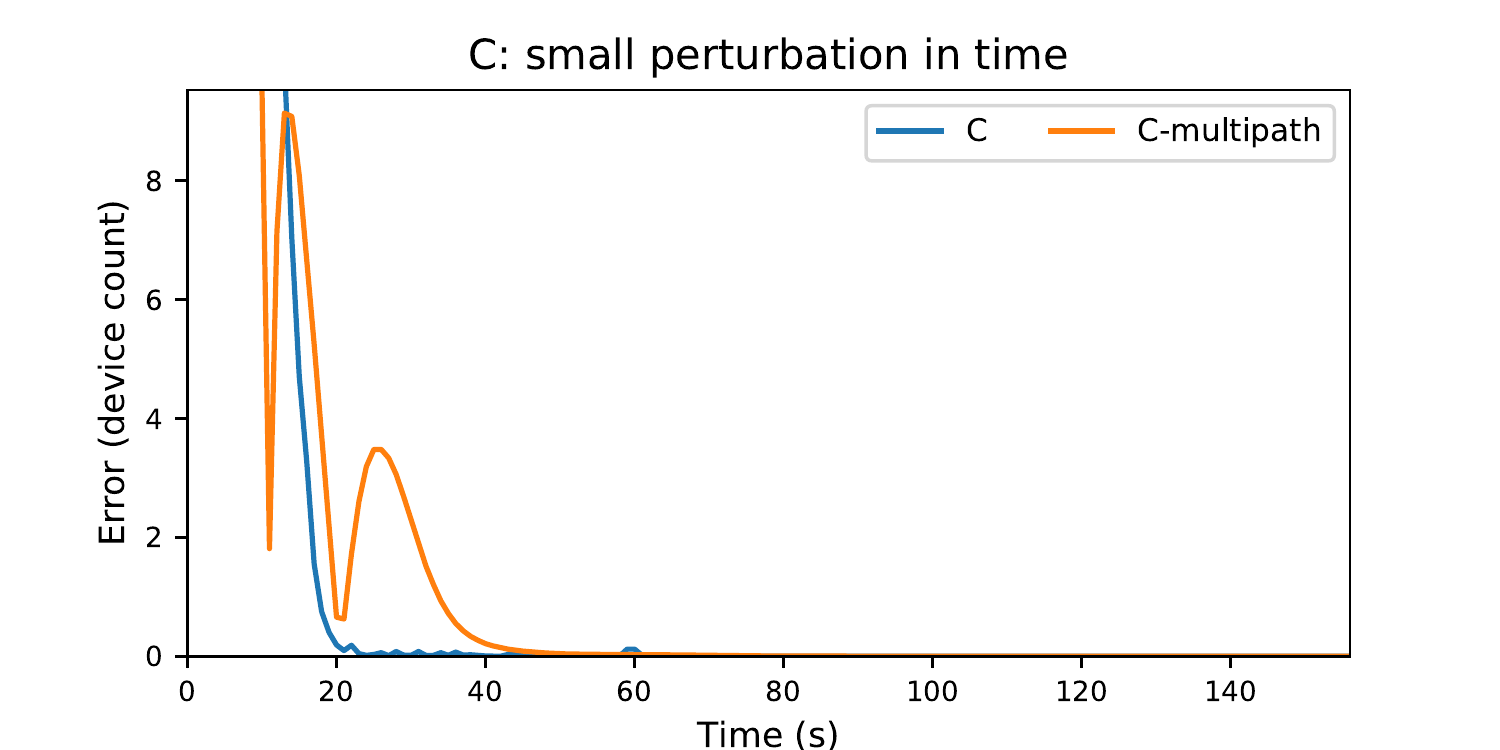}
        \includegraphics[width=.5\textwidth]{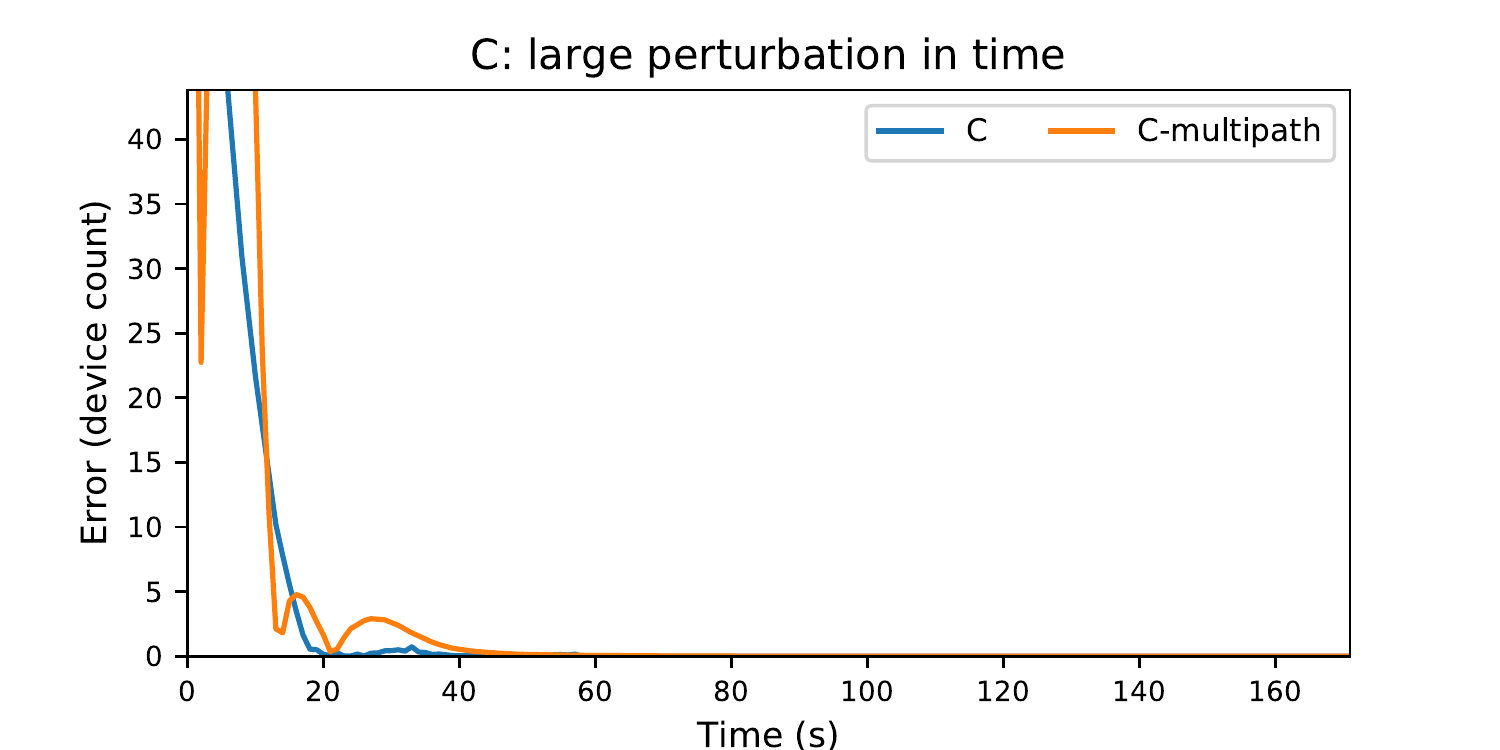}
        \caption{Evaluation of C building blocks: classic C (blue), and multi-path alternative (green).
        Expected values are depicted in red.
        We measure the aggregated value in the source node (first row) and the error (last two rows).
        With small spatial perturbations, the multipath alternative outperforms the spanning-tree-building default implementation; however, it may provide worse estimations at the beginning of transients that require a large reconfiguration.
        Both algorithms stabilise regardless of devices' asynchrony.}
        \label{fig:c}
\end{figure}

\subsection{Alternative T}

Both the T operator and the whole \emph{converging $\repK$} pattern are meant to smooth out the outcome of another computation, which at the limit is returned unaltered. However, it is sometimes useful to introduce a \emph{spatial} coordination among different devices, in order to smooth out the converging process also spatially. This can be accomplished by the following alternative building block, which \emph{decays} towards a \emph{value} with a speed obtained by \emph{averaging} on how close each neighbour is to its goal value.

\begin{Verbatim}[fontsize=\fontsize{7pt}{8pt}, frame=single, commandchars=\\\{\}, codes={\catcode`$=3\catcode`^=7\catcode`_=8}]
\km{def} follow(cur, lim)(average, decay) \{
  \pr{pickHood}(lim) + decay(average(cur - lim))
\}

\km{def} T'(initial, value)(average, decay) \{
  \km{rep} ( initial ) \{ (x) $\toSymK{}$
    follow(\km{nbr}\{x\}, \km{nbr}\{value\})(average, decay)
  \}
\}
\end{Verbatim}

We evaluate the use of T' in tracking a noisy signal, using the following variation of \texttt{T\_track}

\begin{Verbatim}[fontsize=\fontsize{7pt}{8pt}, frame=single, commandchars=\\\{\}, codes={\catcode`$=3\catcode`^=7\catcode`_=8}]
\km{def} T'\_track(value) \{  T'(value,value)(\pr{\TYtask{meanHood}}, x => a*x)\}
\end{Verbatim}
where \texttt{\TYtask{meanHood}} computes the mean value of the provided field, and \texttt{a} is the smoothing parameter.
In the comparison of \texttt{T\_track} and \texttt{T'\_track}, every device perceives the original signal (either a sine or a square wave) summed with a locally generated noise in $[-1, 1]^{10}$ (\texttt{s}).
In particular, \texttt{T'\_track} provides a sort of spatial low-pass filter, that trades a delay in tracking the signal \TYtask{for a smoother} response.
\Cref{fig:t} aggregates the results. T' takes advantage of the spatial smoothing, and performs better overall in case of noisy input.
This comes\TYtask{, however,} at the price of lower reactivity to changes, which becomes evident with large enough values of the smoothing parameter.

\begin{figure}[!tb]
		\includegraphics[width=.5\textwidth]{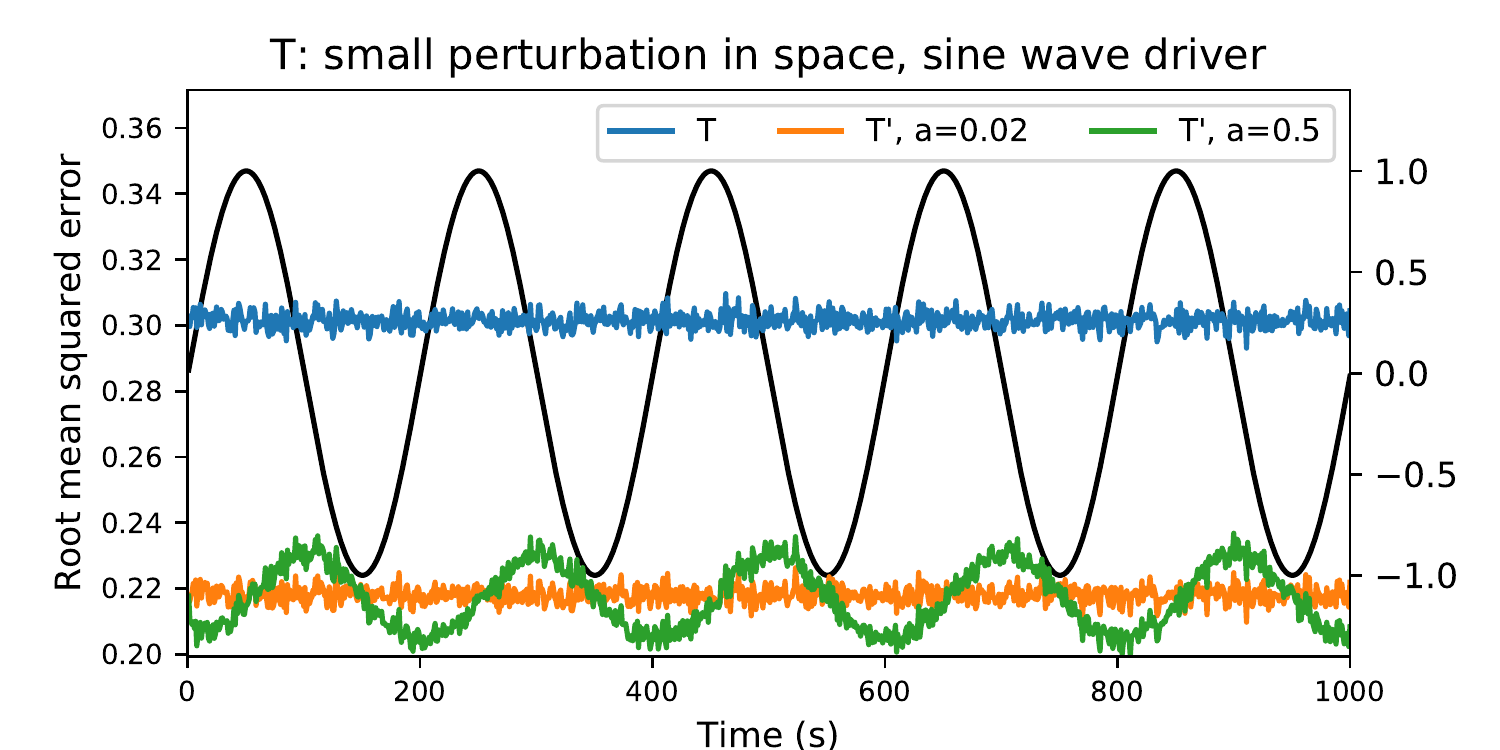}
		\includegraphics[width=.5\textwidth]{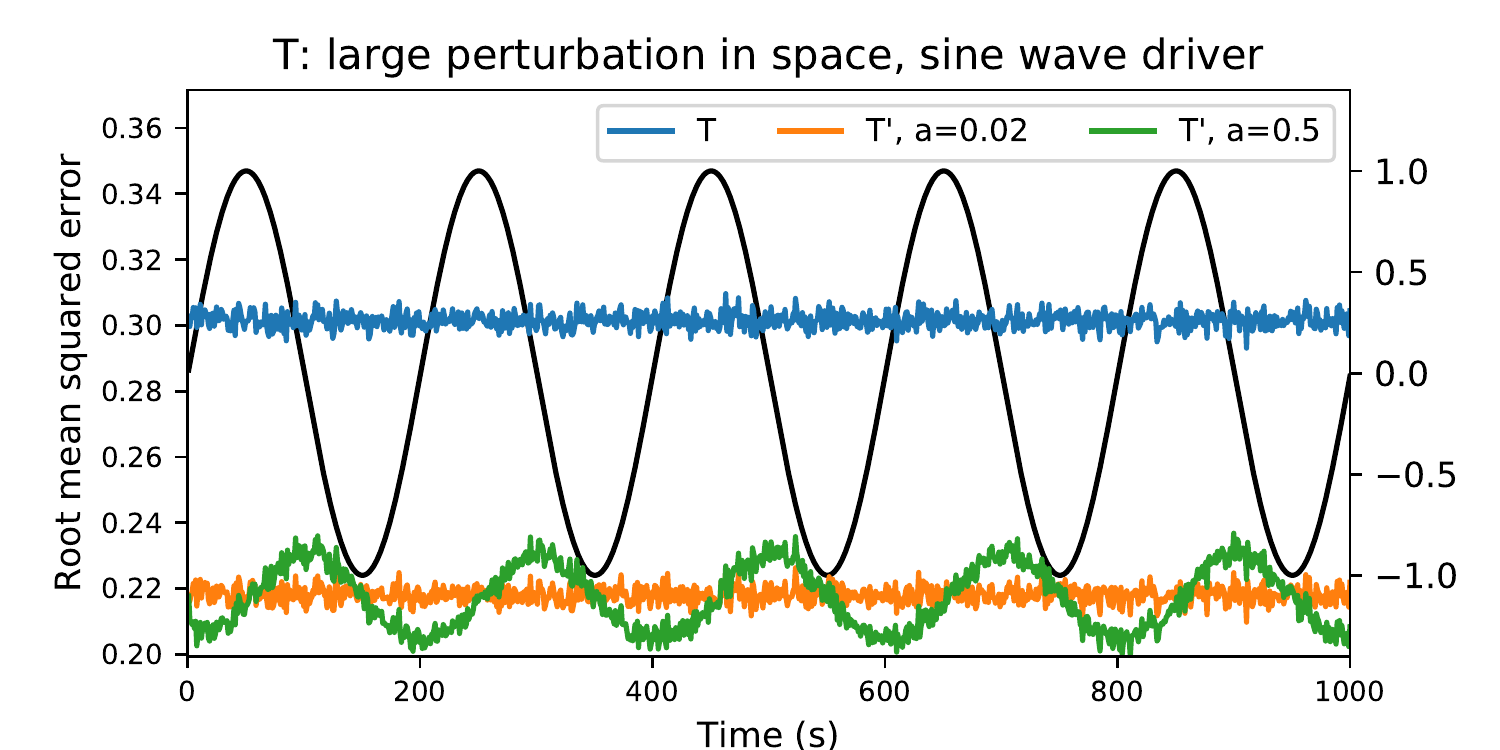}
		\includegraphics[width=.5\textwidth]{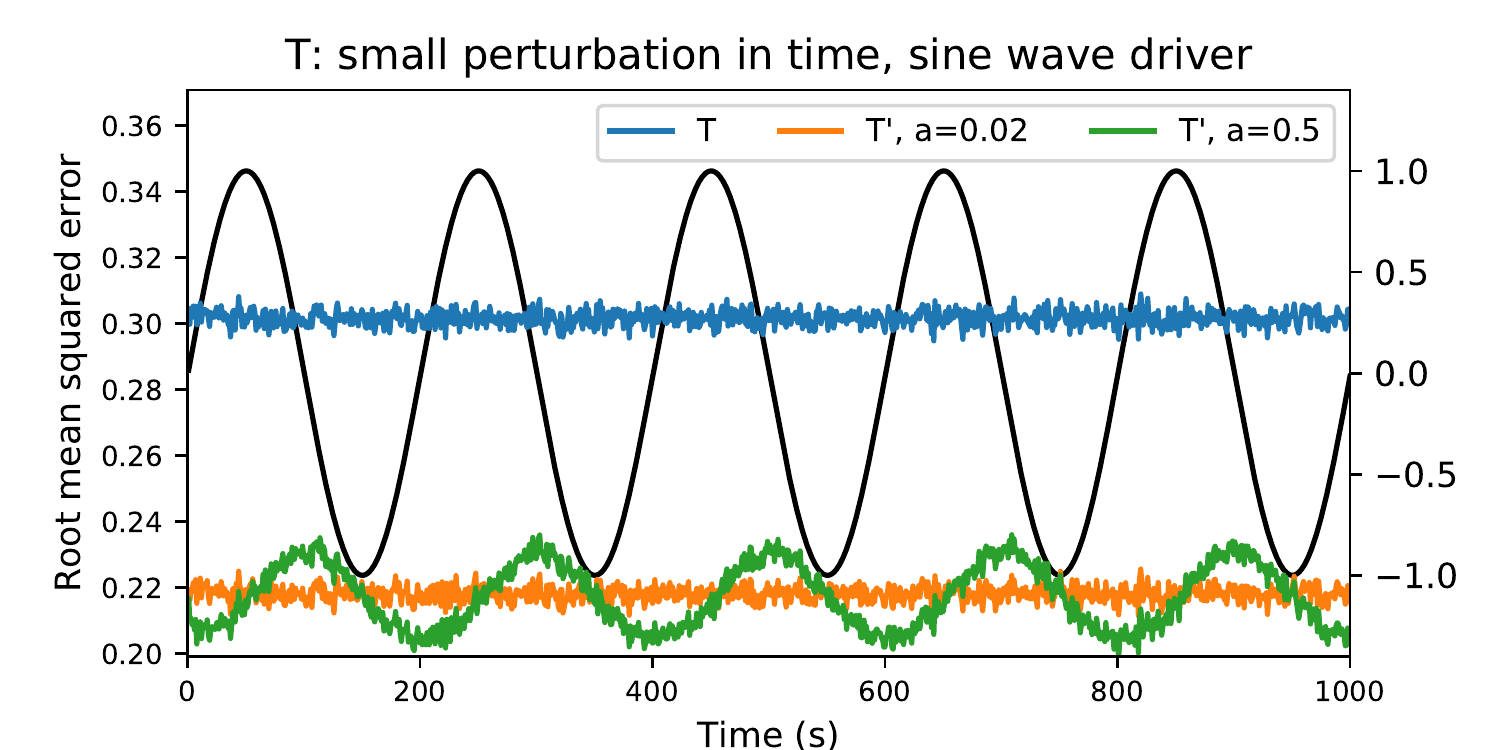}
		\includegraphics[width=.5\textwidth]{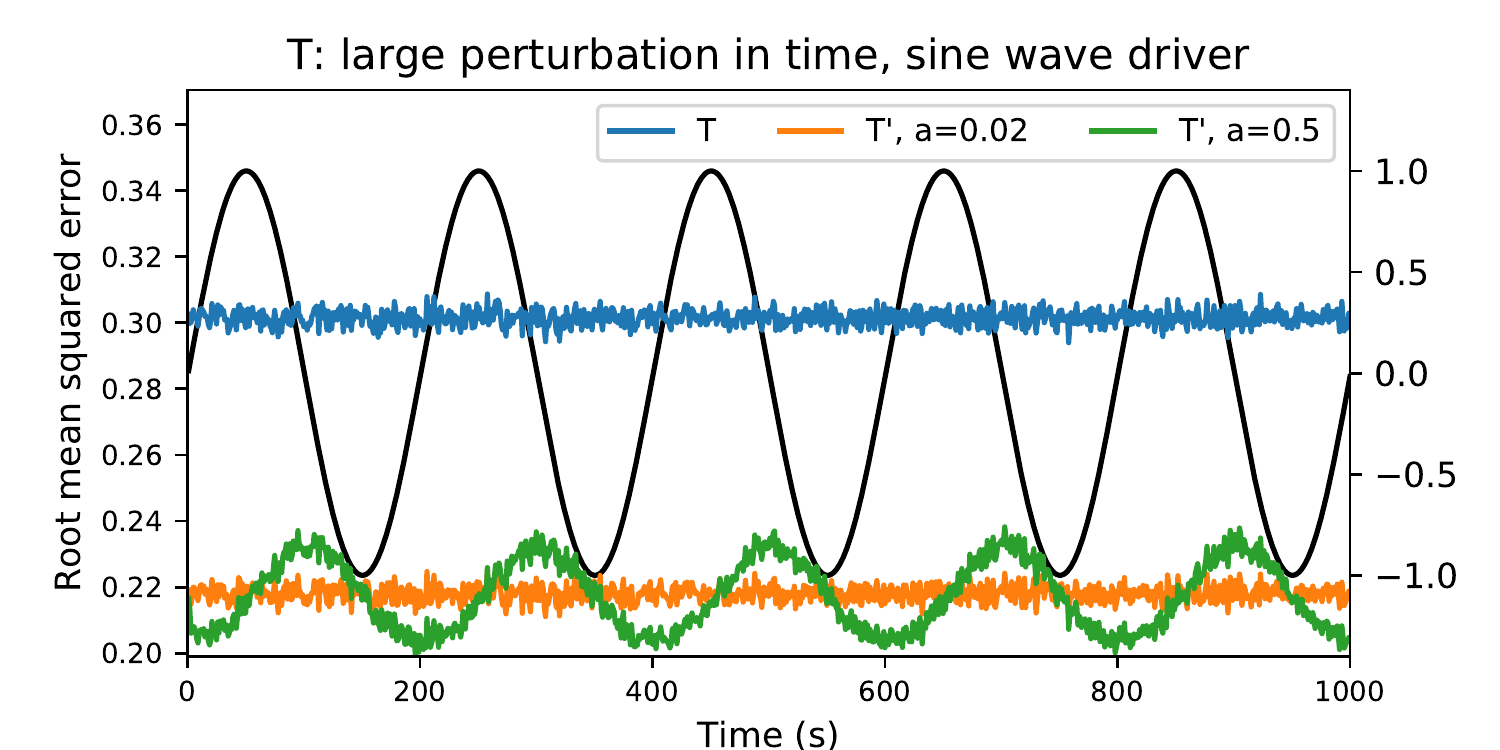}
		\includegraphics[width=.5\textwidth]{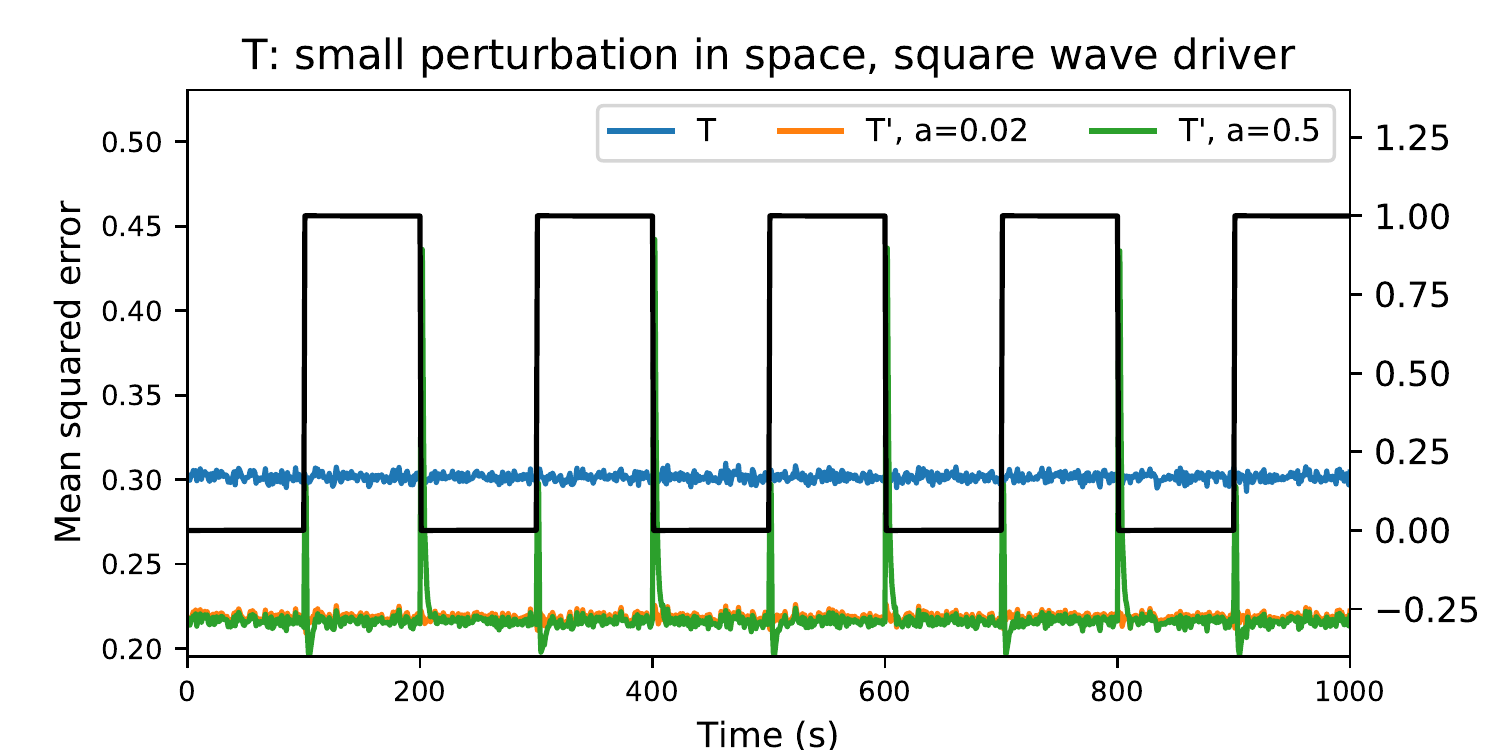}
		\includegraphics[width=.5\textwidth]{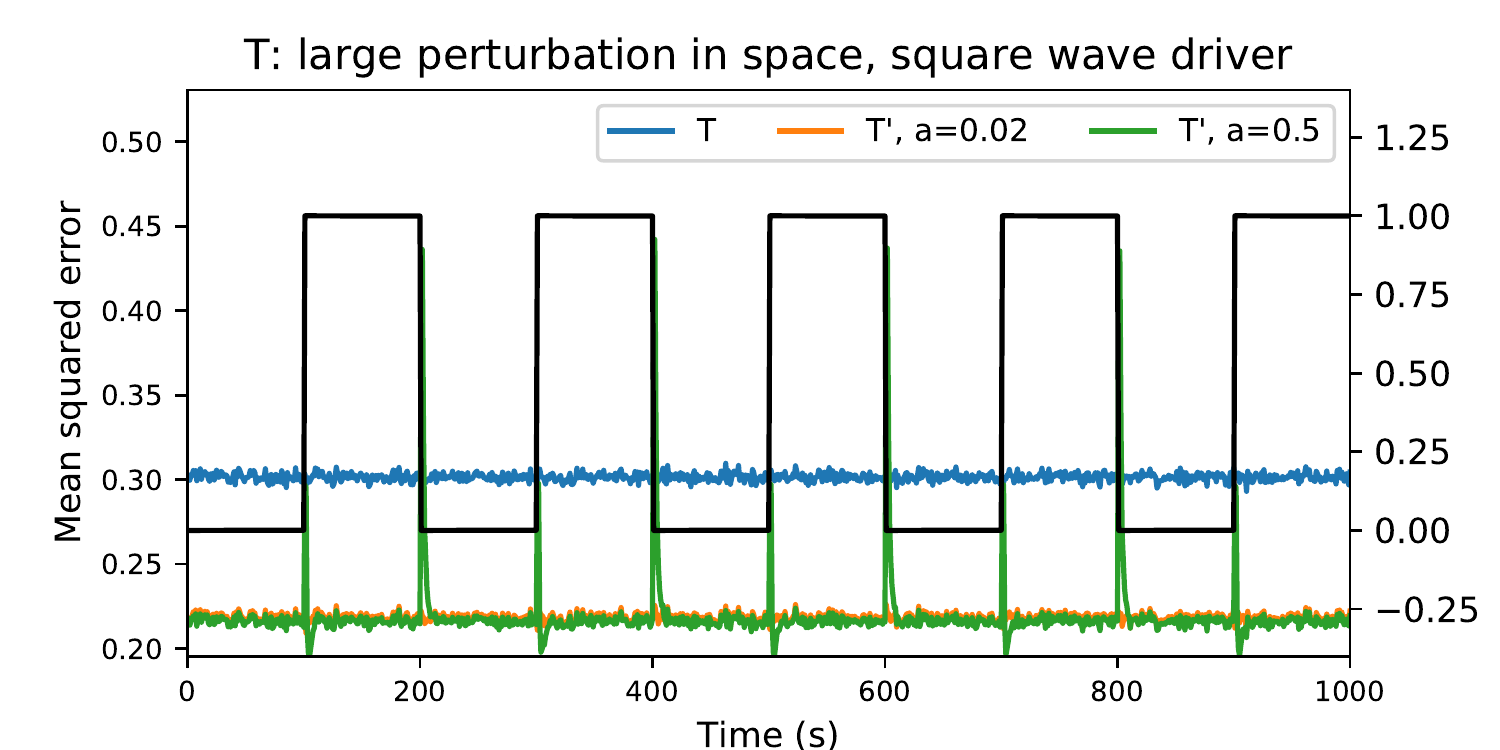}
		\includegraphics[width=.5\textwidth]{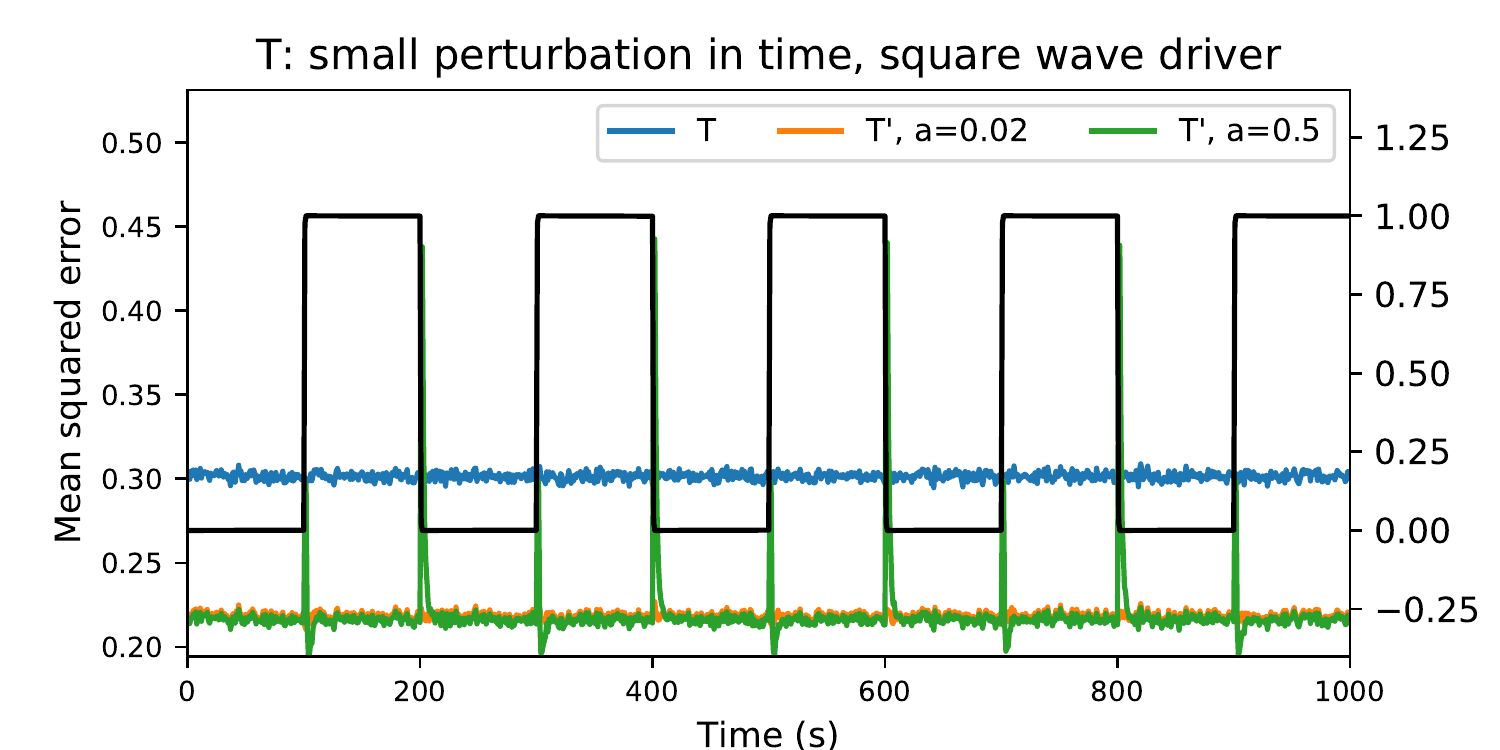}
		\includegraphics[width=.5\textwidth]{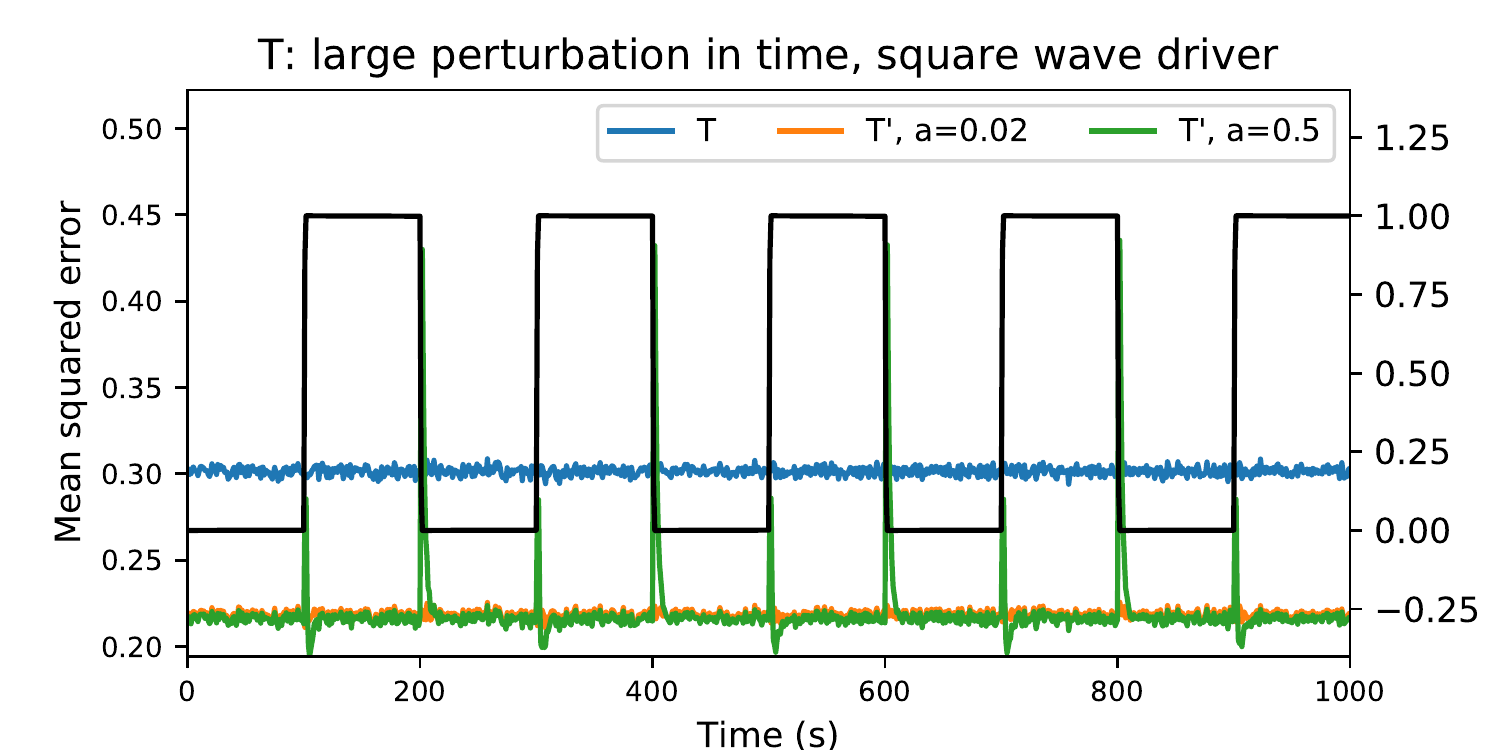}
        \caption{Evaluation of T (blue) and T' building blocks with different smoothing parameter values (\texttt{a} = 0.02 in green, and \texttt{a} = 0.5 in red).
        The driver signal (plotted in black for reference) is locally summed with a random noise in $[-1, 1]^{10}$ and fed to the algorithm for tracking.
        We measure the root mean squared error in the devices\TYtask{'} response for small (left column) or large (right column) perturbations in either space (first and third row) or time (second and last row).
        T' outperforms T in every scenario but the square wave transient: the smoothing with the neighbouring devices, in fact, greatly mitigates the local introduction of noise at the price of a lower reactivity to signal changes.
        The smoothing parameter can be interpreted as \TYtask{controlling a} trade-off between such reactivity and the smoothness of the response.
        In our testbed, T' shows minimal sensibility to any kind of perturbation.}
        \label{fig:t}
\end{figure}


\section{Application Examples}
\label{sec:casestudy}

We now illustrate, with two application examples, how distributed applications can be implemented on top of the proposed building blocks (hiding the low-level coordination mechanisms \texttt{rep} and \texttt{nbr}), and then quickly adjusted and optimised toward specific performance goals by switching the specific building block implementation that is used, using the variants presented in previous section.
Both of the scenarios that we consider are in a pervasive computing environment, and focus on a network of personal devices (e.g., phones, smart watches) spread through an urban environment.
In these scenarios, devices move with the person carrying them along the walkable areas of the city, and can only indirectly influence movement (e.g., by presenting a message to their user).

For the first scenario, we consider a community festival, with acts performing in various venues, and wish to track the number of people watching each act over time.
Here, we will consider a person to be watching an act if they are part of a continuous region of crowd that is closer to that act than to any other act.
This computation can be implemented by using G to partition the space into zones of influence, by means of a potential field of which each act is a source (as in function \texttt{distanceTo}).
We then use C to sum a field counting the number of people closely surrounded by others, and thus forming a crowd (as in function \texttt{summarise}).
Finally, T is used for smoothing both the crowd estimates and the results over time.
%
%
The resulting code, expressed using the functions described in previous section, is as follows:
\begin{Verbatim}[fontsize=\fontsize{7pt}{8pt}, frame=single, commandchars=\\\{\}, codes={\catcode`$=3\catcode`^=7\catcode`_=8}]
\km{def} crowdSize(acts, crowd) \{
  {T\_track}({C\_sum}({G\_distanceTo}(acts), {T\_track}(crowd)))
\}
\end{Verbatim}

\begin{figure}
	\fbox{\includegraphics[width=.48\textwidth]{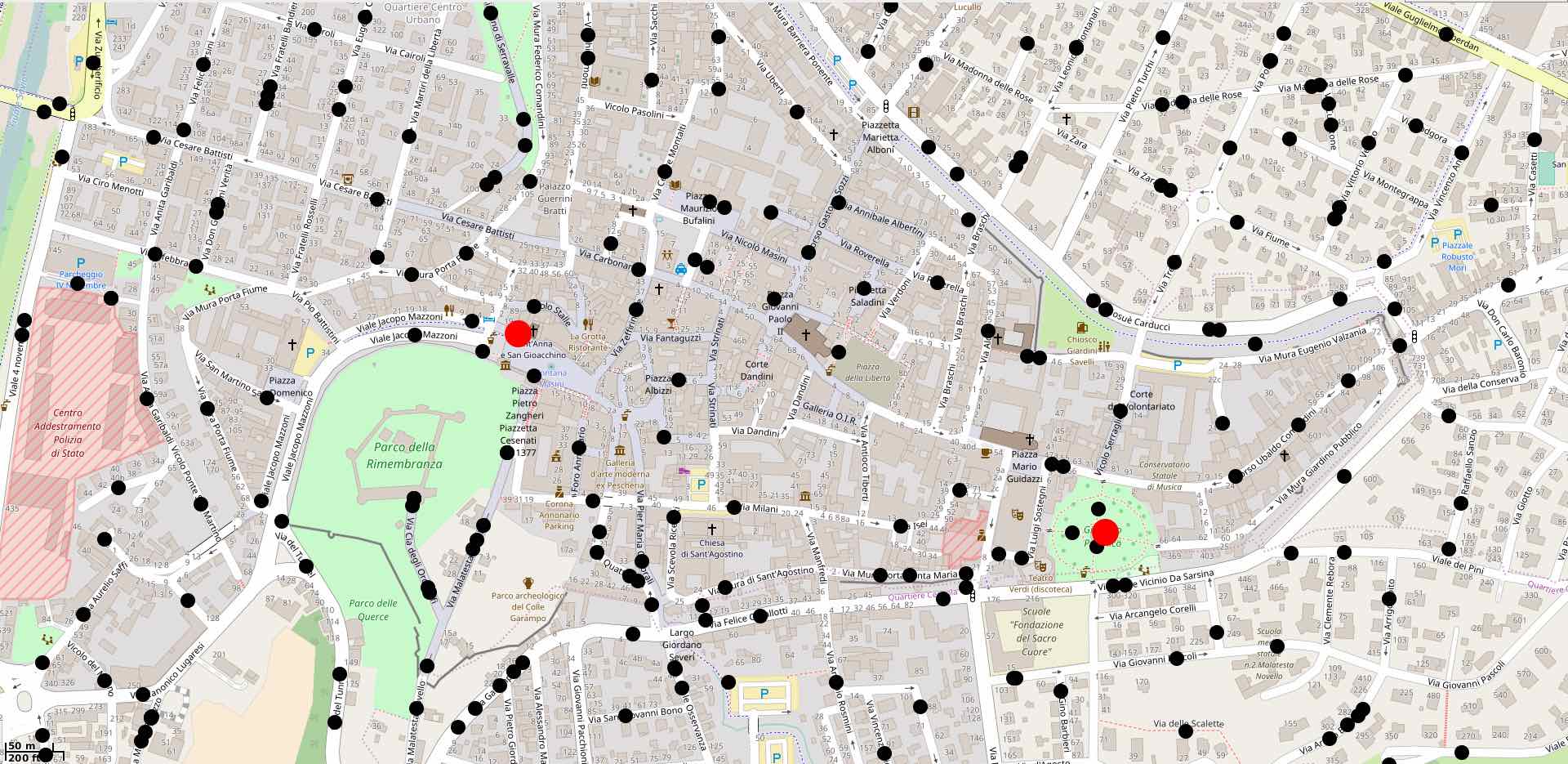}}
	\fbox{\includegraphics[width=.48\textwidth]{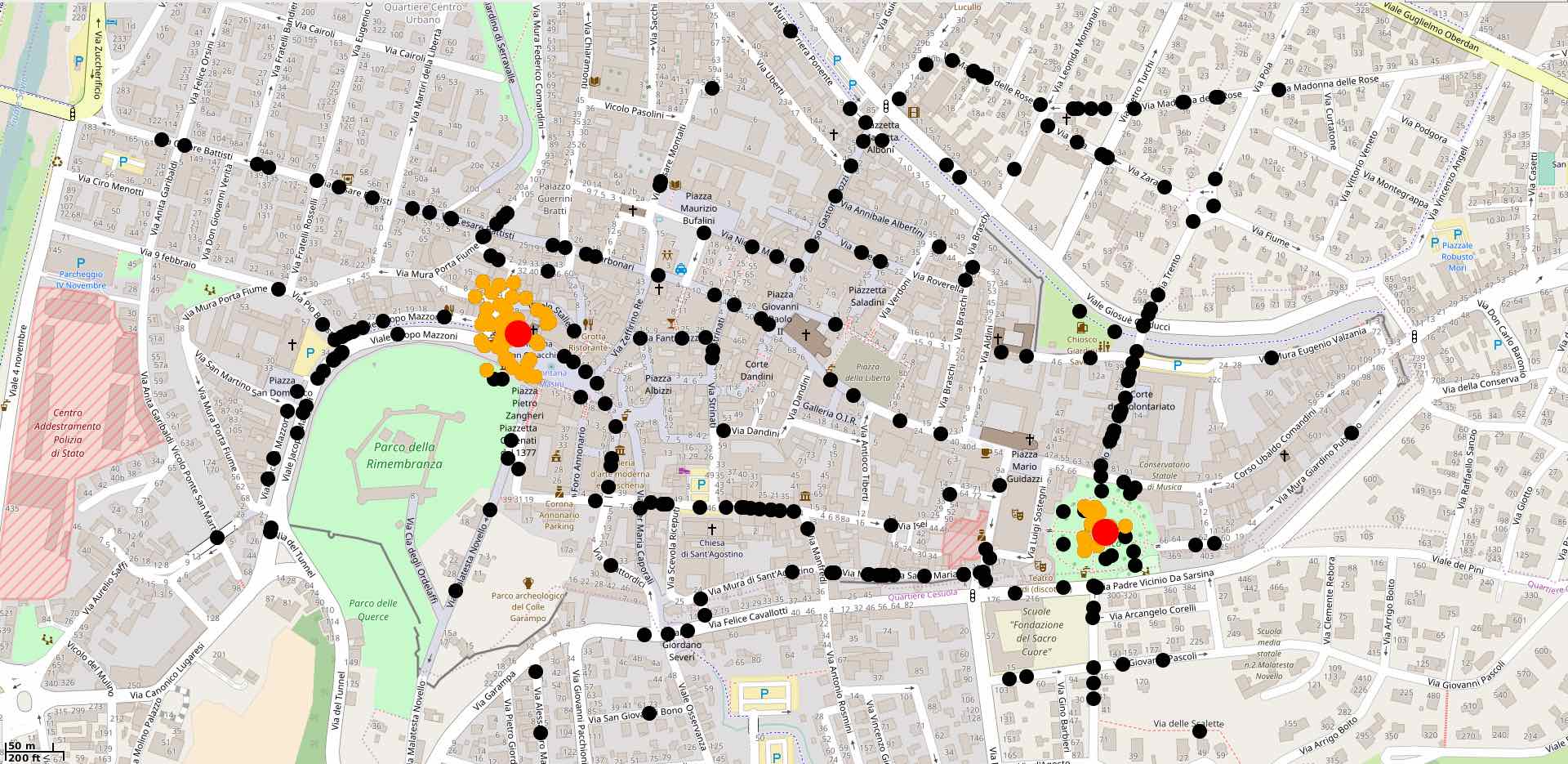}}
	\caption{Screenshots from simulation of crowd size estimation scenario: acts are indicated as red dots, pedestrians are black, and pedestrians who are part of a contiguous crowd are orange. From their initial position, people walk towards an act of interest following the pedestrian roads, becoming counted as part of a crowd once they have clumped up close to an act.}
	\label{fig:crowdshot}
\end{figure}

\todo[inline]{For revision: make graphs that show the evolving number of people over time being tracked}

To test this example application in simulation, we distributed a network of 300 devices randomly distributed across the city centre of the Italian city of Cesena.
In this simulation, pedestrians walk at 1.4 meters per second from their initial position towards an act randomly chosen between the two located in distinct large spaces of the city (Piazza del Popolo and Giardini Savelli), as depicted in \Cref{fig:crowdshot}.
Devices run asynchronously, performing a round of computation and communication every five seconds, and communicating by broadcast within a radius of 150 metres (ignoring buildings and other physical obstacles).
Our implementation is realised in Protelis~\cite{Protelis15} and simulations were performed using Alchemist~\cite{PianiniJOS2013}.
\JBtask{We note that Alchemist is a generalised GIS framework for multi-agent simulations, not a specialised crowd simulator, but higher-fidelity crowd simulations are not necessary for studying the adaptation dynamics of the information system.}

In this scenario, we execute eight variants of the {\tt crowdSize} algorithm, all combinations of the building blocks and alternates developed in the previous section: G or G' (FLEX), C or C' (multipath), and T or T'.
We measure the error for each combination as the absolute value of the difference between estimated and true counts for people watching each act:
$$\frac{1}{|A|}\sum_{a \in A}{|\hat{P_a} - P_a|}$$
where $A$ is the set of acts $a$, $|A|$ is the number of acts, $\hat{P_a}$ is the estimated count of people watching act $a$ as computed by the algorithm, and $P_a$ is the true count of people watching an act.
\begin{figure}
\subfigure[G' improves over G]{\includegraphics[width=.5\textwidth]{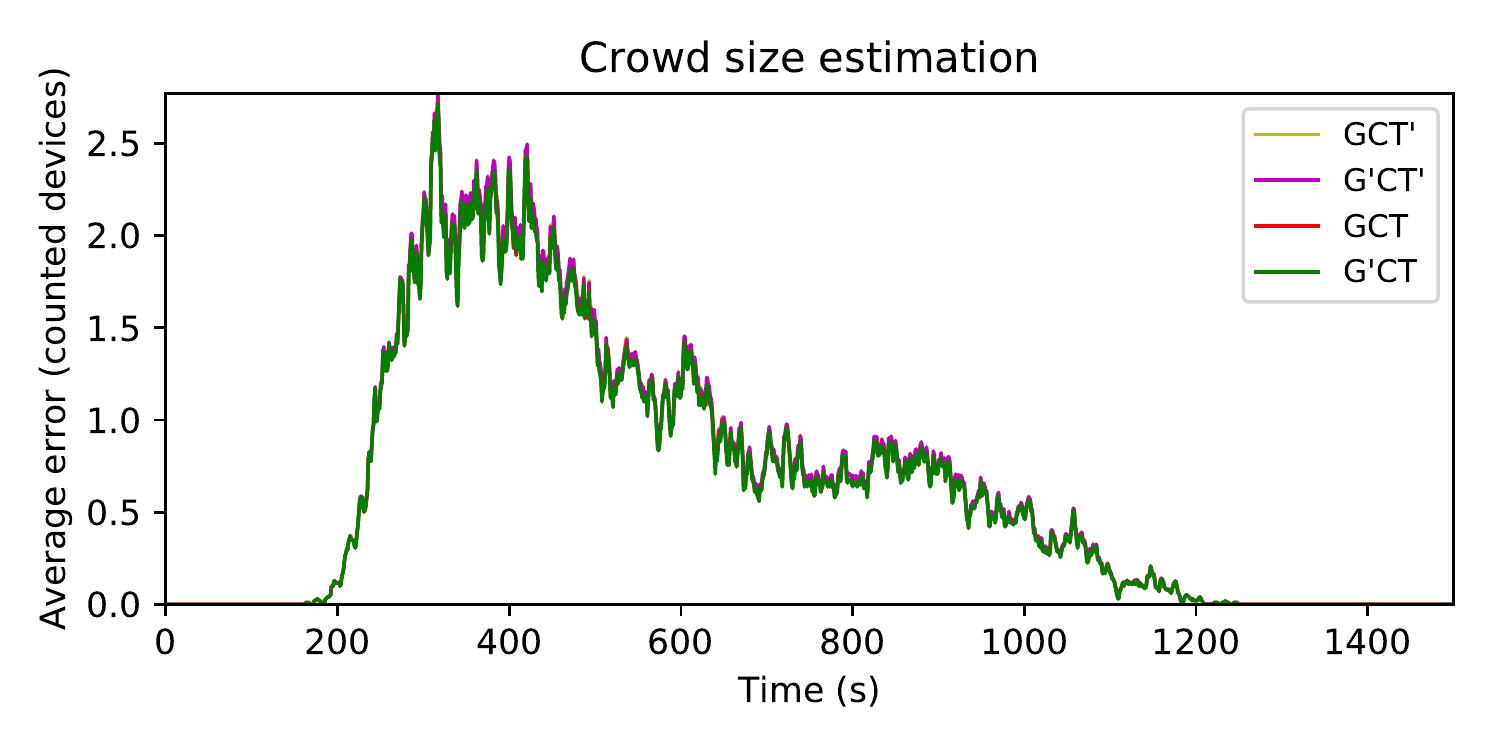}}
\subfigure[C' fails, but is mitigated by G']{\includegraphics[width=.5\textwidth]{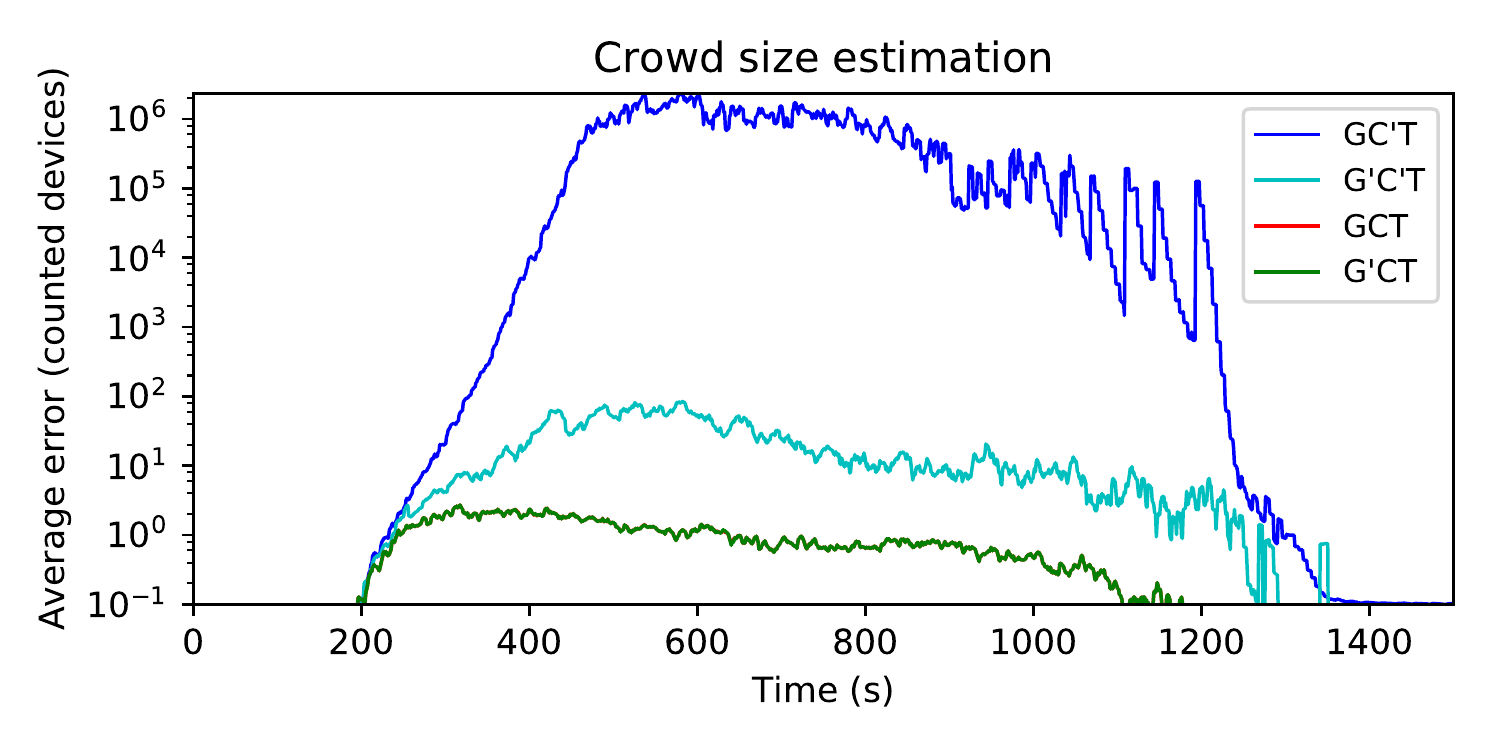}}
\caption{Key results for the crowd size estimation scenario: 
a) Use of G' slightly improves performance over G, while T performs slightly better than T'.
b) The C' algorithm fails badly due to the network being both sparse and volatile, mandating preference of C in this case. The problems with C' can be largely mitigated by substitution of G' instead of G, though the choice of T versus T' does not have any significant effect.}
\label{fig:crowddata}
\end{figure}

\Cref{fig:crowddata} presents key results, averaged over 51 simulation runs.
In these simulations, adopting G' instead of G produces a slight improvement in performance.
On the other hand, it turns out that C' fails badly, always making the results much worse, likely due to the combination of both the high volatility of the network and the sparsity induced by city streets.
This failure, however, can be mitigated by applying G', which produces a potential function that is much more stable in response to large perturbations.
The choice of T versus T' has much less impact: T' performs slightly worse than T in combination with C' and does not mitigate the failure of C'.

The second example considers signaling an evacuation alert signal to a pre-defined zone, along with the proposal of a suggested evacuation path.
This is implemented using T to track whether any device in the zone is currently alerted (using G to create a potential field to a static device selected as coordinator, and C to perform a logical or as in function \texttt{any}), then using G to broadcast that value from the coordinator throughout the zone and again to compute paths to the non-alerted areas outside of the zone.
Finally, the {\tt mux} operator is used to differentiate computations on devices inside and outside of the alert zone.
%
%
\begin{Verbatim}[fontsize=\fontsize{7pt}{8pt}, frame=single, commandchars=\\\{\}, codes={\catcode`$=3\catcode`^=7\catcode`_=8}]
\km{def} evacuationAlert(zone, coordinator, alert) \{
  {G\_distanceTo}(
    \pr{mux}(zone,
      false,
      {G\_broadcast}(coordinator, 
        {T\_track}({C\_any}({G\_distance}(coordinator), alert)))))
\}
\end{Verbatim}

\begin{figure}
	\fbox{\includegraphics[width=.48\textwidth]{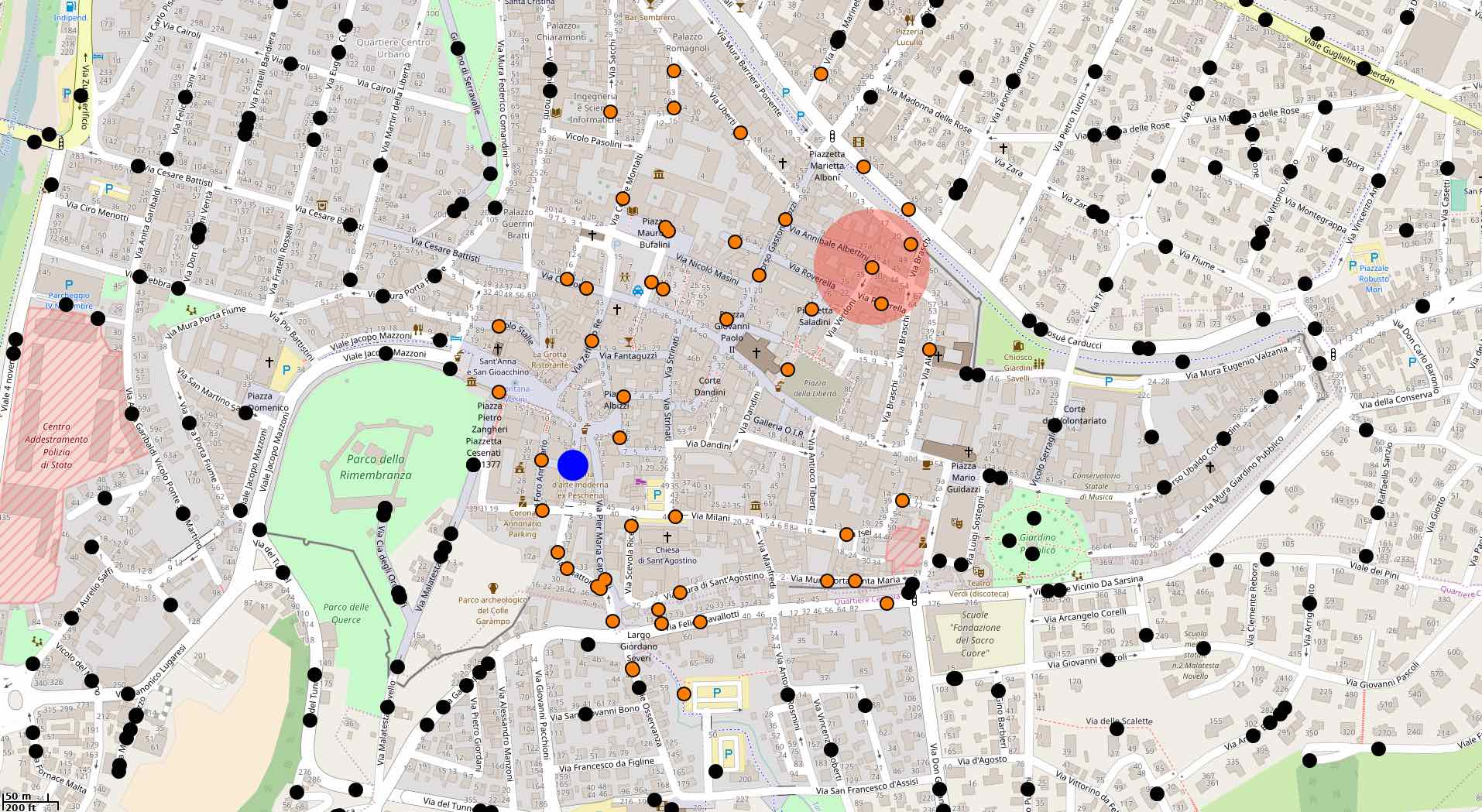}}
	\fbox{\includegraphics[width=.48\textwidth]{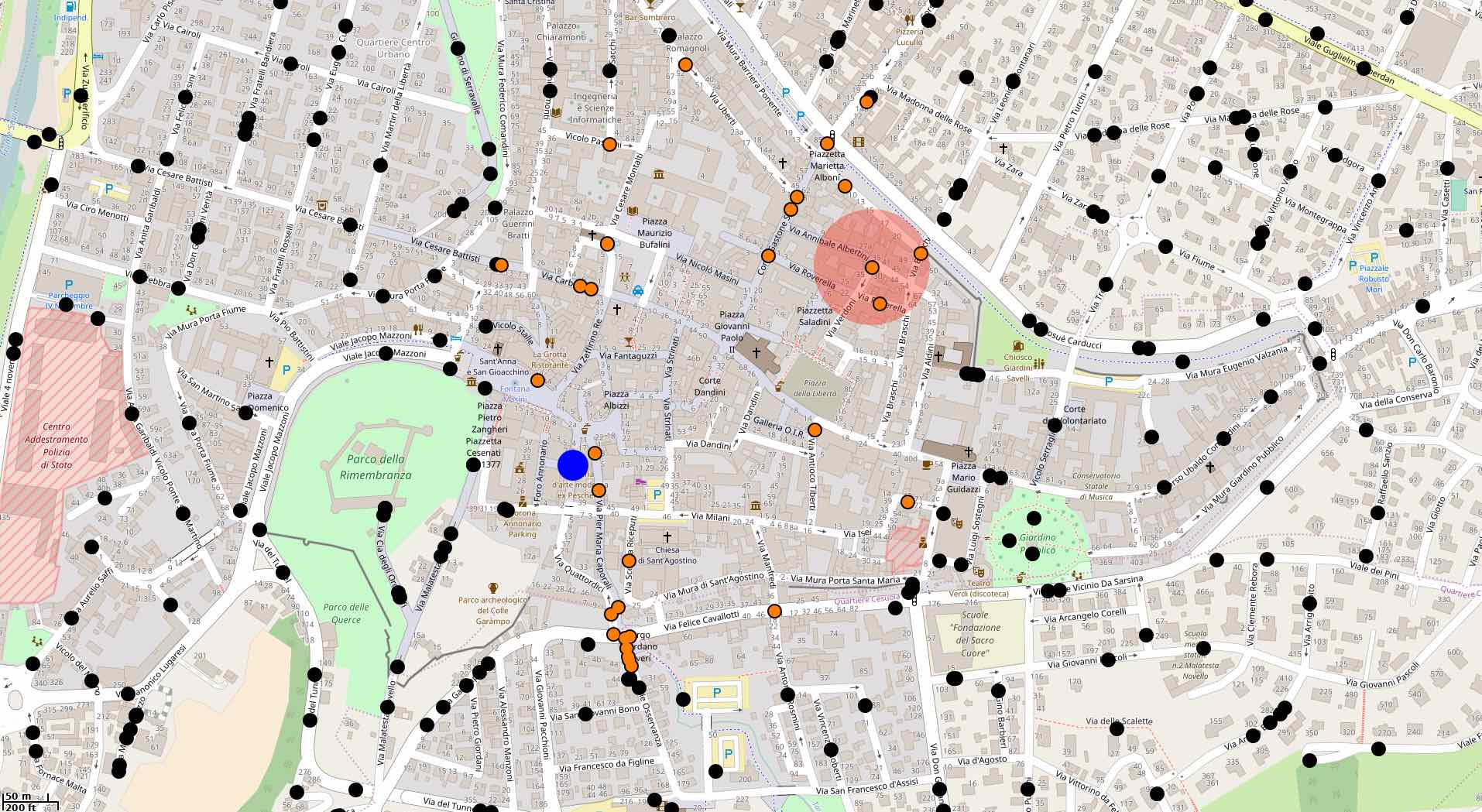}}
	\caption{Screenshots from simulation of evacuation alert scenario: devices are initially randomly scattered through the city centre (black dots). After alert (translucent red circle) is enabled, and devices in the evacuation zone are signaled (orange) by the action of the coordinator (blue) and begin trying to leave the zone.}
	\label{fig:alertshot}
\end{figure}

Simulations for this experiment used the same environment of 300 devices spread through the center of Cesena, with the same model of asynchronous execution and communication, the only difference being that devices perform a round of computation and communication every two seconds rather than every five seconds.
In this simulation, devices are initially stationary, and the alert signal is enabled starting at time $t = 20$ seconds of simulated time from the start of the simulation.
%
Since devices are unable to directly affect the movement of the people holding them, however, we simulate the people acting on the alert not by following the direction provided
by any of the simulated algorithms, but walking toward the closest waypoint outside of the evacuation zone.
Such behaviour is depicted in \Cref{fig:alertshot}.

As before, we execute eight variants, covering all combinations of the three building blocks and their alternates.
We measure the error for each algorithm as the mean of the minimum mean square error between the angles of the suggested evacuation verctor and the optimal one for each node, normalised in [0, 1], with the special rule that devices that are in alert zone when they shouldn't be or not in the alert zone when they should be get the maximum error, namely:
\[
\text{error = }\frac{1}{N}\sum_{d \in D}{
\begin{cases}
0 & \text{not in zone and not alerted}\\
(\frac{\min(|\alpha_d - \hat{\alpha_d|}, 2\pi - |\alpha_d - \hat{\alpha_d|})}{\pi})^2 & \text{in zone and alerted} \\
1 & \text{otherwise (alert/zone mis-match)}
\end{cases}}
\] where $N$ is the number of devices initially inside the zone, $D$ is the collection of all devices, $\hat{\alpha_d}$ is the computed direction (angle) for device $d$, and $\alpha_d$ is its actual ideal direction.
The minimum function is used in order to always pick the smallest angle between the two separating the optimal vector and the suggested one (namely, the difference of the two and $2\pi$ minus that value).
This outputs an error in the $[0, \pi]$ range, that we normalise linearly into $[0, 1]$.

\begin{figure}
\centering
\includegraphics[width=.5\textwidth]{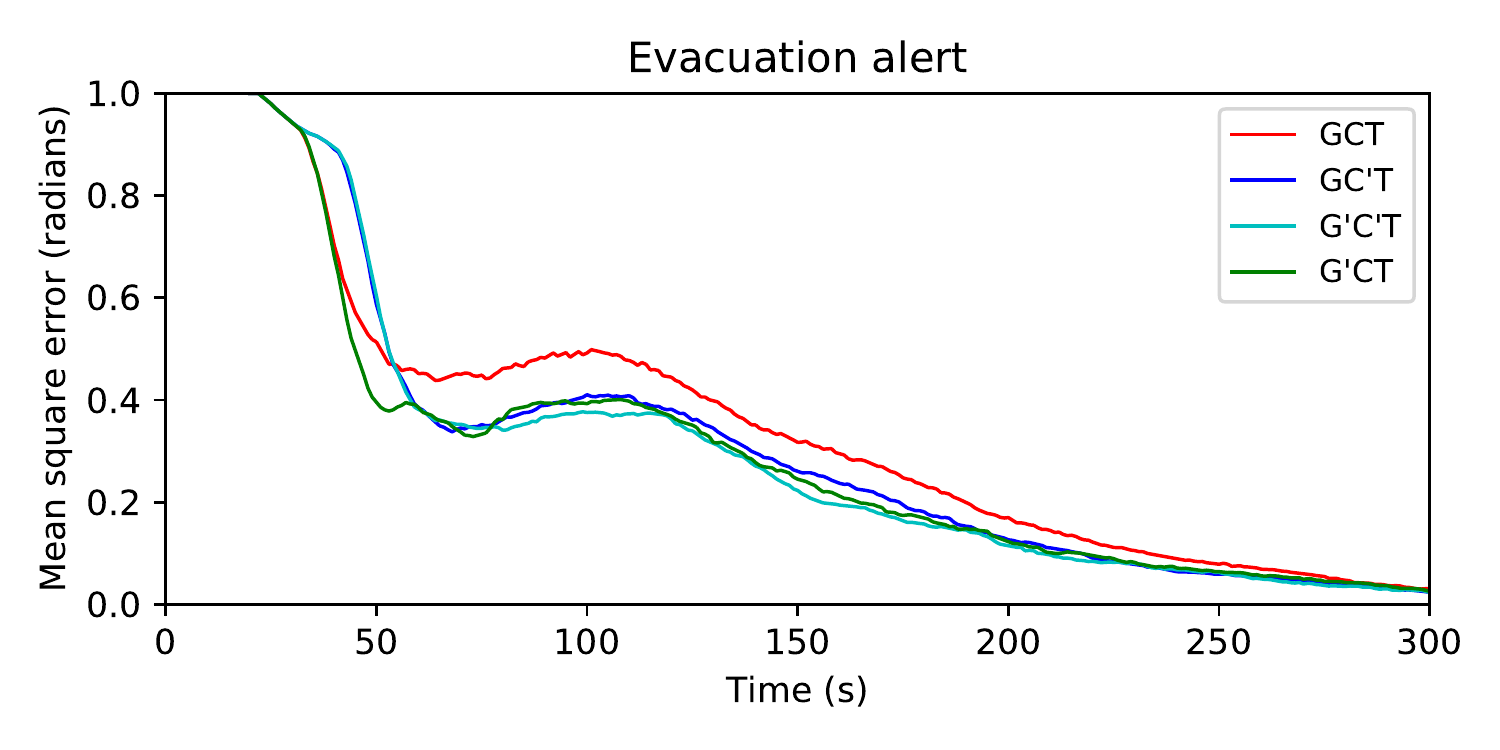}
\caption{Results for the evacuation alert scenario: G' and C' both improve performance significantly over G and C, respectively. Additional incremental improvement can be obtained by using both G' and C'. The choice of T versus T' has no significant effect on performance.}
\label{fig:alert}
\end{figure}

In this scenario, we find that two of the proposed alternative implementations of the self-stabilising building blocks significantly improve performance.
\Cref{fig:alert} shows the results, averaged over 51 simulation runs.
G', in particular, performs from equivalently to much better than G along the whole simulated time span.
The behaviour of C' is more complex: it has a longer reaction time as compared to C, as it is more sensitive to large perturbations. 
As soon as the initial transient phase is over, however, C' provides a consistent improvement over the performance of the original C implementation.
Using C' and G' together provides a further (though smaller) performance increment. 
The choice of T' versus T, however, has no significant impact on performance.

Together, these results illustrate how our approach enables fast, lightweight implementation and optimisation of distributed applications.
Different applications will be best served by different combinations and tradeoffs in the dynamics of building block implementations: for example, choosing G' over G helped in both application scenarios, while C' is useful in the second but not the first, and neither had noisy enough changes for T' to be significantly beneficial over T.
The approach that we have implemented allows such combinations to be rapidly and safely explored, enabling optimisation of distributed systems without their re-design.





\section{Conclusions}
\label{sec:conclusions}

Using the computational field calculus as ``lingua franca'' for an abstract, uniform description of self-organising computations, we have identified a large class of self-stabilising distributed
algorithms, including a set of general ``building block'' operators that conceptually simplify the specification of programs within this class.
Such a class is formalised in terms of a fragment of the field calculus, closed under composition, and flexible enough to also include various alternative \TYtask{implementations} of the building blocks, 
allowing dynamical performance to be optimised while still guaranteed to converge to the same values.
This self-stabilising fragment is at the core of a methodology for efficient engineering of self-organising systems, rooted in modelling and simulation: \emph{(i)} a system specification is constructed using formally-proved self-stabilising building blocks, and \emph{(ii)} alternative implementations of building blocks are switched in selected points of the specification to improve performance, with performance improvement detected by empirical means such as simulations.

An important future direction for improvement is to obtain a more
detailed characterisation for the dynamic trade-space, in order
to enable a more systematic approach to optimisation via mechanism
substitution.
In addition to making human engineering easier, this may also enable
automated substitution optimisation, both during the engineering
process and dynamically at run-time.
\FDtask{Furthermore, alternative definitions of \emph{self-stabilisation} could be inspected, for capturing and describing wider classes of resilient program behaviours (such as replicated gossip \cite{viroli:replicatedgossip}) or for allowing a better modeling of important aspects of spatial computations (such as space-time information),}
\JBtask{as well as integration with dynamical response models such as those presented in~\cite{DasguptaCDC16,MoECAS17}.}
Other important directions for improvement are expansion of the
library of building blocks (including to non-spatial systems), identification of more substitution
relationships between building blocks and high-performance resilient
coordination mechanisms, and development and deployment of
applications based on this approach.

\subsection*{Acknowledgements}

The authors would like to thank the anonymous reviewers for insightful comments and suggestions for improving the presentation.
This work has been partially supported by funding from the European Union's Horizon 2020 research and innovation programme under grant agreement No 644298 HyVar (Damiani), ICT COST Action IC1402 ARVI (Damiani), Ateneo/CSP project RunVar (Damiani), and the United States Air Force and the Defense Advanced Research Projects Agency (DARPA) under Contract No. FA8750- 10-C-0242 (Beal)
and Contract No. HR001117C0049 (Beal).
The views, opinions, and/or findings expressed are those of the author(s) and should not be interpreted as representing the official views or policies of the Department of Defense or the U.S. Government.
This document does not contain technology or technical data controlled under either U.S. International Traffic in Arms Regulation or U.S. Export Administration Regulations.
Approved for public release, distribution unlimited.


\appendix
\section*{Appendix}
\setcounter{section}{1}

\subsection{Proof of self-stabilisation for the fragment} \label{sec:proofs}

In this appendix we report complete proofs for the statements given in Section \ref{sec:fragment_is_ss}. We first prove self-stabilisation for the minimising $\repK$ pattern (Lemma \ref{lem:minimising_termination}), since it is technically more involved than the proof of self-stabilisation for the remainder of the fragment. We then prove self-stabilisation through a variation of the goal results (Lemma \ref{lem:stabilisation}) more suited for inductive reasoning. Theorems \ref{thm:stabilisation} and \ref{thm:substitutability} will then follow by inspecting the proof of those lemmas.

Let $\s_\text{min} = \repK(\e)\{ (\xname) \toSymK{} \funvalue^\mathsf{R}(\minHoodLoc(\funvalue^\mathsf{MP}(\nbrK\{\xname\}, \overline\s), \s), \xname, \overline\e) \}$ be a minimising $\repK$ expression such that $\builtindenot{}{\overline\s} = \overline\dvalue$, $\builtindenot{}{\s} = \dvalue$. Let $P = \overline\deviceId$ be a path in the network (a sequence of pairwise connected devices), and define its \emph{weight} in $\s_\text{min}$ as the result of picking the eventual value $\lvalue_1 = \dvalue(\deviceId_1)$ of $\s$ in the first device $\deviceId_1$, and repeatedly passing it to subsequent devices through the monotonic progressive function, so that $\lvalue_{i+1} = \funvalue^\mathsf{MP}(\lvalue_i, \overline\anyvalue)$ where $\overline\anyvalue$ is the result of projecting fields in $\overline\dvalue(\deviceId_{i+1})$ to their $\deviceId_i$ component (leaving local values untouched). Notice that the weight is well-defined since function $\funvalue^\mathsf{MP}$ is required to be stateless.
		
\begin{lem} \label{lem:minimising_termination}
	Let $\s$ be a minimising $\repK$ expression. Then $\s$ self-stabilises in each device $\deviceId$ to the minimal weight in $\s$ for a path $P$ ending in $\deviceId$.
\end{lem}
\begin{proof}
	Let $\lvalue_\deviceId$ be the minimal weight for a path $P$ ending in $\deviceId$, and let $\deviceId^0, \deviceId^1, \ldots$ be the list of all devices $\deviceId$ ordered by increasing $\lvalue_\deviceId$. Notice that the path $P$ of minimal weight $\lvalue_{\deviceId^i}$ for device $i$ can only pass through nodes such that $\lvalue_{\deviceId^j} \leq \lvalue_{\deviceId^i}$ (thus s.t. $j < i$). In fact, whenever a path $P$ contains a node $j$ the weight of its prefix until $j$ is at least $\lvalue_{\deviceId^j}$; thus any longer prefix has weight strictly greater than $\lvalue_{\deviceId^j}$ since $\funvalue^\mathsf{MP}$ is progressive.
	
	Let $\Cfg_0 \nettran{}{\deviceId_0}{} \Cfg_1 \nettran{}{\deviceId_1}{} \ldots$ be a fair evolution\footnote{Notice that $\deviceId_0$ is the first device firing while $\deviceId^0$ is the device with minimal weight.} and assume w.l.o.g. that all subexpressions of $\s$ not involving $\xname$ have already self-stabilised to computational fields $\overline\dvalue$, $\dvalue$ (as in the definition of weight) in the initial state $\Cfg_0$. We now prove by complete induction on $i$ that device $\deviceId^i$ stabilises to $\lvalue_{\deviceId^i}$ after a certain step $t_i$.
	
	Assume that devices $\deviceId^j$ with $j < i$ are all self-stabilised (from a certain step $t_{i-1}$), and consider the evaluation of expression $\s$ in a device $\deviceId^k$ with $k \geq i$. Since the local argument $\lvalue$ of \texttt{minHoodLoc} is also the weight of the single-node path $P = \deviceId^k$, it has to be at least $\lvalue \geq \lvalue_{\deviceId^k} \geq \lvalue_{\deviceId^i}$. Similarly, the restriction $\fvalue'$ of the field argument $\fvalue$ of \texttt{minHoodLoc} to devices $\deviceId^j$ with $j < i$ has to be at least $\fvalue' \geq \lvalue_{\deviceId^k} \geq \lvalue_{\deviceId^i}$ since it corresponds to weights of (not necessarily minimal) paths $P$ ending in $\deviceId^k$ (obtained by extending a minimal path for a device $\deviceId^j$ with $j < i$ with the additional node $\deviceId^k$). Finally, the complementary restriction $\fvalue''$ of $\fvalue$ to devices $\deviceId^j$ with $j \geq i$ is strictly greater than the minimum value for $\xname$ among those devices, since $\funvalue^\mathsf{MP}$ is progressive.
	
	It follows that as long as the minimum value for $\xname$ among non-stable devices is lower than $\lvalue_{\deviceId^i}$, the result of the \texttt{minHoodLoc} subexpression is strictly greater than \TYtask{this} minimum value. Since the overall value of $\s$ is obtained by combining the output of \texttt{minHoodLoc} with the previous value for $\xname$ through the rising function $\funvalue^\mathsf{R}$ (and a rising function does not drop below the minimum of its arguments), the minimum value for $\xname$ among non-stable devices cannot decrease as long as it is lower than $\lvalue_{\deviceId^i}$, and it cannot drop below $\lvalue_{\deviceId^i}$ if it is already greater than that.
	
	Furthermore, \TYtask{the} minimum has to eventually increase until it reaches at least $\lvalue_{\deviceId^i}$. Recall that a rising function selects its first argument infinitely often (since the order $\vartriangleleft$ is noetherian). Thus each device realising a minimum for $\xname$ among non-stable devices has to eventually evaluate $\s$ to the output of the \texttt{minHoodLoc} subexpression, which is strictly higher than the previous minimum, and it will not be able to reach \TYtask{the} previous minimum afterwards.
	
	Let $t' \geq t_{i-1}$ be the first step in which the minimum for $\xname$ among non-stable devices is at least $\lvalue_{\deviceId^i}$, and consider device $\deviceId^i$. Let $P$ be a path of minimum weight for $\deviceId^i$, then either:
	\begin{itemize}
		\item $P = \deviceId^i$, so that $\lvalue_{\deviceId^i}$ is exactly the local argument of the \texttt{minHoodLoc} operator, hence also the output of it (since the field argument is greater than $\lvalue_{\deviceId^i}$).
		\item $P = Q, \deviceId^i$ where $Q$ ends in $\deviceId^j$ with $j < i$. Since $\funvalue^\mathsf{MP}$ is monotonic non-decreasing, the weight of $Q', \deviceId^i$ (where $Q'$ is minimal for $\deviceId^j$) is not greater than that of $P$; in other words, $P' = Q', \deviceId^i$ is also a path of minimum weight. It follows that $\fvalue(\deviceId^j)$ (where $\fvalue$ is the field argument of the \texttt{minHoodLoc} operator) is exactly $\lvalue_{\deviceId^i}$.
	\end{itemize}
	In both cases, the output of \texttt{minHoodLoc} in $\deviceId^i$ stabilises to $\lvalue_{\deviceId^i}$ from $t'$ on. Let $t_i$ be the first step after $t'$ in which the rising function $\funvalue^\mathsf{R}$ selects its first argument $\lvalue_{\deviceId^i}$. Then expression $\s$ in device $\deviceId^i$ is self-stabilised to $\lvalue_{\deviceId^i}$ from $t_i$ on, concluding the inductive step and the proof.
\end{proof}

Let $\dvalue$ be a computational field as defined in Section \ref{ssec:ss_eventual}. We write $\applySubstitution{\s}{\substitution{\xname}{\dvalue}}$ to indicate an aggregate process in which each device is computing a possibly different substitution $\applySubstitution{\s}{\substitution{\xname}{\dvalue(\deviceId)}}$ of the same expression.

\begin{lem} \label{lem:stabilisation}
	Assume that every built-in operator is self-stabilising. Let $\s$ be an expression with free variables $\overline\xname$ in the self-stabilising fragment, and $\overline\dvalue$ be a sequence of computational fields of the same length. Then $\applySubstitution{\s}{\substitution{\overline\xname}{\overline\dvalue}}$ is self-stabilising.
\end{lem}
\begin{proof}
	The proof proceeds by induction on the syntax of expressions and programs. Let $\s$ be an expression in the fragment, then it can be:
	\begin{itemize}
		\item A variable $\xname_i$, so that $\applySubstitution{\s}{\substitution{\overline\xname}{\overline\dvalue}} = \dvalue_i$ is already self-stabilised.
		
		\item A value $\anyvalue$, so that $\applySubstitution{\s}{\substitution{\overline\xname}{\overline\dvalue}} = \anyvalue$ is already self-stabilised.
		
		\item A $\letK$-expression $\letK \; \xname = \s_1 \; \inK \; \s_2$. Fix an environment $\Envi$, in which expression $\s_1$ self-stabilise\TYtask{s} to $\dvalue$ after fire $t$. After $t$,  $\letK \; \xname = \s_1 \; \inK \; \s_2$ evaluates to the same value of the expression $\applySubstitution{\s_2}{\substitution{\xname}{\dvalue}}$ which is self-stabilising by inductive hypothesis.
		
		\item A functional application $\funvalue(\overline\s)$. Fix an environment $\Envi$, in which all expressions $\overline\s$ self-stabilise to $\overline\dvalue$ after fire $t$. After $t$, if $\funvalue$ is a built-in  function then $\funvalue(\overline\s)$ is already self-stabilised. Otherwise, if $\funvalue$ is a user-defined function then $\funvalue(\overline\s)$ evaluates to the same value of the expression $\applySubstitution{\body{\funvalue}}{\substitution{\args{\funvalue}}{\overline\dvalue}}$ which is self-stabilising by inductive hypothesis.

		\item A conditional $\s = \ifK (\s_1) \{\s_2\} \{\s_3\}$. Fix an environment $\Envi$, in which expression $\s_1$ self-stabilise\TYtask{s} to $\dvalue_\textit{guard}$. Let $\Envi_\truevalue$ be the sub-environment consisting of devices $\deviceId$ such that $\dvalue_\textit{guard}(\deviceId) = \truevalue$, and analogously $\Envi_\falsevalue$. Assume that $\s_2$ self-stabilises to $\dvalue_\truevalue$ in $\Envi_\truevalue$ and $\s_3$ to $\dvalue_\falsevalue$ in $\Envi_\falsevalue$. Since 
%
%
 a conditional is computed in isolation in the above defined sub-environments, $\s$ self-stabilises to $\dvalue  = \dvalue_\truevalue \cup \dvalue_\falsevalue$.
		
		\item A neighbourhood field construction $\nbrK\{\s\}$. Fix an environment $\Envi$, in which expression $\s$ self-stabilise\TYtask{s} to $\dvalue$ after fire $t$. Then $\nbrK\{\s\}$ self-stabilises to the corresponding $\dvalue'$ after one more firing of each device, where $\dvalue'(\deviceId)$ is $\dvalue$ restricted to $\Topo(\deviceId)$.
		
		\item A converging $\repK$: $\s = \repK(\e)\{ (\xname) \toSymK{} \funvalue^\mathsf{C}(\nbrK\{\xname\}, \nbrK\{\s\}, \overline\e) \}$. Fix an environment $\Envi$ and a fair evolution of the network $\Cfg_0 \nettran{}{\deviceId_0}{} \Cfg_1 \nettran{}{\deviceId_1}{} \ldots$, and let $t$ be such that all subexpressions of $\s$ not containing $\xname$ have self-stabilised after $t$. Assume that $\s$ self-stabilises to $\dvalue$; we prove that $\s$ stabilises as well to the same $\dvalue$.
		
		Given any index $i \geq t$, let $d^i$ be the maximum distance $\xname - \dvalue(\deviceId^i)$ of $\xname$ from $\s$ realised by a device $\deviceId^i$ in the network. Let $t_0 = t$ and $t_{i+1}$ be the first firing of device $\deviceId^{t_i}$ after $t_i$. Since $\deviceId^{t_i}$ realises the maximum distance in the whole network $\Cfg_{t_i}$, no device firing between $t_i$ and $t_{i+1}$ can assume a value more distant than $d^{t_i}$ without violating the converging property. Thus $d^i$, $\deviceId^i$ remains the same in the whole interval from $t_i$ to $t_{i+1}$ (excluded).
		
		Finally, in fire $t_{i+1}$ device $\deviceId^{t_i}$ recomputes its value, necessarily obtaining a closer value to $\dvalue(\deviceId^{t_i})$ (by the converging property) thus forcing the overall maximal distance in the network to reduce: $d^{t_{i+1}} < d^{t_i}$. Since the set of possible values is finite, so are the possible distances and eventually the maximal distance $d^i$ will reach $0$.
		
		\item An acyclic $\repK$: $\s = \repK(\e)\{ (\xname) \toSymK{} \funvalue(\muxK(\nbrlt(\s_p), \nbrK\{\xname\}, \s), \overline\s) \}$. Fix an environment $\Envi$ and a fair evolution of the network $\Cfg_0 \nettran{}{\deviceId_0}{} \Cfg_1 \nettran{}{\deviceId_1}{} \ldots$, and let $t$ be such that all subexpressions of $\s$ not containing $\xname$ have self-stabilised after $t$.
		
		Let $t_0 \geq t$ be any fire of the device $\deviceId^0$ of minimal potential $\s_p$ in the network after $t$. Since $\deviceId^0$ is minimal, $\muxK(\nbrlt(\s_p), \nbrK\{\xname\}, \s)$ reduces to $\s$ and the whole $\s$ to $\funvalue(\s, \overline\s)$, which is self-stabilising (after some $t'_0 \geq t_0$) for inductive hypothesis.
		
		Let $t_1 \geq t'_0$ be any fire of the device $\deviceId^1$ of second minimal potential after $t'_0$. Then $\muxK(\nbrlt(\s_p), \nbrK\{\xname\}, \s)$ in $\deviceId^1$ only (possibly) depends on the value of the device of minimal potential, which is already self-stabilised. Thus by inductive hypothesis $\s$ self-stabilises also in $\deviceId^1$ after some index $t'_1 \geq t_1$. By repeating the same reasoning on all devices in order of increasing potential, we obtain a final $t'_n$ after which all devices have self-stabilised.
		
		\item A minimising $\repK$: this case is proved for closed expressions in Lemma \ref{lem:minimising_termination}, and its generalisation to open expressions is straightforward.
	\end{itemize}
\end{proof}

\begin{thm}[Restatement of Theorem \ref{thm:stabilisation} (Fragment Stabilisation)]
	Let $\s$ be a closed expression in the self-stabilising fragment, and assume that every built-in operator is self-stabilising. Then $\s$ is self-stabilising.
\end{thm}
\begin{proof}
	Follows directly from Lemma \ref{lem:stabilisation} when $\s$ has no free variables.
\end{proof}

\begin{thm}[Restatement of Theorem \ref{thm:substitutability} (Substitutability)]
	The following three equivalences hold:
	\begin{itemize}
		\item Each $\repK$ in a self-stabilising fragment 
    self-stabilises to the same value
under arbitrary substitution of the initial condition.
		\item The \emph{converging $\repK$} pattern   
self-stabilises to the same value of
  the single expression $\s$ occurring in it.
		\item The \emph{minimising $\repK$} pattern 
  self-stabilises to the same value of
the analogous pattern where $\funvalue^\mathsf{R}$ is the identity on its first argument.
	\end{itemize}
\end{thm}
\begin{proof}
	Follows from inspecting the proof of Lemmas \ref{lem:minimising_termination} and \ref{lem:stabilisation}.
\end{proof}


\bibliographystyle{plain}
\bibliography{long}

%
%
%
%

\end{document}